\documentclass[10pt,journal,compsoc]{IEEEtran}
\usepackage{algorithm,amsbsy,amsmath,amssymb,epsfig,bbm,mathrsfs,multirow,amsthm,xcolor}
\usepackage{epstopdf}
\usepackage{graphicx,caption,subcaption}
\usepackage{setspace}
\usepackage{algorithmic}
\usepackage{subcaption,float}
\usepackage[hidelinks]{hyperref}
\newtheorem{proposition}{Proposition} 

\newtheorem{definition}{Definition}
\newtheorem{lemma}{Lemma}
\newtheorem{theorem}{\textbf{\textsc{Theorem}}}

\ifCLASSOPTIONcompsoc
\usepackage[nocompress]{cite}
\else
\usepackage{cite}
\fi
\ifCLASSINFOpdf
\else
\fi

\usepackage{setspace}

\begin{document}
%
\title{Dynamic Federated Learning-Based Economic Framework for Internet-of-Vehicles}
\author{Yuris Mulya Saputra, Dinh Thai Hoang, Diep N. Nguyen, Le-Nam Tran, Shimin Gong, and Eryk Dutkiewicz
	\IEEEcompsocitemizethanks{\IEEEcompsocthanksitem Y.~M.~Saputra is with the
		School of Electrical and Data Engineering, University of Technology Sydney, Sydney, NSW 2007, Australia (e-mail: yurismulya.saputra@student.uts.edu.au) and Department of Electrical Engineering and Informatics, Vocational College, Universitas Gadjah Mada, Yogyakarta 55281, Indonesia.
		\IEEEcompsocthanksitem D.~T.~Hoang, D.~N.~Nguyen, and E.~Dutkiewicz are with the
		School of Electrical and Data Engineering, University of Technology Sydney, Sydney, NSW 2007, Australia (e-mail: \{hoang.dinh, diep.nguyen, eryk.dutkiewicz\}@uts.edu.au).
		\IEEEcompsocthanksitem L.~N.~Tran is with the School of Electrical and Electronic Engineering, University College Dublin, Belfield, Dublin 4, Ireland (e-mail: nam.tran@ucd.ie)
		\IEEEcompsocthanksitem S. Gong is with the Sun Yat-sen University, China (e-mail: gongshm5@mail.sysu.edu.cn).
	}
	\thanks{}}

\IEEEtitleabstractindextext{%
\begin{abstract}

Federated learning (FL) can empower Internet-of-Vehicles (IoV) networks by leveraging smart vehicles (SVs) to participate in the learning process with minimum data exchanges and privacy disclosure. The collected data and learned knowledge can help the vehicular service provider (VSP) improve the global model accuracy, e.g., for road safety as well as better profits for both VSP and participating SVs. Nonetheless, there exist major challenges when implementing the FL in IoV networks, such as dynamic activities and diverse quality-of-information (QoI) from a large number of SVs, VSP's limited payment budget, and profit competition among SVs. In this paper, we propose a novel dynamic FL-based economic framework for an IoV network to address these challenges. Specifically, the VSP first implements an SV selection method to determine a set of the best SVs for the FL process according to the significance of their current locations and information history at each learning round. Then, each selected SV can collect on-road information and offer a payment contract to the VSP based on its collected QoI. For that, we develop a multi-principal one-agent contract-based policy to maximize the profits of the VSP and learning SVs under the VSP's limited payment budget and asymmetric information between the VSP and SVs. Through experimental results using real-world on-road datasets, we show that our framework can converge 57\% faster (even with only 10\% of active SVs in the network) and obtain much higher social welfare of the network (up to 27.2 times) compared with those of other baseline FL methods.

\end{abstract}

}

\maketitle

\IEEEdisplaynontitleabstractindextext

\IEEEpeerreviewmaketitle

\IEEEraisesectionheading{\section{Introduction}\label{sec:Int}}

\IEEEPARstart{A}ccording to the recent Allied Market Research's outlook~\cite{AMR:2018}, the global market value of the Internet-of-Vehicles (IoV) is expected to grow by more than 215\% in 2024 due to significant demands of road safety from smart vehicles (SVs), e.g., autonomous cars and electric vehicles. Through the IoV, the vehicular service provider (VSP) can build on-road services to enhance the driving safety for the SV users. In particular, the VSP can first obtain on-road data from active SVs e.g., road conditions, location, and driving activities. Using the centralized learning process, the VSP can then generate the meaningful on-road information with high quality~\cite{Nie:2018}, and share the updated on-road information to requesting customers (e.g., SVs, public transportations, local governments, and mobile users) in its coverage areas. In fact, the aforementioned mechanism has been widely adopted by several commercial on-road applications, e.g., the Placemeter (www.placemeter.com), Google Maps (www.google.com/maps), and Waze (www.waze.com). However, sharing raw on-road data to the VSP in a centralized manner may face the huge computation and storage costs for the VSP, privacy disclosure for the SV users, and severe network congestion in the IoV (when a huge amount of data is transferred over the network). Therefore, federated learning (FL) as a highly-effective distributed machine learning approach has a great potential to address these issues~\cite{Lim:2020}. Specifically, instead of sharing raw on-road data, each SV in the FL-based IoV network can train its local on-road data to generate the local trained model. This local model is then collected by the VSP to update the global on-road model, aiming at enhancing the on-road model accuracy. By doing so, we can mitigate the inherent problems that occur in the centralized learning systems.  

Although the applications of FL-based IoV networks have been studied in a few works, e.g.,~\cite{Samarakoon:2018,Pokhrel:2020,Saputra:2020}, the conventional FL model is inefficient or even impractical to be implemented in real-world IoV networks due to several reasons. First, the behaviors of SVs in the networks are very dynamic in practice, e.g., some SVs move frequently, while some other SVs are occasionally disconnected from the network (e.g., due to unreliable communication channels, multiple handovers, and Internet unvailability at particular periods). Furthermore, training local datasets and collecting local trained models from all the SVs in the networks for the learning process are costly and impractical due to huge communication overheads, especially with a large number of moving SVs. Second, the quality-of-information (QoI), that includes on-road data in various locations at different times~\cite{Song:2014}, obtained by the SVs is remarkably diverse. That makes the concurrent learning process from all SVs even worse when some SVs have low QoI. Third, due to the VSP's limited payment budget and the inherent profit competition among SVs, how to optimally incentivize the ``right'' SVs for their contributions to the FL process while maximizing the profit of VSP still remains as an open issue, especially under the information asymmetry between the VSP and SVs.

To address the first challenge due to dynamic activities of a large number of SVs, solutions to select/schedule only a set of FL local learners, e.g., smartphones, IoT devices, and vehicles, in the wireless networks can be adopted. For example, in~\cite{Yang:2019}, the authors introduce random scheduling, round-robin, and proportional fair methods to choose FL learners. Other random selection methods using asynchronous FL and heterogeneous networks are also discussed in~\cite{Xie:2019} and~\cite{Li:2020}, respectively. 
However, these works assume that the selected FL learners can implement local learning process equally without considering their QoI (the second aforementioned challenge). In fact, the QoI obtained by SVs in the IoV networks varies dramatically~\cite{Song:2014}, e.g., some SVs are not moving for a certain period, and thus they may provide low QoI. Consequently, the low QoI can degrade the prediction accuracy or cause FL performance instability (when the low-quality trained local models are aggregated together with other local models during the learning process)~\cite{Amiri:2020}, leading to a lower profit for the VSP. Furthermore, all the above works do not consider any incentive mechanism to compensate the selected FL learners for the learning process, and thus discourage the FL learners to participate in and contribute high QoI for the next FL processes~\cite{Lim:2020}.  Note that a ``right'' incentive mechanism can also in turn help to ``recruit'' the ``right'' SVs. Designing such an incentive model is challenging, especially under the limited budget of the VSP, the competition among the SVs, as well as the information asymmetry between the VSP and selected SVs.  

In this paper, we aim to address the aforementioned challenges through proposing a novel dynamic FL-based economic framework that selectively engages a subset of best SVs for the FL process at each learning round. 
Specifically, the VSP can first select active SVs whose current locations are within significant areas. These significant areas are determined based on the Average Annual Daily Flow (AADF), i.e., the average number of vehicles passing through pre-defined roads in a particular area on one day~\cite{UK:2020}. 
From the set of active SVs based on the location significance, the VSP can further choose the best SVs based on their information significance (or QoI) history, e.g., due to the VSP's limited payment budget for the learning SVs~\cite{Zheng:2017}. Through combining both location and information significance, the VSP can obtain useful on-road information~\cite{Hussain:2019} and reliable trained model updates~\cite{Lim:2020} from the selected SVs, aiming at enhancing the learning quality of the FL. From this SV selection, the VSP can provide an incentive mechanism for the best SVs which share their current information significance for the FL process. Nonetheless, due to the the VSP's limited payment budget and the existence of multiple learning SVs at each round, each selected SV may  compete with other selected SVs in obtaining the corresponding payment as the incentive of its local training process. To address this issue, the authors in~\cite{Zhan:2020, Khan:2020, Sarikaya:2020} propose a non-collaborative Stackelberg game model where the VSP and mobile devices act as a leader and followers, respectively. However, these methods are only applicable when the mobile devices know the full information of the VSP's payment budget (referred to as \emph{information symmetry}). In practice, the VSP usually keeps its payment budget confidential~(referred to as \emph{information asymmetry}) and the best SVs may not want to follow controls by the VSP completely (due to the conflict of economic interests between them)~\cite{Bolton:2005}. Hence, the above approaches are ineffective to implement in our problem.


Given the above, we develop a \emph{multi-principal one-agent (MPOA)} contract-based economic model~\cite{Bernheim:1986}, where the best SVs (as \emph{principals}) non-collaboratively offer contract agreements (containing information significance and offered payment) to the VSP. The VSP (as \emph{the agent}) is then in charge of optimizing the offered contracts. In this case, we formulate the contract model as a non-collaborative learning contract problem to maximize the profits of the VSP and the best SVs at each round under the VSP's information-asymmetry and common constraints, i.e., the individual rationality (IR) and incentive compatibility (IC). The IR constraints ensure that the VSP always joins the FL process. Meanwhile, the IC constraints guarantee that the VSP always obtains the maximum profit when the best contract determined for the VSP is applied. 
To find the optimal contracts for the SVs, we first transform the contract optimization problem into an equivalent low-complexity problem. We then develop an iterative algorithm which can quickly find the equilibrium solution, i.e., optimal contracts, for the transformed problem. After that, the VSP and selected SVs can execute the FL algorithm to improve the global on-road model accuracy at each learning round. Moreover, at the end of each round, the VSP can update its net profit considering the current round's global model accuracy and freshness (i.e., how up-to-date the global model in the FL process is). We then conduct experiments using real-world on-road datasets from all major roads of 190 areas/local districts and 1.5M traffic accidents in the United Kingdom (UK) between 2000 and 2016. The experiment results demonstrate that our framework can improve the social welfare of the network up to 27.2 times and speed up the learning convergence up to 57\% compared with those of the baseline FL methods~\cite{Yang:2019}.

Although a few recent works, e.g.,~\cite{Kang:2019,Ye:2020}, have studied the contract theory in FL, their proposed approaches have some limitations and are inapplicable in IoV networks. First, they only select the FL learners based on the QoI without considering location significance which is one of the most important factors in IoV networks. Second, both works assume that the FL learners' data with certain quality have been collected 
prior to the FL learner selection. This is impractical in IoV networks where the SVs are dynamically moving, and thus the QoI may vary over time. Third, they are limited to the \emph{one-principal multi-agent}-based contract policy, i.e., not accounting for the competition among FL learners. Moreover, the work in~\cite{Kang:2019} implements the contract optimization only once and utilizes the same FL learners for the entire FL process until reaching a convergence. Meanwhile, the FL learner selection in~\cite{Ye:2020} is determined by the vehicular users (not the central server). As such, there is a possibility that some candidate FL learners with high QoI may not be able to join in the FL process, and thus the overall prediction model accuracy can be degraded. To the end, the major contributions of our work are as follows:

\begin{itemize}
	
	\item Propose a novel dynamic FL-based economic framework for the IoV network, aiming at selecting the best SVs and maximizing profits for the VSP as well as learning SVs to quickly achieve the converged global prediction model for road safety. 
	
	\item Design the dynamic SV selection method taking both location and QoI of SVs into account. As such, we can select the best quality SVs for the FL process, thereby improving the overall learning process and yielding more profits for the VSP. 
	
	\item Develop an MPOA-based contract problem which can incorporate the learning payment competition among the selected SVs under the VSP's limited and unknown payment budget.
	
	\item Propose a transformation method to reduce the complexity of the original optimization problem, and then develop a light-weight iterative algorithm which can quickly find the optimal contracts for SVs. 
	
	\item Theoretically analyze the convergence of the proposed dynamic FL and then provide detailed analysis to show the impacts of the global model accuracy and its freshness on the net profit of the VSP. 
	
	\item Conduct extensive experiments to evaluate the proposed framework using two real-world on-road datasets in the UK. These results provide useful information to help the VSP in designing the effective FL with SV selection method in the IoV network.
	
\end{itemize}
The rest of this paper is organized as follows. Section~\ref{sec:EV_scheme} describes the proposed economic framework for the IoV. Section~\ref{sec:SVS} introduces the proposed SV selection method, and Section~\ref{sec:PF} presents the proposed MPOA contract-based problem and solution. 
Section~\ref{sec:FLS} describes the FL algorithm with the selected SVs. Moreover, the VSP's profit analysis based on the global model accuracy and freshness is provided in Section~\ref{sec:VSP_A}. The performance evaluation is given in Section~\ref{sec:PE}, and then the conclusion is drawn in Section~\ref{sec:Conc}.




%
%

\begin{figure*}[!t]
	\centering
	\includegraphics[scale=0.398]{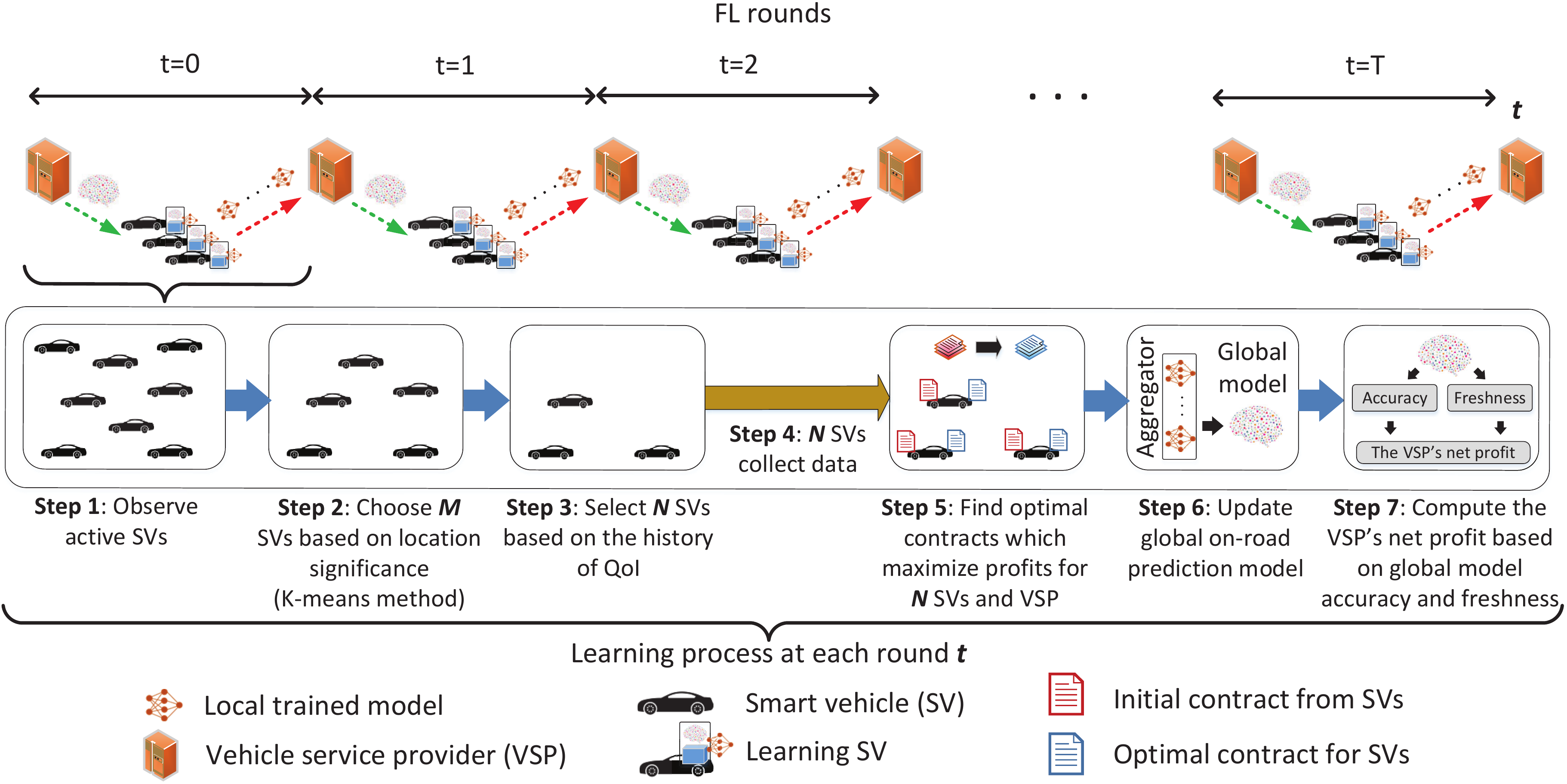}
	\caption{The proposed dynamic FL-based economic framework for the IoV.}
	\vspace{-0.5em}
	\label{fig:SM2}
\end{figure*}

\section{Dynamic Federated Learning-Based Economic Framework}
\label{sec:EV_scheme}


Consider an IoV network including one VSP and multiple active SVs who can participate in the dynamic FL process. Let $\mathcal{I} = \{1,\dots,i,\ldots,I\}$ denote the set of active SVs in the IoV network. Periodically, each SV-$i$ can monitor its current on-road status through its embedded sensor devices, e.g., the camera, weather, and global positioning system (GPS) sensors. From this status, the SV can pre-process and extract meaningful on-road information and store it in a log file at the SV's local storage. In particular, the information may include the SV's visited location history, date, time, light condition (e.g., dark or light), weather condition (e.g., windy, snowy, or rainy), and road surface condition (e.g., dry or wet). To collaborate with the VSP in the FL process, the interested SV users can first deploy an on-road service application developed by the VSP at their SVs' built-in Android or iOS platforms, e.g., Live Traffic NSW (www.livetraffic.com) and Live Traffic Info (www.highwaysengland.co.uk). Then, the SV users can enable the application's access permission to allow the VSP accessing basic information including the SVs' current location and information significance in a secure manner~\cite{Corser:2016}. In this way, the VSP can observe useful information to select the best SVs (which act as the FL learners at each learning round, e.g., every hour or day) within its coverage (e.g., via road side units~\cite{ITS:2017}) with minimum privacy disclosure of the SVs' collected on-road data. Note that the accessed current location information is only used to help the VSP select the best SVs for each learning round (as similarly implemented in location-based services such as Google Maps, Placemeter, and Waze), and is totally different from on-road sensing data that the selected SVs collect when they run on the road. Specifically, let $\mathcal{T} = \{0,1,\ldots,t,\ldots,T\}$ denote the set of rounds of an FL process and $\mathcal{M}(t) = \{1,\dots,m,\ldots,M\}$, where $\mathcal{M}(t) \subset \mathcal{I}$, be the set of selected SVs based on their location significance at round $t$. The selection of $M$ active SVs guarantees that they will provide useful on-road information due to their locations in the significant areas~\cite{Hussain:2019}. $\mathcal{N}(t) = \{1,\dots,n,\ldots,N\}$, where $\mathcal{N}(t) \subset \mathcal{M}(t)$, denotes the set of selected best quality SVs based on their information significance at round $t$. The selection of $N$ best SVs can ensure that the local trained model updates are reliable to provide high learning accuracy with faster convergence~\cite{Kang:2019}. Alternatively, we denote the collected data size, i.e., number of samples, at SV-$n$ for each learning round $t$ as $\eta_n(t)$. 

To compensate the selected SVs in $\mathcal{N}(t)$, the VSP allocates a maximum payment budget of $B_{max}(t)$ at each round $t$. Due to its economic benefit, the VSP practically keeps its willingness to pay (in regards to the received information significance from the FL learners) as a private information for the SVs, and it can be expressed as the type of the VSP~\cite{Bolton:2005}. This type is influenced by the VSP's current payment budget: a higher type indicates the willingness to pay more for the SVs in $\mathcal{N}(t)$ due to its higher payment budget. For that, we can define a finite set of the VSP's types as $\Theta = \{\theta_1,\ldots,\theta_j,\ldots,\theta_{J}\}$, where $j, j \in \mathcal{J} = \{1,\ldots,j,\ldots,J\}$, indicates the type index and we have $\theta_J > \theta_{J-1} > \ldots > \theta_j > \ldots > \theta_2 > \theta_1$ without loss of generality. 
Correspondingly, we denote $B_j(t)$ to be the payment budget for the VSP with type $\theta_j$, where $B_j(t) = \frac{\theta_j}{\theta_J}B_{max}(t)$, $j \in \mathcal{J}$. As the type of the VSP is private, the SVs only can observe the VSP's type distribution, i.e., $\rho_j(t)$, where $\sum_{j=1}^{J} \rho_j(t) = 1, \forall j \in \mathcal{J}, \forall t \in \mathcal{T}$~\cite{Xu:2015}. According to $B_j(t)$, the VSP requires to determine the payment proportion for each SV-$n$.  Particularly, given the VSP with type $\theta_j$, we define payment proportion vector of all $N$ SVs as $\boldsymbol{\varrho}(t)=[\boldsymbol{\varrho}_1(t),\ldots,\boldsymbol{\varrho}_j(t),\ldots,\boldsymbol{\varrho}_J(t)]$, where $\boldsymbol{\varrho}_j(t)=[\varrho_j^1(t),\ldots,\varrho_j^n(t),\ldots,\varrho_j^{N}(t)]$, and $0 \leq \varrho_j^n(t) \leq 1, \forall j \in \mathcal{J}, \forall n \in \mathcal{N}(t)$. 
Moreover, for each learning round $t$, we denote $\boldsymbol{\zeta}(t) = [\boldsymbol{\zeta}_1(t),\ldots,\boldsymbol{\zeta}_j(t),\ldots,\boldsymbol{\zeta}_J(t)]$ and $\boldsymbol{\varphi}(t) = [\boldsymbol{\varphi}_1(t),\ldots,\boldsymbol{\varphi}_j(t),\ldots,\boldsymbol{\varphi}_J(t)]$ to be the vectors of information significance and corresponding payments of all participating SVs for all types, respectively. In this case, $\boldsymbol{\zeta}_j(t) = \big[\zeta_j^1(t),\ldots,\zeta_j^n(t),\ldots,\zeta_j^{N}(t)\big]$, and $\boldsymbol{\varphi}_j(t) = \big[\varphi_j^1(t),\ldots,\varphi_j^n(t),\ldots,\varphi_j^{N}(t)\big]$. As the VSP's type 
increases, each SV-$n$ can offer higher information significance to the VSP since it can obtain higher payment from the VSP. Overall, the proposed dynamic FL-based economic framework for the IoV can be illustrated in Fig.~\ref{fig:SM2}. For each learning round $t$, the VSP performs the following processes:
\begin{itemize}
	\item Step 1: The VSP observes active SVs in the network.
	
	\item Step 2: The VSP selects $M$ active SVs based on their location significance using the K-means algorithm.
	
	\item Step 3: From $M$ SVs, the VSP selects $N$ best SVs which have the highest QoI (defined in Section~\ref{sec:SVS}).
	
	\item Step 4: The $N$ selected SVs collect on-road information data for the current learning round.
	
	\item Step 5: In the end of each learning round, these $N$ SVs send initial contracts based on the collected on-road data to the VSP. Then, the VSP will determine and send the optimal contracts to the SVs.
	
	\item Step 6: Upon receiving the optimal contracts, all $N$ SVs train their local datasets to obtain local trained models. These local models are then collected by the VSP to update the current global model. 
	
	\item Step 7: The VSP updates its net profit based on the current round's global model accuracy and freshness.
\end{itemize}
At each learning round, the VSP may choose a different set of best SVs because of their diverse mobility behaviors. The above processes are repeated until the global on-road prediction model converges or after a pre-defined number of learning rounds is achieved. More details can be found in the following sections.

%

\section{Location and Information Significance-Based SV Selection}\label{sec:SVS}

In this section, we discuss on how the VSP can choose a set of best active SVs based on the location and information significance prior to the FL process at each learning round $t$. 

\subsection{Location Significance-Based SV Selection}
\label{subsec:LASV}

In this method, the VSP can monitor the current location of the active SVs in the IoV using the GPS information stored in their on-road service applications locally. This current location information sharing can be securely protected, e.g., using anonymization or obfuscation-based approaches~\cite{Corser:2016,Lim:2017}.
From the active SVs observed in each learning round, the VSP can select $M$ SVs, where $M \leq I$, whose locations are in the significant areas (i.e., the areas are of high interests by the VSP and/or the areas with high vehicle traffic volume based on their total AADFs for specific periods), e.g., big cities, central business districts, or tourist attractions. Let $\mathcal{D} = \{1,\ldots,d,\ldots,D\}$ denote the set of considered areas, the fixed total AADF for each area-$d$ prior to the learning process can be expressed by~\cite{UK:2020} 
\begin{equation}
\label{eqn:for15a}
\begin{aligned}
V_{d} = \sum_{\varpi_1=1}^{\varpi_1^{max}}\sum_{\varpi_2=1}^{\varpi_2^{max}} \frac{v_d^{\varpi_1,\varpi_2}}{365},
\end{aligned}
\end{equation}
where $\varpi_1 \in \{1,2,\ldots,\varpi_1^{max}\}$ and $\varpi_2 \{1,2,\ldots,\varpi_2^{max}\}$ specify the indices of considered year periods and days for a year (i.e., $\varpi_2^{max} = 365$), respectively. In addition, $v_d^{\varpi_1,\varpi_2}$ represents the vehicle traffic volume in area-$d$ for day $\varpi_2$ of year $\varpi_1$. The area with higher $V_d$ value implies the area with higher significance. The reason is that many vehicles tend to visit the centralized areas where profitable employment market, desirable lifestyle, better educational occasion, and tourist attractions exist. In this case, the vehicles in those areas typically accommodate more useful and meaningful on-road information, e.g., traffic jams and/or accidents, due to their high vehicle traffic volume~\cite{Hussain:2019}. 

Using $V_d$ from all areas in $\mathcal{D}$, the VSP can categorize those areas into insignificant and significant areas. For that, we implement a binary classification scheme using the K-means algorithm~\cite{Bradley:2000} in which groups `0' and `1' represent insignificant and significant areas, respectively. Consider that the VSP has a road traffic dataset containing area IDs, their major road locations, and their total AADFs of the vehicles $V_d, \forall d \in \mathcal{D}$. To attain the sets of insignificant and significant areas, the total AADF centroids ${\hat V}^0$ and ${\hat V}^1$ that minimize the overall squared distance between each $V_d$ and the centroids can be determined using the following: 
\begin{equation}
\label{eqn3j}
\begin{aligned}
\min_{\{\boldsymbol{\sigma}, \mathbf{\hat V}\}}\sum_{d \in \mathcal{D}}\sigma_d^0 (V_d - {\hat V}^0)^2 + \sigma_d^1 (V_d - {\hat V}^1)^2,
\end{aligned}
\end{equation}
\begin{align}
\text{s.t.} \quad &\sigma_d^0 + \sigma_d^1 = 1, \forall d \in \mathcal{D}, \label{eqn3j2} \\
&\sigma_d^0, \sigma_d^1  \in \{0,1\}, \forall d \in \mathcal{D}, \label{eqn3j3} 
\end{align}
where $\sigma_d^0$ and $\sigma_d^1$ are binary variables which indicate the total AADF closeness of area-$d$ to the centroids ${\hat V}^0$ and ${\hat V}^1$, respectively. The constraints (\ref{eqn3j2}) imply that the total AADF of each area-$d$ can be classified into one group only. 
To find the optimal sets of insignificant and significant areas, the centroids ${\hat V}^0$ and ${\hat V}^1$ are first updated at each K-means iteration $\iota$, i.e., ${\hat V}^0_{(\iota+1)}=\frac{\underset{d \in \mathcal{D}}{\sum}\sigma^{0}_{d,(\iota)}V_d}{\underset{d \in \mathcal{D}}{\sum}\sigma^{0}_{d,(\iota)}}$
and ${\hat V}^1_{(\iota+1)}=\frac{\underset{d \in \mathcal{D}}{\sum}\sigma^{1}_{d,(\iota)}V_d}{\underset{d \in \mathcal{D}}{\sum}\sigma^{1}_{d,(\iota)}}$, respectively, until ${\hat V}^0_{(\iota+1)} = {\hat V}^0_{(\iota)}$ and ${\hat V}^1_{(\iota+1)} = {\hat V}^1_{(\iota)}$. Hence, the VSP can obtain the optimal sets of insignificant areas, i.e., $\mathcal{D}^0 \subset \mathcal{D}$, and significant areas, i.e., $\mathcal{D}^1 \subset \mathcal{D}$, based on the value `1' of both $\sigma_d^0$ and $\sigma_d^1$, $\forall d \in \mathcal{D}$, for all learning rounds. Then, the VSP can dynamically select $M$ SVs whose current locations are within the set of significant areas $\mathcal{D}^1$ at each learning round $t$.

\subsection{Information Significance-Based SV Selection}

From the selected $M$ SVs, the VSP can further choose $N$ SVs as the FL learners, where $N \leq M$, to efficiently execute the FL algorithm for each round $t$. To this end, we introduce the information significance-based SV selection method to obtain $N$ SVs which have the highest QoI. In particular, adopted from~\cite{Song:2014}, suppose that $\mathcal{K} = \{1,\ldots,k,\ldots,K\}$ and $\mathcal{L} = \{1,\ldots,l,\ldots,L\}$ to be the set of captured on-road information timespans, e.g, between day 1 and day 7 for a week, and location IDs, e.g., between location 1 and location 100, from the SVs in $\mathcal{M}(t)$, respectively. Utilizing the timespan-$k$ and location-$l$, each SV-$m$ can calculate a spatio-temporal variability $\eta_m^{k,l}(t)$, i.e., the data size in location-$l$ for timespan-$k$, from its total data size, i.e., $\eta_m(t)$, where $\eta_m(t) = \sum_{k=1}^K\sum_{l=1}^L\eta_m^{k,l}(t)$.
This calculation can be conducted by the on-road service application provided by the VSP and stored at the SVs.

To implement the FL process at each round $t$, the VSP can determine the required spatio-temporal variability at timespan-$k$ and location-$l$ from SV-$m$, i.e., $y_m^{k,l}(t)$. However, due to the different mobile activities of the SV-$m$, $\forall m \in \mathcal{M}(t)$, on the roads in diverse locations for different times, each SV-$m$ may have various actual spatio-temporal variability for each timespan-$k$ and location-$l$, i.e.,~\cite{Song:2014} 
\begin{equation}
\label{eqn:for30c}
\begin{aligned}
x_m^{k,l}(t) = \quad
\left\{	\begin{array}{ll}
\eta_m^{k,l}(t), &\text{ if $\eta_m^{k,l}(t) \leq  y_m^{k,l}(t)$}, \\
y_m^{k,l}(t), &\text{ otherwise}.
\end{array}	\right.\\
\end{aligned}
\end{equation}
Equation (\ref{eqn:for30c}) implies that the actual spatio-temporal variabilities cannot exceed the required ones.
From $x_m^{k,l}(t)$ and $y_m^{k,l}(t)$, we can define two matrices of actual and required spatio-temporal variabilities for each SV-$m$ as follows:
\begin{equation}
\label{eqn:for30d}
\begin{aligned}
\mathbf{X}_m(t) &= 
\begin{pmatrix}
x_m^{1,1}(t) & x_m^{1,2}(t) & \cdots & x_m^{1,L}(t) \\
x_m^{2,1}(t) & x_m^{2,2}(t) & \cdots & x_m^{2,L}(t) \\
\vdots  & \vdots  & \ddots & \vdots  \\
x_m^{K,1}(t) & x_m^{K,2}(t) & \cdots & x_m^{K,L}(t) 
\end{pmatrix} \text{, and } \\
\mathbf{Y}_m(t) &= 
\begin{pmatrix}
y_m^{1,1}(t) & y_m^{1,2}(t) & \cdots & y_m^{1,L}(t) \\
y_m^{2,1}(t) & y_m^{2,2}(t) & \cdots & y_m^{2,L}(t) \\
\vdots  & \vdots  & \ddots & \vdots  \\
y_m^{K,1}(t) & y_m^{K,2}(t) & \cdots & y_m^{K,L}(t) 
\end{pmatrix}.
\end{aligned}
\end{equation}
Then, we can derive the information significance metric as the normalized spatial length of a matrix~\cite{Custodio:2010} in Frobenius norm using (\ref{eqn:for30d}), which is
\begin{equation}
\label{eqn:for30f}
\begin{aligned}
\zeta^m(t) &= 1 - \frac{\|\mathbf{Y}_m(t) - \mathbf{X}_m(t)\|_F}{\|\mathbf{Y}_m(t)\|_F}\\
&= 1 - \frac{\sqrt{\overset{K}{\underset{k=1}{\sum}}\overset{L}{\underset{l=1}{\sum}}\Big(y_m^{k,l}(t) - x_m^{k,l}(t)\Big)^2}}{\sqrt{\overset{K}{\underset{k=1}{\sum}}\overset{L}{\underset{l=1}{\sum}}\Big(y_m^{k,l}(t)\Big)^2}}.
\end{aligned}
\end{equation}
From (\ref{eqn:for30f}), we can observe that the information significance for SV-$m$, $\forall m \in \mathcal{M}(t)$, is between 0 and 1, i.e., $0 \leq \zeta^m(t) \leq 1$. Specifically, $\zeta^m(t) = 0$ implies that no actual spatio-temporal variability is collected by the SV-$m$. Meanwhile, $\zeta^m(t) = 1$ specifies that the required spatio-temporal variabilities for all timespans and location IDs of the SV-$m$ are fully satisfied. Using this information significance metric, the VSP can determine the dataset quality of the SVs in $\mathcal{M}(t)$ without disclosing any sensitive information of the SVs' actual datasets. Then, the VSP can select $N$ SVs that will implement the FL process based on their largest $\zeta^m(t)$ values for each round $t$, i.e.,
\begin{equation}
\label{eqn:for4}
\begin{aligned}
\mathcal{N}(t) = \max_{[N]}\{\zeta^1(t),\ldots,\zeta^M(t)\}.
\end{aligned}
\end{equation}
Equation (\ref{eqn:for4}) specifies that only $N$ SVs will exchange their local on-road models to update the global on-road prediction model at each round $t$ in the FL process.

%
%
%
%
%
%
%
%
%
%
%
%
%
%
%
%
%
%
%
%
%

\section{MPOA-Based Learning Contract Optimization Problem and Solution}\label{sec:PF}

Based on the list of selected SVs, the VSP can inform the selected SVs in $\mathcal{N}(t)$ to first collect on-road data for a particular period. Then, the SVs can send initial contracts at each learning round $t$ including the information significance $\zeta_j^n(t) = \frac{\theta_j}{\theta_J}\zeta^n(t), \forall n \in \mathcal{N}(t), \forall j \in \mathcal{J}$, based on the collected data, and payment $\varphi_j^n(t) = \upsilon\zeta_j^n(t), \forall n \in \mathcal{N}(t), \forall j \in \mathcal{J}$~\cite{Saputra:2020}, where $\upsilon$ is the information significance price unit. These initial contracts will be used to optimize the non-collaborative contract problem at the VSP. In this section, we first describe the MPOA-based contract optimization problem for the VSP (as the agent) and learning SVs in $\mathcal{N}(t)$ (as the principals).  
After that, a transformation method is introduced to reduce the problem complexity and an iterative algorithm is developed to find the equilibrium contract for the selected SVs.

\subsection{Profit Optimization for the VSP}

Considering the payment proportion vector for the VSP with type $\theta_j$, i.e., $\boldsymbol{\varrho}_j(t)$, information significance vector for the VSP with type $\theta_j$, i.e,  $\boldsymbol{\zeta}_j(t)$, and payment vector for the VSP with type $\theta_j$, i.e., $\boldsymbol{\varphi}_j(t)$, the profit that the VSP with type $\theta_j$ can obtain from performing the learning process together with SVs in $\mathcal{N}(t)$ at round $t$ can be formulated by
\begin{equation}
\setlength{\belowdisplayskip}{5pt} 
\label{eqn:for30}
\begin{aligned}
{\mu}^j_{VSP}(t) = \theta_jS(\boldsymbol{\varrho}_j(t),\boldsymbol{\zeta}_j(t)) - C(\boldsymbol{\varrho}_j(t), \boldsymbol{\varphi}_j(t)).
\end{aligned}
\end{equation}
Specifically, $S(\boldsymbol{\varrho}_j(t),\boldsymbol{\zeta}_j(t))$ and $C(\boldsymbol{\varrho}_j(t), \boldsymbol{\varphi}_j(t))$ indicate the satisfaction and cost functions of the VSP, respectively. Moreover, the $\theta_j$ describes the weight of $S(\boldsymbol{\varrho}_j(t),\boldsymbol{\zeta}_j(t))$ for the VSP with type $\theta_j$. As such, the VSP with a higher type has a higher weight due to its willingness to pay more towards the FL process. For the satisfaction function, we use a squared-root function as similarly implemented in~\cite{Xu:2015}. The reason is that the satisfaction increases when the QoI, i.e., information significance, gets better. However, the VSP may have less interest to further increase the satisfaction when the higher data quality triggers less on-road model accuracy improvement~\cite{Samuelson:2005}. Thus, the satisfaction function can be expressed by
\begin{equation}
\label{eqn:for30a}
\begin{aligned}
S(\boldsymbol{\varrho}_j(t),\boldsymbol{\zeta}_j(t)) = \lambda\sqrt{\sum_{n \in \mathcal{N}(t)}\varrho_j^n(t)\zeta_j^n(t)},
\end{aligned}
\end{equation} 
where $\lambda > 0$ is a conversion parameter corresponding to the monetary unit of obtaining a certain information significance based on the existing data trading market~\cite{Xu:2015}. 
For the cost function of the VSP, we can formulate it as the total payment transfer to the partipating SVs in $\mathcal{N}(t)$ with respect to their information significance values used in the FL process, which is
\begin{equation}
\label{eqn:for30a2}
\begin{aligned}
C(\boldsymbol{\varrho}_j(t), \boldsymbol{\varphi}_j(t)) = \sum_{n \in \mathcal{N}(t)}\varrho_j^n(t)\varphi_j^n(t).
\end{aligned}
\end{equation} 
From (\ref{eqn:for30})-(\ref{eqn:for30a2}), we can devise an optimization problem to maximize the profit of VSP with type $\theta_j$ by
\begin{equation}
\label{eqn:for3b}
\begin{aligned}
(\mathbf{P}_1) \phantom{10} & \max_{\boldsymbol{\varrho}_j(t)} \phantom{5} \theta_j S(\boldsymbol{\varrho}_j(t),\boldsymbol{\zeta}_j(t)) -  C(\boldsymbol{\varrho}_j(t),\boldsymbol{\varphi}_j(t)).
\end{aligned}
\end{equation}
\begin{align}
&\text{ s.t. } \quad C(\boldsymbol{\varrho}_j(t),\boldsymbol{\varphi}_j(t) \leq B_j(t), \label{eqn:for3c}  \\
&0 \leq \varrho_j^n(t) \leq 1, \forall n \in \mathcal{N}(t), \label{eqn:for3c2} 
\end{align}
where the constraint (\ref{eqn:for3c}) implies that the total payment for SVs in the FL process cannot exceed the maximum payment budget of the VSP with type $\theta_j$, i.e., $B_j(t)$. Due to the convexity of objective function (\ref{eqn:for3b}), i.e., the satisfaction and cost functions are concave and linear functions, respectively, and the linearity of the constraints~(\ref{eqn:for3c})-(\ref{eqn:for3c2}) in $(\mathbf{P}_1)$, it can be guaranteed that we can find the optimal $\boldsymbol{\hat \varrho}_j(t) = [{\hat \varrho}_j^1(t),\ldots,{\hat \varrho}_j^n(t),\ldots,{\hat \varrho}_j^{N}], \forall j \in \mathcal{J}$, straightforwardly.

\subsection{Profit Optimization for Learning SVs}

In this section, we formulate the FL contract optimization problem to maximize the expected profit of each SV-$n, n \in \mathcal{N}(t)$, individually under the constraints of the VSP. Specifically, from the optimal payment proportion $\boldsymbol{\hat \varrho}_j(t), \forall j \in \mathcal{J}$, each SV-$n$ can calculate its expected profit through taking all possible types of the VSP into account such that
\begin{equation}
\label{eqn:for6}
\begin{aligned}
&\mu_n(\boldsymbol{\zeta}(t), \boldsymbol{\varphi}(t)) = \sum_{j=1}^{J}\Big({\hat \varrho}_j^n(t)\varphi_j^n(t) - {\hat \varrho}_j^n(t)\zeta_j^n(t)\xi_n(t)\Big) \rho_j(t),
\end{aligned}
\end{equation}
where $\xi_n(t)$ indicates the computation and memory costs of SV-$n$ in training the local dataset with certain information significance at round $t$. Additionally, the utilization of $\sum(.)$ specifies that the expected profit of each SV-$n$ at round $t$ depends on the VSP's type distribution $\rho_j(t), \forall j \in \mathcal{J}$.

In this optimization problem, we can obtain the optimal contracts which satisfy the individual rationality (IR) and incentive compatibility (IC) constraints from the VSP. The IR constraints guarantee that the VSP with a certain type will always obtain a positive profit as described in Definition~\ref{DefV1}.
\begin{definition}{}\label{DefV1}
	IR constraints: The VSP must achieve a non-negative profit, i.e.,
	\begin{equation}
	\label{eqn:for7a}
	\begin{aligned}
	\theta_jS(\boldsymbol{\hat \varrho}_j(t),\boldsymbol{\zeta}_j(t)) - C(\boldsymbol{\hat \varrho}_j(t), \boldsymbol{\varphi}_j(t)) \geq 0, \forall j \in \mathcal{J},
	\end{aligned}
	\end{equation}
	for each $t \in \mathcal{T}$ to participate in the FL contract optimization.
\end{definition}

Due to the information asymmetry between the VSP and SVs, in addition to the IR constraints, the optimal contracts must also satisfy the IC constraints which ensure the feasibility of the contracts. In particular, the IC constraints guarantee that the VSP can always obtain the maximum profit if an appropriate contract, i.e., the contract designed for that current type of the VSP, is utilized. This condition can be defined in the following Definition~\ref{DefV2}. 

\begin{definition}{}\label{DefV2}
	IC constraints: The VSP with current type $\theta_j$ will likely to choose a contract designed for its current type $\theta_j$ rather than with another type ${\theta_{j^*}}$, i.e.,
	\begin{align}
	&\theta_jS(\boldsymbol{\hat \varrho}_j(t),\boldsymbol{\zeta}_j(t)) - C(\boldsymbol{\hat \varrho}_j(t), \boldsymbol{\varphi}_j(t)) \geq \label{eqn:for7b} \\
	&\theta_jS(\boldsymbol{\hat \varrho}_{j^*}(t),\boldsymbol{\zeta}_{j^*}(t)) - C(\boldsymbol{\hat \varrho}_{j^*}(t), \boldsymbol{\varphi}_{j^*}(t)), j \neq {j^*}, \forall j, {j^*} \in \mathcal{J}. \nonumber 
	\end{align}
\end{definition}

Given the expected profit of each SV-$n$ in~(\ref{eqn:for6}) as well as IR and IC constraints of the VSP in~(\ref{eqn:for7a}) and (\ref{eqn:for7b}), respectively, we can formulate a non-collaborative learning contract optimization 
due to the selfish nature of participating SVs. As such, the FL contract optimization problem $(\mathbf{P}_2)$ to maximize the expected profit for SV-$n$ independently at the VSP for each $t$ can be expressed by
\begin{equation}
\label{eqn:for9}
\begin{aligned}
(\mathbf{P}_2) \phantom{10} & \underset{\boldsymbol{\zeta}(t), \boldsymbol{\varphi}(t)}{\text{max}} \phantom{5} \mu_{n}\big(\boldsymbol{\zeta}(t), \boldsymbol{\varphi}(t)\big), \forall n \in \mathcal{N}(t),
\end{aligned}
\end{equation}
\begin{align}
&\text{ s.t. } \quad C(\boldsymbol{\hat \varrho}_j(t),\boldsymbol{\varphi}_j(t) \leq B_j(t), \forall j \in \mathcal{J}, \label{eqn:for9a} \\
&\theta_jS(\boldsymbol{\hat \varrho}_j(t),\boldsymbol{\zeta}_j(t)) - C(\boldsymbol{\hat \varrho}_j(t), \boldsymbol{\varphi}_j(t)) \geq 0, \forall j \in \mathcal{J}, \label{eqn:for9b} \\
&\theta_jS(\boldsymbol{\hat \varrho}_j(t),\boldsymbol{\zeta}_j(t)) - C(\boldsymbol{\hat \varrho}_j(t), \boldsymbol{\varphi}_j(t)) \geq  \label{eqn:for9c} \\
&\theta_jS(\boldsymbol{\hat \varrho}_{j^*}(t),\boldsymbol{\zeta}_{j^*}(t)) - C(\boldsymbol{\hat \varrho}_{j^*}(t), \boldsymbol{\varphi}_{j^*}(t)), j \neq {j^*}, \forall j, {j^*} \in \mathcal{J}. \nonumber
\end{align}
Using the optimization problem $(\mathbf{P}_2)$, each SV-$n$ via the VSP can maximize its profit under the competition from other SVs in $\mathcal{N}(t)$ and the limited payment budget of the VSP.


\subsection{Social Welfare of the Internet-of-Vehicles}

To further derive the social welfare of the IoV, i.e, the total actual profits of the VSP and all learning SVs in $\mathcal{N}(t)$, we can find the actual profit of each SV-$n$ when the VSP has the type $\theta_j$ as ${\hat \varrho}_j^n(t)\varphi_j^n(t) - {\hat \varrho}_j^n(t)\zeta_j^n(t)\xi_n(t)$.
Then, given the VSP has type $\theta_j$, we can calculate the total actual profit of all SVs in $\mathcal{N}(t)$ by
\begin{equation}
\label{eqn:for6b}
\begin{aligned}
\sum_{n \in \mathcal{N}(t)}\Big({\hat \varrho}_j^n(t)\varphi_j^n(t) - {\hat \varrho}_j^n(t)\zeta_j^n(t)\xi_n(t)\Big).
\end{aligned}
\end{equation}
As a result, we can calculate the social welfare at round $t$ for type $\theta_j$ according to~(\ref{eqn:for30})-(\ref{eqn:for30a2}) and (\ref{eqn:for6b}) by
\begin{align}
&\mu^j_{SW}(t) = \theta_j\lambda\sqrt{\sum_{n \in \mathcal{N}(t)}{\hat \varrho}_j^n(t)\zeta_j^n(t)} - \sum_{n \in \mathcal{N}(t)}{\hat \varrho}_j^n(t)\varphi_j^n(t) \nonumber \\ &+ \sum_{n \in \mathcal{N}(t)}\Big({\hat \varrho}_j^n(t)\varphi_j^n(t) - {\hat \varrho}_j^n(t)\zeta_j^n(t)\xi_n(t)\Big) \label{eqn:for6c} \\
&= \theta_j\lambda\sqrt{\sum_{n \in \mathcal{N}(t)}{\hat \varrho}_j^n(t)\zeta_j^n(t)} - \sum_{n \in \mathcal{N}(t)}{\hat \varrho}_j^n(t)\zeta_j^n(t)\xi_n(t), \forall j \in \mathcal{J}. \nonumber
\end{align}


\subsection{Contract Problem Transformation}

In Proposition~\ref{lemma0a}, we can observe that solving $(\mathbf{P}_2)$ requires computational complexity $O(J^2)$. As a result, this will incur a large of computational resources, especially when the number of VSP's types and the number of learning rounds are high.

\begin{proposition}
	\label{lemma0a}
	The computational complexity of solving $(\mathbf{P}_2)$ is $O(J^2)$. 
\end{proposition}
\begin{proof}
	From (\ref{eqn:for9a})-(\ref{eqn:for9c}), we can observe that there are $J$ payment budget constraints, $J$ IR constraints, and $J(J-1)$ IC constraints. 
	Thus, the computational complexity of solving problem $(\mathbf{P}_2)$ is $O(J^2)$.
\end{proof}
Then, using the following transformation method, i.e., transforming the IR and IC constraints, we can reduce the computational complexity of solving $(\mathbf{P}_2)$ from $O(J^2)$ into $O(J)$. Specifically, we first demonstrate that when the VSP's type is higher (i.e., the willingness to pay more due to higher payment budget), the learning SV-$n, \forall n \in \mathcal{N}(t)$, will offer higher information significance values to the VSP. This aims to obtain higher payments from the VSP. This condition can be formally written in Lemma~\ref{lemma1}.
\begin{lemma}
	\label{lemma1}
	Let $\big(\boldsymbol{\zeta}(t), \boldsymbol{\varphi}(t)\big)$ denote any feasible contracts from the SVs to the VSP such that $\boldsymbol{\zeta}_j(t) \geq \boldsymbol{\zeta}_{j^*}(t)$ if and only if $\theta_j \geq \theta_{j^*}$, where $j \neq {j^*}, j, {j^*} \in \mathcal{J}$.
\end{lemma}
\begin{proof}
	See Appendix~\ref{appx:lemma1}.
\end{proof}

Since $\boldsymbol{\zeta}_j(t) \geq \boldsymbol{\zeta}_{j^*}(t)$, in the following Proposition~\ref{prop1}, we show that the SVs will request higher payment to the VSP correspondingly when the VSP has type $\theta_j$ compared with when it occupies type $\theta_{j^*}$.
\begin{proposition}
	\label{prop1}
	If $\boldsymbol{\zeta}_j(t) \geq \boldsymbol{\zeta}_{j^*}(t)$, then $\boldsymbol{\varphi}_j(t) \geq \boldsymbol{\varphi}_{j^*}(t)$,	where $j \neq {j^*}, j, {j^*} \in \mathcal{J}$.
\end{proposition}
\begin{proof}
	See Appendix~\ref{appx:prop1}.
\end{proof}

According to Lemma~\ref{lemma1} and Proposition~\ref{prop1}, we can observe in the following Proposition~\ref{prop2} that the profit of the VSP is a monotonic increasing function of its type.
\begin{proposition}
	\label{prop2}
	For any feasible contract $\big(\boldsymbol{\zeta}(t), \boldsymbol{\varphi}(t)\big)$, the profit of the VSP must satisfy
	\begin{align}
	&\theta_jS(\boldsymbol{\hat \varrho}_j(t),\boldsymbol{\zeta}_j(t)) - C(\boldsymbol{\hat \varrho}_j(t), \boldsymbol{\varphi}_j(t)) \geq \nonumber \\
	&\theta_{j^*}S(\boldsymbol{\hat \varrho}_{j^*}(t),\boldsymbol{\zeta}_{j^*}(t)) - C(\boldsymbol{\hat \varrho}_{j^*}(t), \boldsymbol{\varphi}_{j^*}(t)) \label{eqn:for10a},
	\end{align}
	where $\theta_j \geq \theta_{j^*}$, $j \neq {j^*}, j, {j^*} \in \mathcal{J}$.
\end{proposition}
\begin{proof}
	See Appendix~\ref{appx:prop2}.
\end{proof}
From Proposition~\ref{prop2}, we can reduce the number of IR constraints by using the VSP's minimum type, i.e., $\theta_1$. As such, we can derive that
\begin{align}
&\theta_jS(\boldsymbol{\hat \varrho}_j(t),\boldsymbol{\zeta}_j(t)) - C(\boldsymbol{\hat \varrho}_j(t), \boldsymbol{\varphi}_j(t)) \geq \nonumber \\ 
&\theta_jS(\boldsymbol{\hat \varrho}_1(t),\boldsymbol{\zeta}_1(t)) - C(\boldsymbol{\hat \varrho}_1(t), \boldsymbol{\varphi}_1(t)) \label{eqn:for12}   \\
& \geq \theta_{1}S(\boldsymbol{\hat \varrho}_1(t),\boldsymbol{\zeta}_{1}(t)) - C(\boldsymbol{\hat \varrho}_1(t), \boldsymbol{\varphi}_{1}(t)) \geq 0. \nonumber
\end{align}
In other words, the IR constraints for other types $\theta_j$, $j > 1$, will automatically hold if and only if we can satisfy the IR constraint for $\theta_1$. To this end, we can modify the IR constraints in~(\ref{eqn:for9b}) into
\begin{equation}
\label{eqn:for13}
\begin{aligned}
\theta_{1}S(\boldsymbol{\hat \varrho}_1(t),\boldsymbol{\zeta}_{1}(t)) - C(\boldsymbol{\hat \varrho}_1(t), \boldsymbol{\varphi}_{1}(t)) \geq 0.
\end{aligned}
\end{equation}

For the IC constraints in~(\ref{eqn:for9c}), we can also reduce the number of constraints using the following transformation stated in Lemma~\ref{lemma2}.
\begin{lemma}
	\label{lemma2}
	The IC constraints in~(\ref{eqn:for9c}) of $(\mathbf{P}_2)$ can be modified into the following local downward incentive constraints (LDIC), i.e.,
	\begin{align}
	&\theta_jS(\boldsymbol{\hat \varrho}_j(t),\boldsymbol{\zeta}_j(t)) - C(\boldsymbol{\hat \varrho}_j(t), \boldsymbol{\varphi}_j(t)) \geq \nonumber \\
	&\theta_jS(\boldsymbol{\hat \varrho}_{j-1}(t),\boldsymbol{\zeta}_{j-1}(t)) - C(\boldsymbol{\hat \varrho}_{j-1}(t), \boldsymbol{\varphi}_{j-1}(t)), \label{eqn:for14} \\
	&\forall j \in \{2, \ldots, J\}, \nonumber
	\end{align}
	where $\boldsymbol{\zeta}_j(t) \geq \boldsymbol{\zeta}_{j-1}(t), \forall j \in \{2, \ldots, J\}$.
	
\end{lemma}
\begin{proof}
	See Appendix~\ref{appx:lemma2}.
\end{proof}
The conditions in~(\ref{eqn:for14}) imply that if the IC constraint regarding the type that is one type lower than type $\theta_j$ holds, i.e., $\theta_{j-1}$, then all other IC constraints are also satisfied as long as the conditions in Lemma~\ref{lemma1} hold. From~(\ref{eqn:for13})-(\ref{eqn:for14}), we can update the optimization problem $(\mathbf{P}_2)$ into the problem $(\mathbf{P}_3)$ as follows:
\begin{equation}
\label{eqn:for9rev}
\begin{aligned}
(\mathbf{P}_3) \phantom{10} & \underset{\boldsymbol{\zeta}(t), \boldsymbol{\varphi}(t)}{\text{max}} \phantom{5} \mu_{n}\big(\boldsymbol{\zeta}(t), \boldsymbol{\varphi}(t)\big), \forall n \in \mathcal{N}(t),
\end{aligned}
\end{equation}
\begin{align}
&\text{ s.t. } \quad \text{ 
	(\ref{eqn:for9a}), and}, \nonumber \\ 
&\theta_{1}S(\boldsymbol{\hat \varrho}_1(t),\boldsymbol{\zeta}_{1}(t)) - C(\boldsymbol{\hat \varrho}_1(t), \boldsymbol{\varphi}_{1}(t)) \geq 0,  \label{eqn:for9arev} \\
&\theta_jS(\boldsymbol{\hat \varrho}_j(t),\boldsymbol{\zeta}_j(t)) - C(\boldsymbol{\hat \varrho}_j(t), \boldsymbol{\varphi}_j(t)) \geq \nonumber\\
&\theta_jS(\boldsymbol{\hat \varrho}_{j-1}(t),\boldsymbol{\zeta}_{j-1}(t)) - C(\boldsymbol{\hat \varrho}_{j-1}(t), \boldsymbol{\varphi}_{j-1}(t)), \nonumber \\  
&\forall j \in \{2, \ldots, J\}, \label{eqn:for9brev} \\
&\boldsymbol{\zeta}_j(t) \geq \boldsymbol{\zeta}_{j-1}(t), \forall j \in \{2, \ldots, J\}. \label{eqn:for9crev}
\end{align}
From problem $(\mathbf{P}_3)$, we demonstrate that the problem now can be solved with computational complexity $O(J)$ as stated in the following Proposition~\ref{lemma0b}.
\begin{proposition}
	\label{lemma0b}
	The computational complexity of solving $(\mathbf{P}_3)$ is $O(J)$. 
\end{proposition}
\begin{proof}
	From the constraints (\ref{eqn:for9a}), (\ref{eqn:for9arev})-(\ref{eqn:for9crev}), there are $J + 1 + (J - 1) + (J - 1) = 3J - 1$ number of the constraints. Hence, the computational complexity of solving problem $(\mathbf{P}_3)$ becomes $O(J)$.
\end{proof}

\subsection{Learning Contract Iterative Algorithm}


Based on the problem $(\mathbf{P}_3)$, we can find the optimal contracts $\Big(\boldsymbol{\hat \zeta}(t), \boldsymbol{\hat \varphi}(t)\Big)$ by implementing an iterative algorithm 
as described in~Algorithm~\ref{ICEA}. Specifically, the SVs in $\mathcal{N}(t)$ first offer the initial contracts at iteration $\varsigma = 0$, i.e., $\Big(\boldsymbol{\zeta}^n_{(\varsigma)}(t),\boldsymbol{\varphi}^n_{(\varsigma)}(t)\Big)$, $\forall n \in \mathcal{N}(t)$, where $\boldsymbol{\zeta}^n_{(\varsigma)}(t) = [\zeta^n_{1,(\varsigma)}(t),\ldots,\zeta^n_{j,(\varsigma)}(t),\ldots,\zeta^n_{J,(\varsigma)}(t)]$ and $\boldsymbol{\varphi}^n_{(\varsigma)}(t) = [\varphi^n_{1,(\varsigma)}(t),\ldots,\varphi^n_{j,(\varsigma)}(t),\ldots,\varphi^n_{J,(\varsigma)}(t)]$. Given $\Big(\boldsymbol{\zeta}^n_{(\varsigma)}(t),\boldsymbol{\varphi}^n_{(\varsigma)}(t)\Big)$,  $\forall n \in \mathcal{N}(t)$, the VSP can find the optimal payment proportion values $\boldsymbol{\hat \varrho}_{(\varsigma)}(t)$ which maximize the objective function of $(\mathbf{P}_1)$. Then, using the obtained $\boldsymbol{\hat \varrho}_{(\varsigma)}(t)$ and other SVs' current contracts $\Big(\boldsymbol{\zeta}^{-n}_{(\varsigma)}(t),\boldsymbol{\varphi}^{-n}_{(\varsigma)}(t)\Big)$~\cite{Zhu:2015}, the VSP can iteratively find a new contract for the SV-$n$, $n \in \mathcal{N}(t)$, to maximize the objective function of $(\mathbf{P}_3)$ at iteration $\varsigma+1$. 

The algorithm continues until the gaps between the expected profits at iteration $\varsigma$ and $\varsigma+1$ of all the SVs is equal or less than the optimality tolerance $\gamma$, 
and thus the algorithm converges and the equilibrium contract solution can be found (as proven in Theorems~\ref{theorem_equi1} and \ref{theorem_equi2}). In this case, the expected profit of the SV-$n$ at the equilibrium contract solution $\Big(\boldsymbol{\hat \zeta}^n(t),\boldsymbol{\hat \varphi}^n(t)\Big)$ will achieve the highest value compared with all other contract solutions it takes, given the equilibrium contract of other SVs $\Big(\boldsymbol{\hat \zeta}^{-n}(t), \boldsymbol{\hat \varphi}^{-n}(t)\Big)$. 
Alternatively, no SV-$n$ can improve its expected profit by unilaterally deviating from its equilibrium strategy $\Big(\boldsymbol{\hat \zeta}^n(t),\boldsymbol{\hat \varphi}^n(t)\Big)$, as described in Definition~\ref{DefV3}. 

\begin{definition}{}\label{DefV3}
	Equilibrium solution for non-collaborative learning contract optimization problem: The optimal contracts $\Big(\boldsymbol{\hat \zeta}(t), \boldsymbol{\hat \varphi}(t)\Big) = \Big(\boldsymbol{\hat \zeta}^n(t), \boldsymbol{\hat \varphi}^n(t), \boldsymbol{\hat \zeta}^{-n}(t), \boldsymbol{\hat \varphi}^{-n}(t)\Big)$ are the equilibrium solution of $(\mathbf{P}_3)$ at learning round $t$ if and only if 
	\begin{align}
	&\mu_n(\boldsymbol{\hat \zeta}^n(t), \boldsymbol{\hat \varphi}^n(t), \boldsymbol{\hat \zeta}^{-n}(t), \boldsymbol{\hat \varphi}^{-n}(t)) \geq \nonumber \\ 
	&\mu_n(\boldsymbol{\zeta}^n(t), \boldsymbol{\varphi}^n(t), \boldsymbol{\hat \zeta}^{-n}(t), \boldsymbol{\hat \varphi}^{-n}(t)), \forall n \in \mathcal{N}(t), \label{eqn:for20}
	\end{align}
	hold and the optimal contracts $\Big(\boldsymbol{\hat \zeta}(t), \boldsymbol{\hat \varphi}(t)\Big)$ still satisfy the constraints (\ref{eqn:for9arev})-(\ref{eqn:for9crev}).
\end{definition}

\setlength{\textfloatsep}{5pt}
\begin{algorithm}[h]
	\caption{Learning Contract Iterative Algorithm} \label{ICEA}
	{\fontsize{9}{10}\selectfont
		\begin{algorithmic}[1] 
			\STATE Set $\varsigma = 0$ and $\gamma > 0$
			
			\STATE The VSP informs the selected SVs in $\mathcal{N}(t)$ to send initial contracts
			
			\FOR{$\forall n \in \mathcal{N}(t)$}
			
			\STATE SV-$n$ offers $\Big(\boldsymbol{\zeta}^n_{(\varsigma)}(t),\boldsymbol{\varphi}^n_{(\varsigma)}(t)\Big)$ to the VSP
			
			\ENDFOR
			
			\REPEAT
			
			\STATE Obtain $\boldsymbol{\hat \varrho}_{(\varsigma)}(t)$ that maximize $(\mathbf{P}_1)$ given $\Big(\boldsymbol{\zeta}_{(\varsigma)}(t), \boldsymbol{\varphi}_{(\varsigma)}(t)\Big)$
			
			\FOR{$\forall n \in \mathcal{N}(t)$}
			
			\STATE Find the new contract $\Big(\boldsymbol{\zeta}^n_{new}(t),\boldsymbol{\varphi}^n_{new}(t)\Big)$, which maximizes $(\mathbf{P}_3)$ using $\boldsymbol{\hat \varrho}_{(\varsigma)}(t)$ and $\Big(\boldsymbol{\zeta}^{-n}_{(\varsigma)}(t),\boldsymbol{\varphi}^{-n}_{(\varsigma)}(t)\Big)$
			
			\IF{$\bigg[{ \mu}_{n}\Big(\boldsymbol{\zeta}^n_{new}(t),\boldsymbol{\varphi}^n_{new}(t),\boldsymbol{\zeta}^{-n}_{(\varsigma)}(t),\boldsymbol{\varphi}^{-n}_{\varsigma)}(t)\Big) - {\mu}_{n}\Big(\boldsymbol{\zeta}^n_{(\varsigma)}(t),\boldsymbol{\varphi}^n_{(\varsigma)}(t),\boldsymbol{\zeta}^{-n}_{(\varsigma)}(t),\boldsymbol{\varphi}^{-n}_{(\varsigma)}(t)\Big)\bigg] > \gamma$}
			
			\STATE Set $\Big(\boldsymbol{\zeta}^n_{(\varsigma+1)}(t),\boldsymbol{\varphi}^n_{(\varsigma+1)}(t)\Big) = \Big(\boldsymbol{\zeta}^n_{new}(t),\boldsymbol{\varphi}^n_{new}(t)\Big)$ 
			
			\ELSE
			
			\STATE Set $\Big(\boldsymbol{\zeta}^n_{(\varsigma+1)}(t),\boldsymbol{\varphi}^n_{(\varsigma+1)}(t)\Big) = \Big(\boldsymbol{\zeta}^n_{(\varsigma)}(t),\boldsymbol{\varphi}^n_{(\varsigma)}(t)\Big)$ 
			
			\ENDIF
			
			\ENDFOR
			
			\STATE $\varsigma = \varsigma + 1$
			
			\UNTIL {${\mu}_{n}\Big(\boldsymbol{\zeta}^n_{(\varsigma)}(t),\boldsymbol{\varphi}^n_{(\varsigma)}(t),\boldsymbol{\zeta}^{-n}_{(\varsigma)}(t),\boldsymbol{\varphi}^{-n}_{(\varsigma)}(t)\Big), \forall n \in \mathcal{N}(t)$, do not change anymore}
			
			\STATE Generate optimal contracts $\Big(\boldsymbol{\hat \zeta}(t), \boldsymbol{\hat \varphi}(t)\Big)$
			
	\end{algorithmic}}
\end{algorithm}

\subsection{Convergence, Equilibrium, and Complexity Analysis of the Learning Contract Iterative Algorithm}

To analyze the convergence and equilibrium contract for Algorithm~\ref{ICEA}, we consider two steps for each iteration $\varsigma$. First, each SV-$n$, $n \in \mathcal{N}(t)$, generates contract $\Big(\boldsymbol{\zeta}^n_{(\varsigma)}(t),\boldsymbol{\varphi}^n_{(\varsigma)}(t)\Big)$ and sends it to the VSP. Then, the VSP helps the SVs to find their optimal payment proportions $\boldsymbol{\hat \varrho}_{(\varsigma)}(t)$ locally due to the unknown contract information among the SVs. In this case, we can show that Algorithm~\ref{ICEA} converges to the equilibrium contract solution by using the best response method~\cite{Zhu:2015}. In particular, the best response of SV-$n$ at iteration $\varsigma+1$ given $\Big(\boldsymbol{\zeta}^{-n}_{(\varsigma)}(t),\boldsymbol{\varphi}^{-n}_{(\varsigma)}(t)\Big)$ can be expressed by
\begin{align}
&\Phi_{(\varsigma+1)}^n\Big(\boldsymbol{\zeta}^{-n}_{(\varsigma)}(t),\boldsymbol{\varphi}^{-n}_{(\varsigma)}(t)\Big) = \label{eqn:for11a} \\ &\underset{\{\boldsymbol{\zeta}^{n}_{new}(t),\boldsymbol{\varphi}^{n}_{new}(t)\} \in \mathbb{C}_n}{\arg \max} { \mu}_{n}\Big(\boldsymbol{\zeta}^n_{new}(t),\boldsymbol{\varphi}^n_{new}(t),\boldsymbol{\zeta}^{-n}_{(\varsigma)}(t),\boldsymbol{\varphi}^{-n}_{\varsigma)}(t)\Big) \nonumber,
\end{align}
where $\mathbb{C}_n$ is the non-empty contract space~\cite{Fraysse:1993} for SV-$n$ and $\mathbb{C} = {\prod}_{n \in \mathcal{N}(t)}\mathbb{C}_n$. As the monotonicity is not ensured for each SV-$n$, we can use $\Big(\boldsymbol{\zeta}^n_{new}(t),\boldsymbol{\varphi}^n_{new}(t)\Big) \in \	\Phi_{(\varsigma+1)}^n$ to be $\Big(\boldsymbol{\zeta}^n_{(\varsigma+1)}(t),\boldsymbol{\varphi}^n_{(\varsigma+1)}(t)\Big)$ if the following condition satisfies
\begin{align}
\bigg[{\mu}_{n}\Big(\boldsymbol{\zeta}^n_{new}(t),\boldsymbol{\varphi}^n_{new}(t),\boldsymbol{\zeta}^{-n}_{(\varsigma)}(t),\boldsymbol{\varphi}^{-n}_{\varsigma)}(t)\Big) - \nonumber \\ {\mu}_{n}\Big(\boldsymbol{\zeta}^n_{(\varsigma)}(t),\boldsymbol{\varphi}^n_{(\varsigma)}(t),\boldsymbol{\zeta}^{-n}_{(\varsigma)}(t),\boldsymbol{\varphi}^{-n}_{(\varsigma)}(t)\Big)\bigg] > \gamma. \label{eqn:for11c}
\end{align}
The above iteration stops when the algorithm converges for all SVs in $\mathcal{N}(t)$ as stated in Theorem~\ref{theorem_equi1}. 

\begin{theorem}\label{theorem_equi1}
	The best response iterative Algorithm~\ref{ICEA} converges under the optimality tolerance $\gamma$.
\end{theorem}
\begin{proof}
	See Appendix~\ref{appx:theorem1}.
\end{proof}
Next, we need to ensure that the Algorithm~\ref{ICEA} also converges to the equilibrium solution $\Big(\boldsymbol{\hat \zeta}(t), \boldsymbol{\hat \varphi}(t)\Big)$. Specifically, we first investigate that the equilibrium solution exists by finding a fixed point in a set-valued function $\Phi$, $\Phi: \mathbb{C} \rightarrow 2^\mathbb{C}$, which is $\Phi = \Big[\Phi^n\Big(\boldsymbol{\zeta}^{-n}(t), \boldsymbol{ \varphi}^{-n}(t)\Big),\Phi^{-n}\Big(\boldsymbol{\zeta}^n(t), \boldsymbol{\varphi}^n(t)\Big)\Big]$.
The revelation of this fixed point is equivalent to the equilibrium solution~\cite{Bernheim:1986, Fraysse:1993}, in which the Algorithm~\ref{ICEA} converges to the equilibrium contract solution, i.e., all the SVs in $\mathcal{N}(t)$ obtain the maximum expected profits where $\Big(\boldsymbol{\hat \zeta}(t), \boldsymbol{\hat \varphi}(t)\Big)$ is found. This is formally stated in Theorem~\ref{theorem_equi2}.
\begin{theorem}\label{theorem_equi2}
	If $\Big(\boldsymbol{\hat \zeta}(t), \boldsymbol{\hat \varphi}(t)\Big)$ is a fixed point in $\Phi$, then the best response iterative process in Algorithm~\ref{ICEA} converges to its equilibrium contract solution $\Big(\boldsymbol{\hat \zeta}(t), \boldsymbol{\hat \varphi}(t)\Big)$.
\end{theorem}
\begin{proof}
	See Appendix~\ref{appx:theorem2}.
\end{proof}
Finally, the complexity of the Algorithm 1 with $N$ SVs is formally stated in the following Theorem~\ref{theorem3}.
\begin{theorem}\label{theorem3} 
	Algorithm~\ref{ICEA} has polynomial complexity $\log(N) + e^z +O(1/N)$, where $z$ is the Euler constant.
\end{theorem}
\begin{proof}
	See Appendix~\ref{appx:theorem3}.
\end{proof} 

\section{Federated Learning with SV Selection}\label{sec:FLS}

\subsection{Learning Process}


To improve the global on-road prediction model accuracy, we implement the FL process with SV selection at each round~\cite{Amiri:2020} using the latest information significance (based on the above optimal contracts $\Big(\boldsymbol{\hat \zeta}(t), \boldsymbol{\hat \varphi}(t)\Big)$ of SVs in $\mathcal{N}(t)$). 
For the FL algorithm, we adopt a deep learning approach leveraging deep neural networks (DNN) for a classification prediction model, i.e., when the output layer of DNN produces discrete prediction values. Particularly, we specify $\mathbf{A}_n(t)$ and $\mathbf{G}_n(t)$ to be the training feature data (using features as columns) and ground-truth label data within $\eta_n(t)$ of SV-$n$, $n \in \mathcal{N}(t)$, respectively. Thus, the total number of samples for all SVs in $\mathcal{N}(t)$ can be derived by $\eta(t) = \sum_{n \in \mathcal{N}(t)}\eta_n(t)$.
When multiple layers of the DNN are considered, we have $\mathbf{A}_n^{\ell}(t)$ such that $\mathbf{A}_n^{1}(t) = \mathbf{A}_n(t)$, where $\ell$ is the training layer, $\ell \in [1,2,\ldots,\ell_{max}]$. For each layer-$\ell$, we can compute the training output data $\mathbf{\hat G}_n^{\ell}(t)$ as described by
\begin{equation}
\label{eqn3a}
\begin{aligned}
\mathbf{\hat G}_n^{\ell}(t) = \alpha_n^\ell \Big(\mathbf{A}_n^\ell(t)\mathbf{W}(t)\Big),
\end{aligned}
\end{equation} 
where $\mathbf{W}(t)$ is the global model matrix containing weights. $\alpha_n^\ell(.)$ indicates the activation function of SV-$n$ for nonlinear transformation. For $\ell \in [1,2,\ldots,\ell_{max}-1]$, we adopt a \emph{tanh} activation function~\cite{Zhang2:2018}, i.e.,
\begin{equation}
\label{eqn3a2}
\begin{aligned}
\alpha_n^\ell\Big(\mathbf{A}_n^\ell(t)\mathbf{W}(t)\Big) =  \frac{e^{\mathbf{A}_n^\ell(t)\mathbf{W}(t)}-e^{-\mathbf{A}_n^\ell(t)\mathbf{W}(t)}}{e^{\mathbf{A}_n^\ell(t)\mathbf{W}(t)}+e^{-\mathbf{A}_n^\ell(t)\mathbf{W}(t)}}.
\end{aligned}
\end{equation}
For the final layer, i.e., $\ell = \ell_{max}$, we utilize a \emph{softmax} activation function which is typically used to interpret a probability distribution for a classification model~\cite{Zhang2:2018}, i.e.,
\begin{equation}
\label{eqn3a3}
\begin{aligned}
\alpha_n^{\ell_{max}}\Big(\mathbf{A}_n^{\ell_{max}}(t)\mathbf{W}(t)\Big) =  \frac{e^{\mathbf{A}_n^{\ell_{max}}(t)\mathbf{W}(t)}}{\sum e^{\mathbf{A}_n^{\ell_{max}}(t)\mathbf{W}(t)}}.
\end{aligned}
\end{equation}

To extract more meaningful features from the training input, we utilize several hidden layers $\ell$, where $1 < \ell < \ell_{max}$, such that $\mathbf{A}_n^{\ell +1}(t) = \mathbf{\hat G}_n^{\ell}(t)$. At the final layer $\ell_{max}$, we can obtain the predicted label data $\mathbf{\hat G}_n^{\ell_{max}}(t)$ to calculate the local loss function of SV-$n$ at each $t$ using a squared Frobenius norm as follows: 
\begin{equation}
\label{eqn3b2}
\begin{aligned}
\epsilon_n\Big(\mathbf{W}(t)\Big) &= \frac{1}{\eta_n(t)}\Big\|\mathbf{G}_n(t) - \mathbf{\hat G}_n^{\ell_{max}}(t)\Big\|_F^2 \\
&= \frac{1}{\eta_n(t)}\sum_{s=1}^{\eta_n(t)}\Big(g_n^s(t) - {\hat g}_n^s(t)\Big)^2,
\end{aligned}
\end{equation}
where $g_n^s(t)$ and ${\hat g}_n^s(t)$ are the elements of ground truth label data $\mathbf{G}_n(t)$ and predicted label data $\mathbf{\hat G}_n^{\ell_{max}}(t)$ at row-$s$ for round $t$, respectively. From $\epsilon_n\Big(\mathbf{W}(t)\Big)$, we can compute the local loss gradient at SV-$n$ by	
\begin{equation}
\label{eqn3c}
\begin{aligned}
\nabla \mathbf{W}_n(t) = \frac{\partial \epsilon_n\Big(\mathbf{W}(t)\Big)}{\partial \mathbf{W}(t)}.
\end{aligned}
\end{equation}
Using this gradient $\nabla \mathbf{W}_n(t)$, each SV-$n$ can update the local on-road model $\mathbf{W}_n(t)$ through utilizing the \emph{Adam} optimizer~\cite{Kingma:2015} which minimizes $\epsilon_n\Big(\mathbf{W}(t)\Big)$. This optimizer operates as the adaptive learning rate to obtain high robustness and fast convergence to the $\mathbf{W}_n(t)$. In this case, we can derive the local update rules of the exponential moving average of the $\nabla \mathbf{W}_n(t)$, i.e., $p_n(t)$, and the squared $\nabla \mathbf{W}_n(t)$ to obtain the variance, i.e., $q_n(t)$, as follows:
\begin{equation}
\begin{aligned}
\label{eqn3d1}
p_n^{(\tau+1)}(t) &= \beta_{p_n}^{(\tau)}(t) p_n^{(\tau)}(t) + \Big(1 - \beta_{p_n}^{(\tau)}(t)\Big)\nabla \mathbf{W}_n(t),\\
q_n^{(\tau+1)}(t) &= \beta_{q_n}^{(\tau)}(t) q_n^{(\tau)}(t) + \Big(1 - \beta_{q_n}^{(\tau)}(t)\Big)\Big(\nabla \mathbf{W}_n(t)\Big)^2,
\end{aligned}
\end{equation}
where $\beta_{p_n}^{(\tau)}(t) \in [0,1)$ and $\beta_{q_n}^{(\tau)}(t) \in [0,1)$ represent $p_n^{(\tau)}(t)$'s and $q_n^{(\tau)}(t)$'s steps of the exponential decays at local iteration $\tau$, respectively. Using the learning step $\kappa_n$, we can determine how often the $\mathbf{W}_n(t)$ is updated for the next local iteration $\tau+1$, which can be calculated by
\begin{equation}
\label{eqn3d}
\begin{aligned}
\kappa_n^{(\tau+1)}(t) = \kappa_n\frac{\sqrt{1 - \beta_{q_n}^{(\tau+1)}(t)}}{1 - \beta_{p_n}^{(\tau+1)}(t)}.
\end{aligned}
\end{equation}
Then, we can update the  $\mathbf{W}_n^{(\tau+1)}(t)$ for the next $\tau+1$ by
\begin{equation}
\label{eqn3e}
\begin{aligned}
\mathbf{W}_n^{(\tau+1)}(t) = \mathbf{W}_n^{(\tau)}(t) - \kappa_n^{(\tau+1)}(t)\frac{{p_n}^{(\tau+1)}(t)}{\sqrt{{q_n}^{(\tau+1)}(t)} + \varepsilon},
\end{aligned}
\end{equation}
where $\varepsilon$ specifies a constant to avoid zero division when $\sqrt{{q_n}^{(\tau+1)}(t)}$ reaches zero. This local process continues until a pre-defined local iteration threshold $\tau_{\emph{\mbox{th}}}$ is reached, and thus $\kappa_n^{(\tau_{\emph{\mbox{th}}})}(t)$ and $\mathbf{W}_n^{(\tau_{\emph{\mbox{th}}})}(t)$ are obtained. To this end, the SV-$n$ can send its $\mathbf{W}_n^{(\tau_{\emph{\mbox{th}}})}(t)$ to the VSP for the global on-road model update $\mathbf{W}(t+1)$ by aggregating all $\mathbf{W}_n^{(\tau_{\emph{\mbox{th}}})}(t), \forall n \in \mathcal{N}(t)$, as follows:
\begin{equation}
\label{eqn3i}
\begin{aligned}
\mathbf{W}(t+1) = \frac{1}{\eta(t)}\sum_{n \in \mathcal{N}(t)}\eta_n(t)\mathbf{W}_n^{(\tau_{\emph{\mbox{th}}})}(t).
\end{aligned}
\end{equation}
To this end, we can compute the global loss function at round $t+1$ which is expressed by
\begin{equation}
\label{eqn3i2}
\begin{aligned}
\Psi(\mathbf{W}(t+1)) = \frac{1}{N}\sum_{n \in \mathcal{N}(t)}\epsilon_n\Big(\mathbf{W}_n^{(\tau_{\emph{\mbox{th}}})}(t)\Big).
\end{aligned}
\end{equation}
This global process repeats until the global loss converges or the number of rounds reaches a given threshold $t_{\emph{\mbox{th}}}$, and thus the final global on-road model $\mathbf{W}^*$ and the final global loss $\Psi^*(\mathbf{W}^*)$ are obtained. The proposed FL algorithm with SV selection and MPOA contract-based model is summarized in Algorithm~\ref{DDL-PC}. 

\begin{algorithm}[t]
	{\fontsize{9}{10}\selectfont
		\caption{FL Process with SV Selection Algorithm} \label{DDL-PC}
		
		\begin{algorithmic}[1] 
			
			\STATE Set initial $\tau_{th}$, $t_{th}$, $\mathbf{W}(0)$, and $t=0$
			
			\WHILE{$t \leq t_{th} \text{ {\bf and} } \Psi(\mathbf{W}(t)) \text{ does not converge }$}
			
			\STATE The VSP determines SVs in $\mathcal{N}(t) \subset \mathcal{M}(t) \subset \mathcal{I}$ using (\ref{eqn:for4})
			
			\STATE All SVs in $\mathcal{N}(t)$ collect on-road data for a certain time
			
			\STATE Execute Algorithm~\ref{ICEA} for all SVs in $\mathcal{N}(t)$ to obtain optimal contracts $\Big(\boldsymbol{\hat \zeta}(t), \boldsymbol{\hat \varphi}(t)\Big)$
			
			\FOR{$\forall n \in \mathcal{N}(t)$}
			
			\STATE Set $\mathbf{A}_n(t)$, $\mathbf{G}_n(t)$ from $\eta_n(t)$ based on the optimal $\boldsymbol{\hat \zeta}^n(t)$
			
			\STATE Calculate $\mathbf{\hat G}_n^{\ell_{max}}(t)$ using $\mathbf{A}_n(t)$ and $\mathbf{W}(t)$
			
			\STATE Derive $\epsilon_n\Big(\mathbf{W}(t)\Big)$ and $\mathbf{W}_n^{(\tau_{\emph{\mbox{th}}})}(t)$
			
			\STATE Send $\mathbf{W}_n^{(\tau_{\emph{\mbox{th}}})}(t)$ to the VSP
			
			\ENDFOR
			
			\STATE Update the global on-road model $\mathbf{W}(t+1)$ using $\mathbf{W}_n^{(\tau_{\emph{\mbox{th}}})}(t)$ and $\eta_n(t)$, $\forall n \in \mathcal{N}(t)$
			
			\STATE Determine the global loss $\Psi(\mathbf{W}(t+1))$ using $\epsilon_n\Big(\mathbf{W}_n^{(\tau_{\emph{\mbox{th}}})}(t)\Big)$, $\forall n \in \mathcal{N}(t)$
			
			\STATE $t = t + 1$
			
			\ENDWHILE
			
			\STATE Obtain the final global on-road model $\mathbf{W}^*$ and global loss $\Psi^*(\mathbf{W}^*)$
			%
			
	\end{algorithmic}}
\end{algorithm}

\subsection{Convergence Analysis}

In this section, we investigate the convergence of the proposed FL algorithm with SV selection using the gap between the expected global loss after $t^\diamond$ rounds, i.e., $\mathbb{E}[\Psi(\mathbf{W}(t^\diamond))]$ and the final global loss, i.e., $\Psi^*(\mathbf{W}^*)$. Consider the probability that an SV will be choosen by the VSP for the FL execution at any round is $\frac{N}{I}$. Then, we can demonstrate that the global loss gap is upper bounded by the expected squared L2-norm global model gap, i.e., $\mathbb{E}[\|\mathbf{W}(t^\diamond) - \mathbf{W}^*\|_2^2]$, and it eventually reaches zero as stated in Theorem~\ref{theorem1}.

\begin{theorem}
	\label{theorem1}
	The proposed FL with SV selection in Algorithm~\ref{DDL-PC} will converge to the minimum global loss $\Psi^*(\mathbf{W}^*)$ according to the global loss gap condition $\big[\mathbb{E}[\Psi(\mathbf{W}(t^\diamond))] - \Psi^*(\mathbf{W}^*)\big] \leq \frac{\delta}{2}\mathbb{E}[\|\mathbf{W}(t^\diamond) - \mathbf{W}^*\|_2^2]$,
	where $\delta$ is a positive constant. 
\end{theorem}

\begin{proof}
	Given that the global loss function after $t^\diamond$ rounds, i.e., $\Psi(\mathbf{W}(t^\diamond))$, is $\delta$-smooth, we can show that $\mathbb{E}[\Psi(\mathbf{W}(t^\diamond))] - \Psi^*(\mathbf{W}^*) \leq \frac{\delta}{2}\mathbb{E}[\|\mathbf{W}(t^\diamond) - \mathbf{W}^*\|_2^2]$. More details are provided in Appendix~\ref{appx:theorem4}.
\end{proof}

To this end, it is intractable to provide the exact convergence rate theoretically due to the dynamic set of learning SVs with various QoI at each learning round. Nonetheless, in the simulation results of Section~\ref{sec:PE}, i.e., Fig.~\ref{fig:accuracy_iid} and Fig.~\ref{fig:accuracy_niid}, we can show that the proposed FL algorithm with SV selection can speed up the convergence to the $\Psi^*(\mathbf{W}^*)$, where $\big[\mathbb{E}[\Psi(\mathbf{W}(t^\diamond))] - \Psi^*(\mathbf{W}^*)\big] = 0$, at a certain learning round $t^\diamond$.

\section{VSP's Profit Analysis Based on the Global Model Accuracy and Freshness}
\label{sec:VSP_A}

When more active SVs with high information significance are recruited for the FL execution, the converged global on-road model with high accuracy can be achieved faster. In this way, the VSP can obtain more up-to-date global on-road model, and thus achieve a higher freshness value when selling this model to the potential customers, such as SVs, public transportations, and mobile users. The adoption of information freshness in FL has been recently investigated in the literature. For example, the authors in~\cite{Yang:2020} propose an FL scheduling strategy considering the freshness of FL learners' local trained models, aiming at computing the staleness level of the local models for the global model update. In~\cite{Lim3:2020}, the authors present an incentive mechanism for the FL process accounting for the freshness of collected data at each FL learner towards the completion of FL process. Different from the existing works, we consider the freshness of produced global model at each learning round and analyze its impact on the VSP's profit. 

In practice, the profit of the VSP is influenced by not only the decision of contract optimization, but also the obtained global model accuracy and freshness. To the best of our knowledge, this is the first analysis in the literature accounting for the influence of global model accuracy and freshness towards the economic aspect of the typical FL process. 
Specifically, we introduce a dynamic global model economic value considering the product of global model accuracy and freshness at each learning round $t$, which is expressed by
\begin{equation}
\label{eqn10}
\begin{aligned}
\omega(t) = a\chi(t) e^{-bt},
\end{aligned}
\end{equation}
where $a,b$ are control parameters which can be defined in advance based on the accuracy and freshness of global model, $\chi(t)$ is the global prediction model accuracy at learning round $t$, and $e^{-bt}$ is the global model freshness function with respect to the learning round $t$~\cite{Chen10:2020,Abdellatif:2020}. The relationship between the global model accuracy and freshness in (\ref{eqn10}) implies that there is a trade-off between them in the FL process. Specifically, to obtain a global model with high accuracy, the FL algorithm usually requires to be executed for multiple learning rounds, that leads to a lower freshness value, and vice versa. Furthermore, the proposed dynamic global model economic value can be used to control the VSP's net profit effectively. Specifically, if the $\omega(t)$ is higher, the satisfaction in (\ref{eqn:for30a}) will increase and vice versa. Thus, the net profit of the VSP and social welfare for type $\theta_j$ after the current learning round is completed can be respectively defined as follows:
\begin{equation}
\label{eqn11}
\begin{aligned}
{\hat \mu}^j_{VSP}(t) = \theta_j\omega(t)\lambda\sqrt{\sum_{n \in \mathcal{N}(t)}{\hat \varrho}_j^n(t)\zeta_j^n(t)} - \sum_{n \in \mathcal{N}(t)}{\hat \varrho}_j^n(t)\varphi_j^n(t),
\end{aligned}
\end{equation}
and
\begin{equation}
\label{eqn11b}
\begin{aligned}
{\hat \mu}^j_{SW}(t) &= \theta_j\omega(t)\lambda\sqrt{\sum_{n \in \mathcal{N}(t)}{\hat \varrho}_j^n(t)\zeta_j^n(t)} \\&- \sum_{n \in \mathcal{N}(t)}{\hat \varrho}_j^n(t)\zeta_j^n(t)\xi_n(t), \forall j \in \mathcal{J}.
\end{aligned}
\end{equation}

From (\ref{eqn11}), we can observe the relationship between the net profit of the VSP and the global model economic value (through the global model accuracy and freshness values). In particular, the earlier the VSP can obtain the high-accurate global prediction model, the more freshness and thus the higher economic value the global model can achieve (as illustrated in Fig.~\ref{fig:acc_freshness}). However, in practice, at the early learning rounds, the accuracy of the global prediction model is not really high due to insufficient number of trained samples in the FL process (as illustrated in Fig.~\ref{fig:acc_freshness}(c))~\cite{Lim:2020}. Moreover, as shown in Fig.~\ref{fig:acc_freshness}(d), the net profit of the VSP will first increase up to the first 15 learning rounds due to high freshness value and increasing accuracy of the global model. Although the global model accuracy gets improved for the rest of learning rounds, the existence of lower freshness value will trigger lower global model economic value, and thus the net profit of the VSP also decreases gradually. Therefore, this observation provides an insightful information for the VSP to determine the best global model economic value which improves its profit in the IoV accounting for the trade-off between the global model accuracy and freshness for each learning round.

\begin{figure}[h]
	\begin{center}
		$\begin{array}{cc} 
		\epsfxsize=1.61 in \epsffile{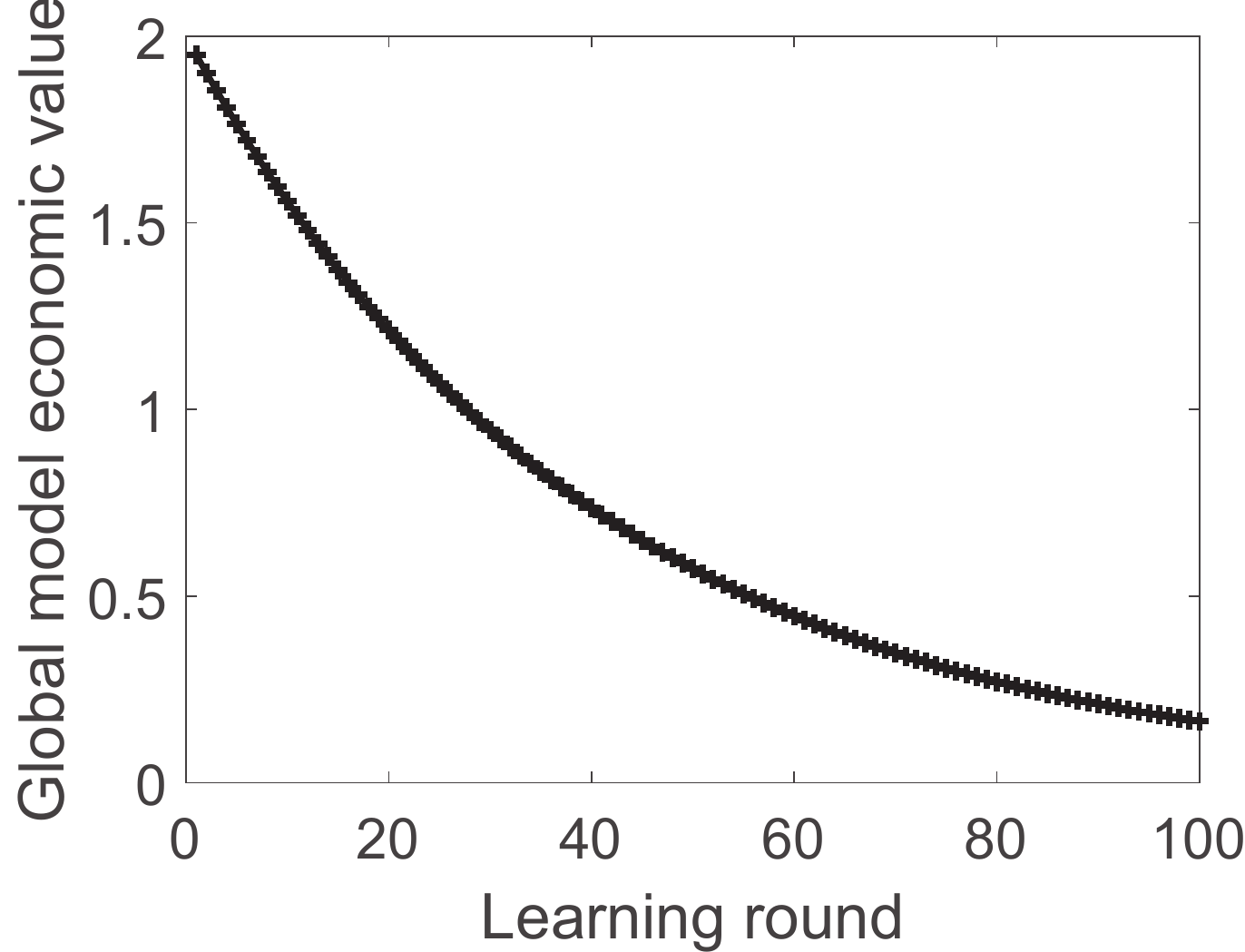} & 
		\hspace*{-.1cm}
		\epsfxsize=1.6 in \epsffile{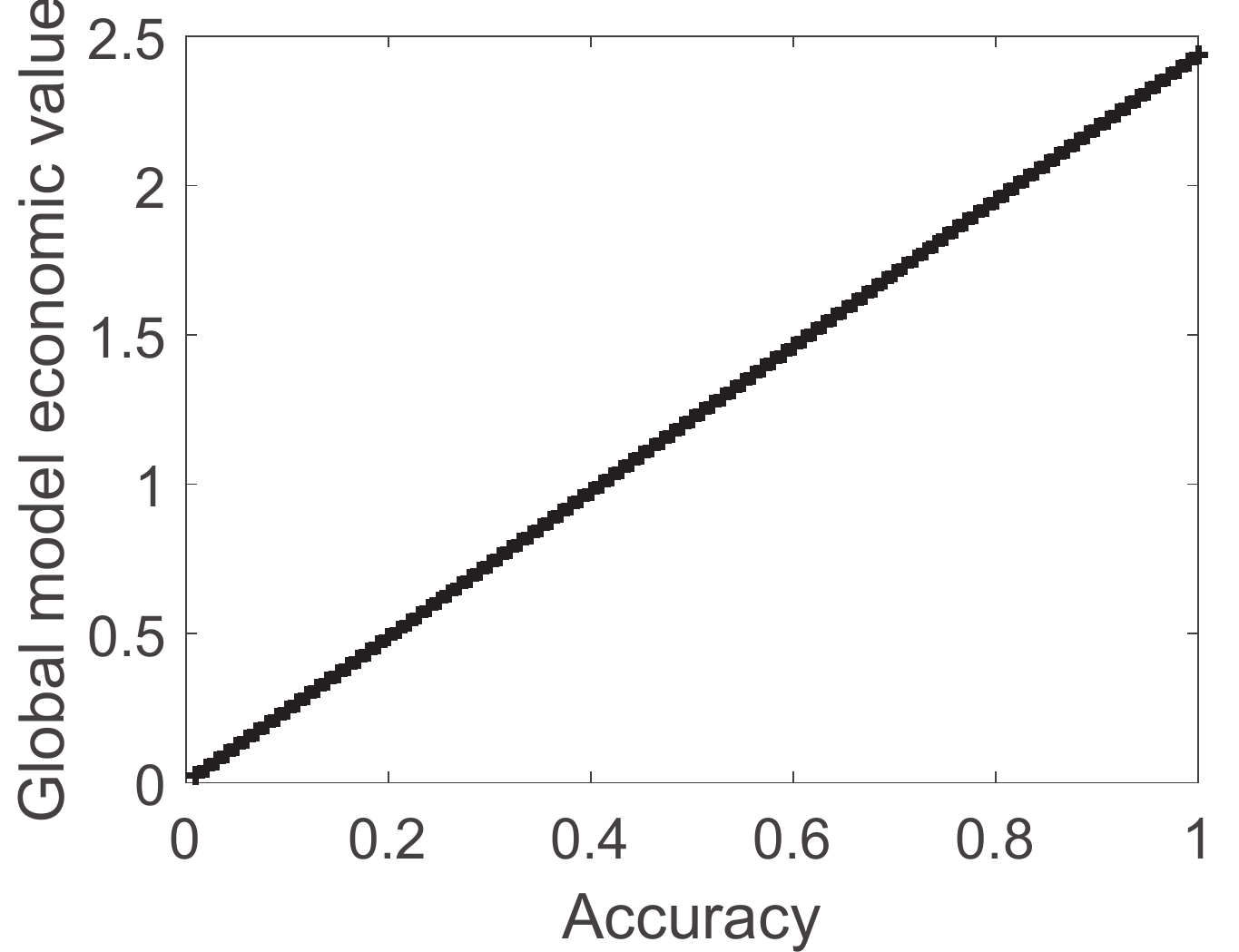} \\ [0.1cm]
		\text{\footnotesize (a) Different learning rounds} & \text{\footnotesize (b) Different accuracy values} \\ [0.2cm]
		\vspace*{-0cm}
		\epsfxsize=1.6 in \epsffile{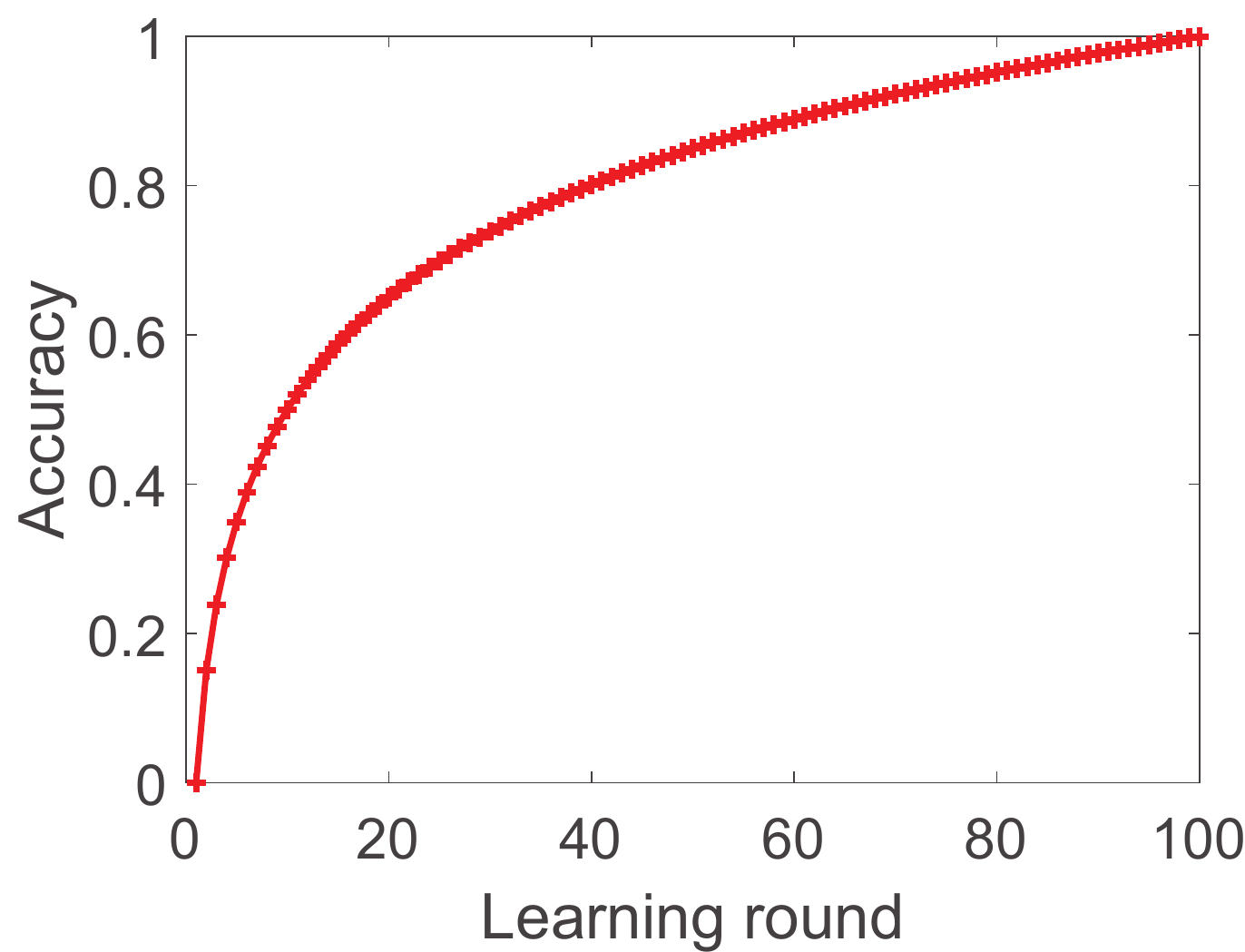} & 
		\hspace*{-.1cm}
		\epsfxsize=1.67 in \epsffile{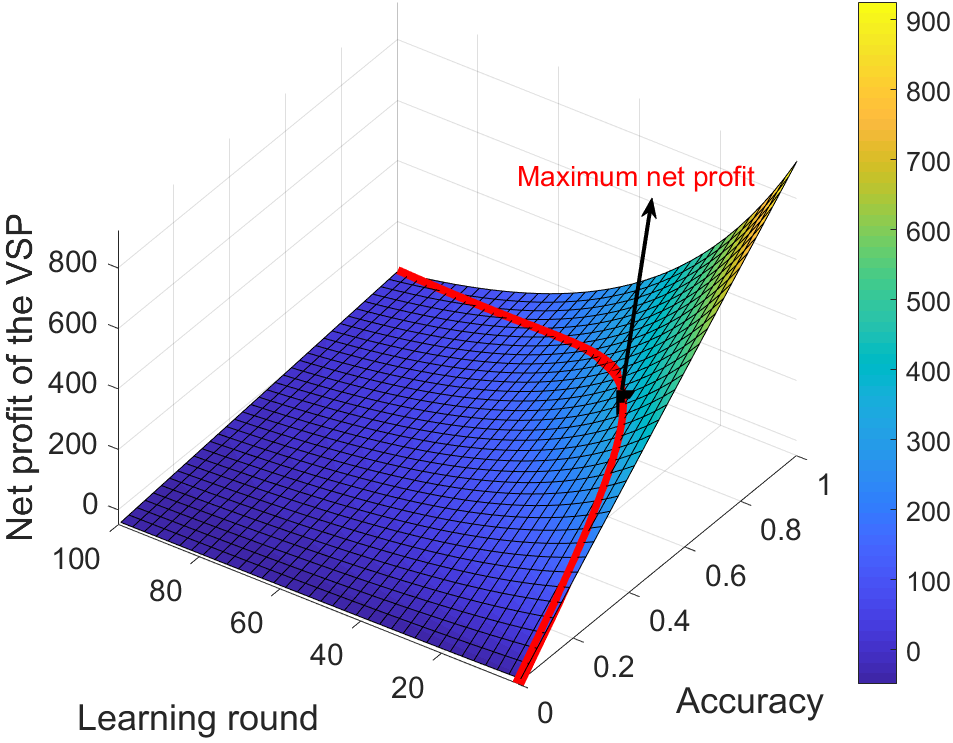} \\ [0.1cm]
		\text{\footnotesize (c) Accuracy vs learning rounds} & \text{\footnotesize (d) Net profit of the VSP} \\ [0.2cm]
		\end{array}$
		\vspace{-0.5\baselineskip}
		\caption{The impact of global model accuracy and freshness (represented by learning rounds) on the global model economic value and the VSP's net profit.}
		\label{fig:acc_freshness}
	\end{center}
\end{figure} 

\section{Performance Evaluation} 
\label{sec:PE}

\subsection{Dataset Preparation} 
\label{subsec:DP}

We carry our experiments using the actual two road traffic datasets in the UK between 2000 and 2016~\cite{UK_Accident:2017} to evaluate the performance of the proposed economic framework. For the first dataset, there are 275K traffic AADFs from all major roads of 190 areas/local districts in the UK that will be applied to obtain the set of significant areas. From this dataset, we divide AADFs from all major roads based on the area IDs such that each area is represented by the total AADF $V_d$ as defined in (\ref{eqn:for15a}). For the second dataset, it includes  1.5M traffic accidents of all major roads in the UK along with 34 features. Based on these features, we extract 7 useful features including accident location with 416 location IDs, accident date, accident time, accident additional conditions, i.e., light, weather, and road surface conditions (as \emph{training features}) and accident severity (as \emph{training label}). We also convert the accident date and time into categorical features, i.e., 7-day (i.e., $1,2,\ldots,7$) and 24-hour (i.e., $0,2,\ldots,23$) categories, respectively. 

\subsection{Experiment Setup} 
\label{subsec:ES}

We use \emph{TensorFlow CPU 2.2} containing \emph{TensorFlow Federated} in a shared cluster. 
We also consider 100 active SVs. i.e., $I=100$, to compare the performance of our proposed method, i.e., proposed-LIS, with other baseline FL methods, i.e., random scheduling and round-robin scheduling~\cite{Yang:2019}. For the random scheduling, the VSP selects $N$ SVs from $I$ active SVs randomly. For the round-robin scheduling, the VSP chooses $N$ SVs from $I$ active SVs sequentially in a round-robin manner. Both methods do not consider the SVs' location and information significance at each round. We also present a location-significance scheduling (referred to as proposed-LS), i.e., the VSP selects $N$ SVs from $M$ active SVs randomly without considering the SVs' information significance at each round, to show the impact of selected SVs' current location areas on the performance results. To this end, we do not consider an information-significance scheduling because this method may provide inaccurate and useless on-road information regardless its QoI. 
We then divide the active SVs into 3 categories, i.e., high-QoI SVs, i.e., \emph{h-SVs}, medium-QoI SVs, i.e., \emph{m-SVs}, and low-QoI SVs, i.e., \emph{l-SVs}. In particular, we generate their data sizes and spatio-temporal variabilities in the following order: \emph{h-SVs} $>$ \emph{m-SVs} $>$ \emph{l-SVs}. At each learning round, all SVs are randomly scattered in both significant and insignificant areas (which are determined by the K-means algorithm in Section.~\ref{subsec:LASV}) as illustrated in Fig.~\ref{k_means_result}.
We also account for various values of $N$, i.e., $5 \text{ and } 10$ SVs, for each round to evaluate the proposed framework efficiency with respect to various payment budgets.

In the contract optimization, we consider one agent, i.e., the VSP, and $N$ principals corresponding to $N$ learning SVs. We set $\lambda$ at 1.2 monetary unit (MU) per $\zeta^n(t)=0.1, \forall n \in \mathcal{N}(t), \forall t \in \mathcal{T}$~\cite{Xu:2015}. We also set $\upsilon$ and $\xi_n(t)$ at 2.1MU and 0.5MU per $\zeta^n(t)=0.1, \forall n \in \mathcal{N}(t), \forall t \in \mathcal{T}$, respectively. To demonstrate various performances of the VSP, we consider 10 types, i.e., $J=10$, with the same distribution of the types, i.e., $\rho_j(t)=0.1$ at each round. Then, we set $B_{max}(t)$ at 125MU and 250MU~\cite{Zhao:2016} at each round for the cases of 5 and 10 learning SVs, respectively.

Next, we split the accident dataset into 0.8 training set and 0.2 testing set. Additionally, we consider two data distribution scenarios, i.e., i.i.d and non-i.i.d scenarios. For i.i.d scenario, it occurs when most of the learning SVs in the IoV network have visited many locations at different times frequently, e.g., taxi or long-trip public transport, and thus they can collect on-road data with various accident severity levels. In this case, we distribute the training set randomly to the \emph{h-SVs}, \emph{m-SVs}, and \emph{l-SVs} based on their information significance. Meanwhile, for non-i.i.d scenario, it takes place when most of the learning SVs have only passed through specific limited locations most of the times, e.g., the school bus or city public transport, such that each learning SV has very-limited and fixed accident severity information. Specifically, we first divide the training set randomly into 3 training subsets for \emph{h-SVs}, \emph{m-SVs}, and \emph{l-SVs}. Then, we sort each training subset according to the training label. Each sorted training subset is distributed to the corresponding number of learning SVs based on their information significance. To implement the DNN, we use two hidden layers with 128 and 64 neurons per layer and one final layer with 3 neurons (which corresponding to 3 labels). We also use the Adam optimizer with initial step size 0.01.

\begin{figure}[!t]
	\centering
	\includegraphics[scale=0.55]{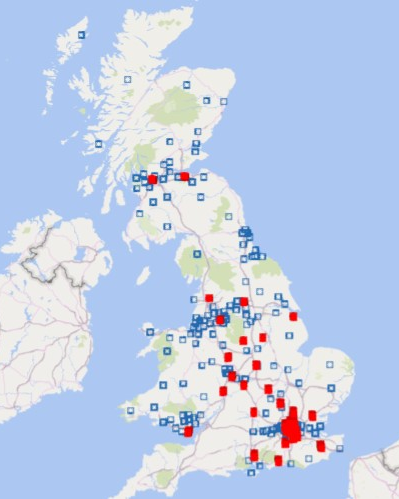}
	\caption{An illustration of significant (red) and insignificant (blue) areas in the UK from the dataset in~\cite{UK_Accident:2017}.}
	\vspace{-0em}
	\label{k_means_result}
\end{figure}
\subsection{Learning Contract Performance}

\subsubsection{Common constraint validity and profit of the VSP} 

\begin{figure}[!]
	\begin{center}
		$\begin{array}{cc} 
		\epsfxsize=1.65 in \epsffile{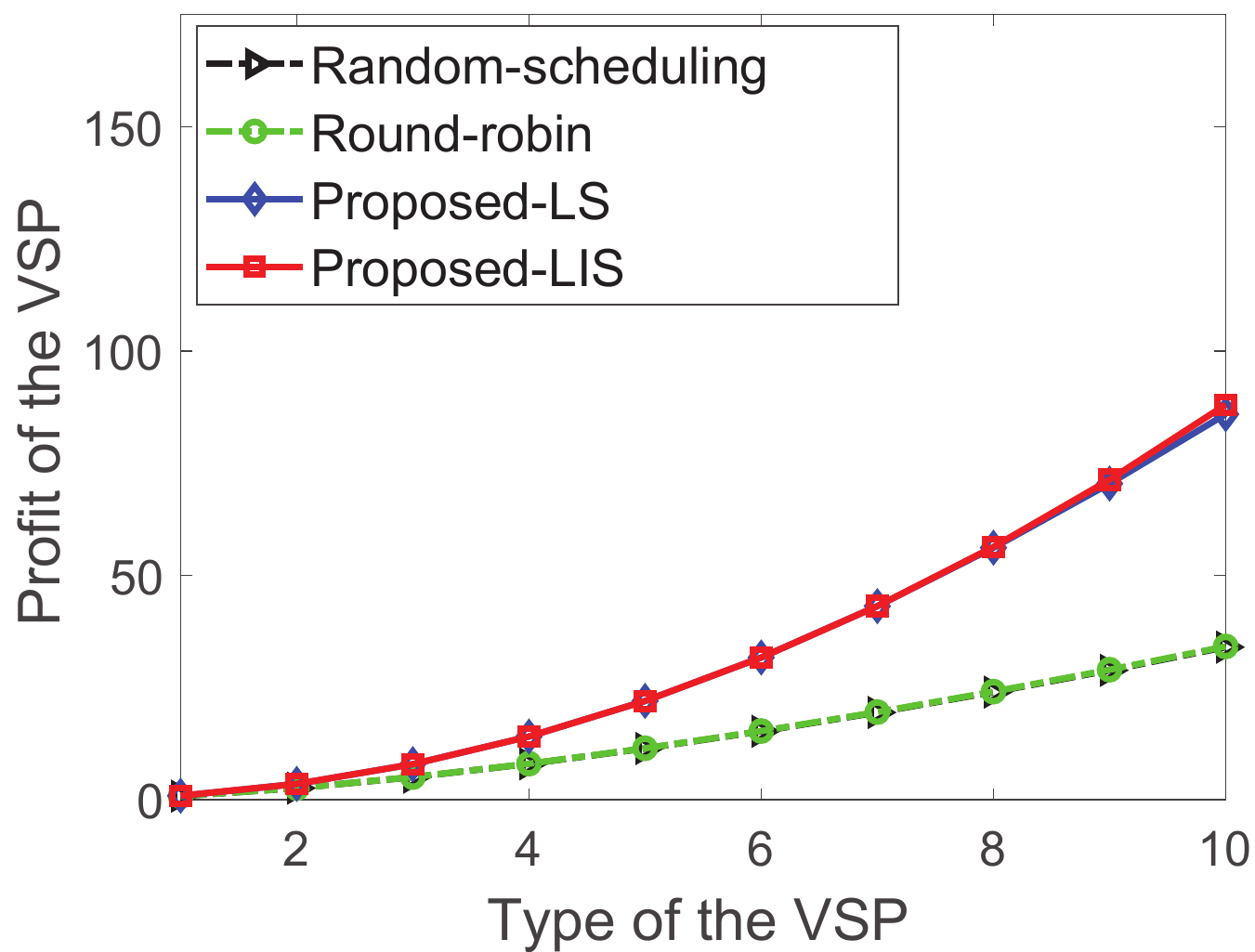} &
		\hspace*{-.3cm}
		\epsfxsize=1.65 in \epsffile{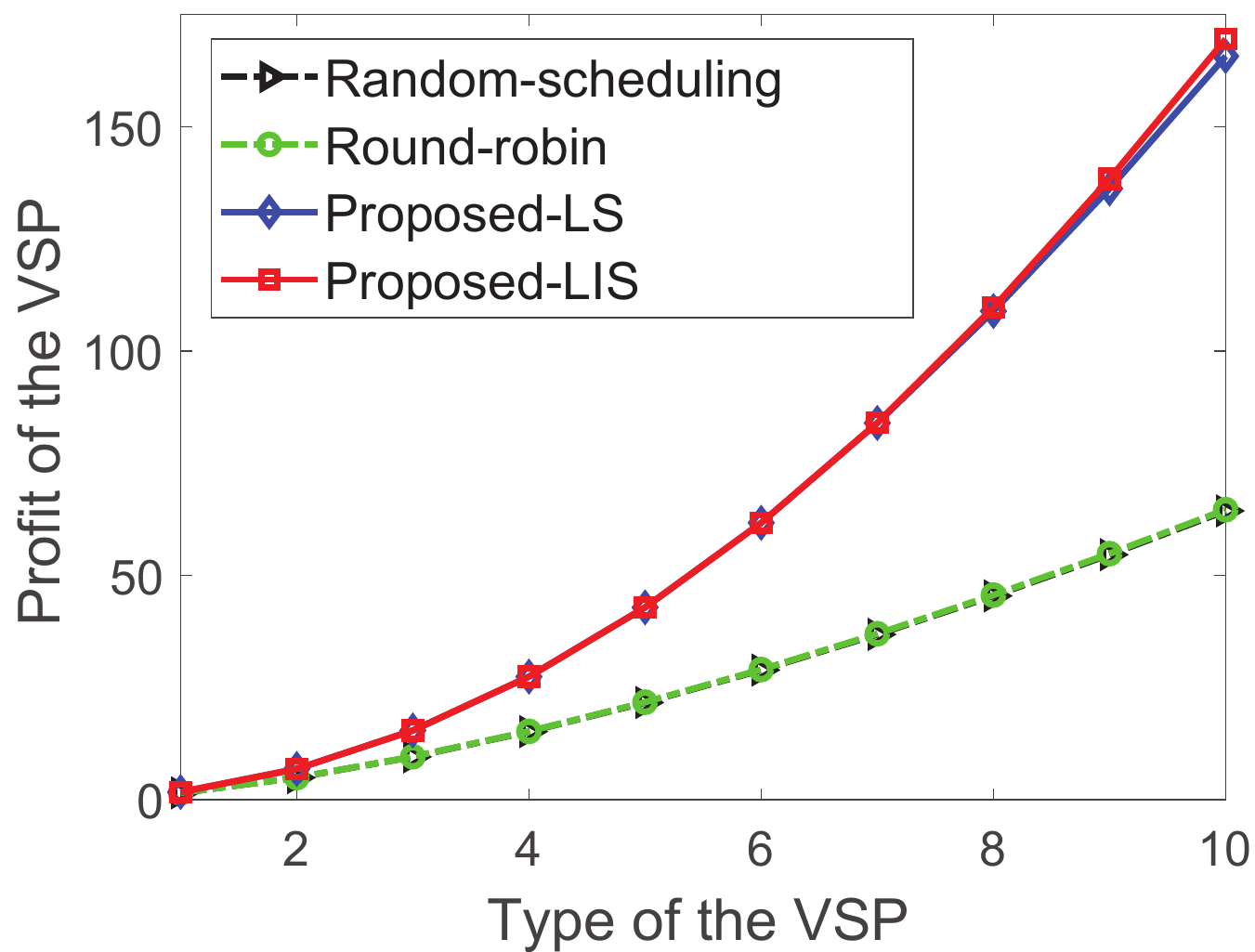} \\ [-0.1cm]
		\text{\footnotesize (a) IR validity of 5 SVs} & \text{\footnotesize (b) IR validity of 10 SVs} \\ [0.2cm]
		\vspace*{-0cm}
		\epsfxsize=1.65 in \epsffile{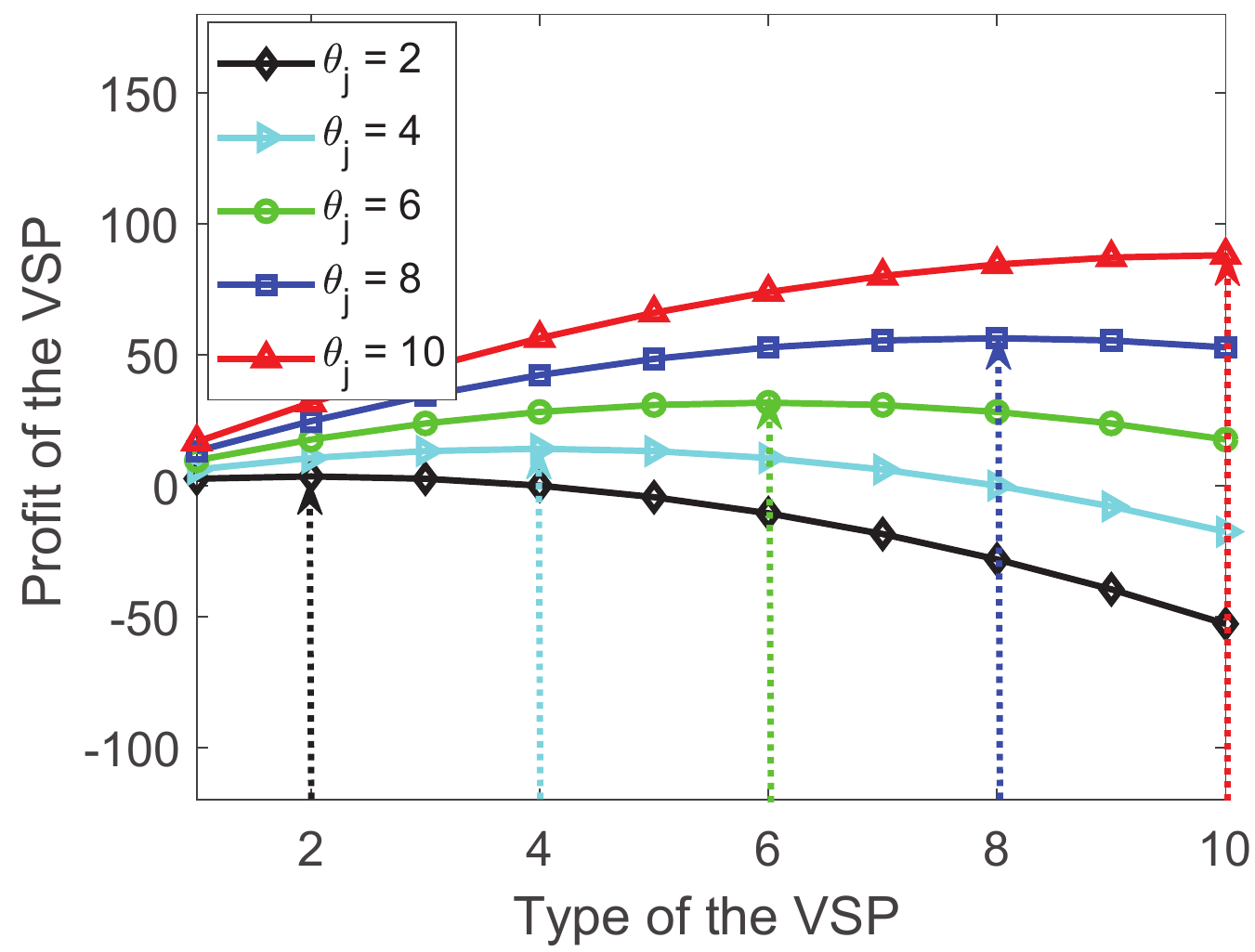} &
		\hspace*{-.3cm}
		\epsfxsize=1.65 in \epsffile{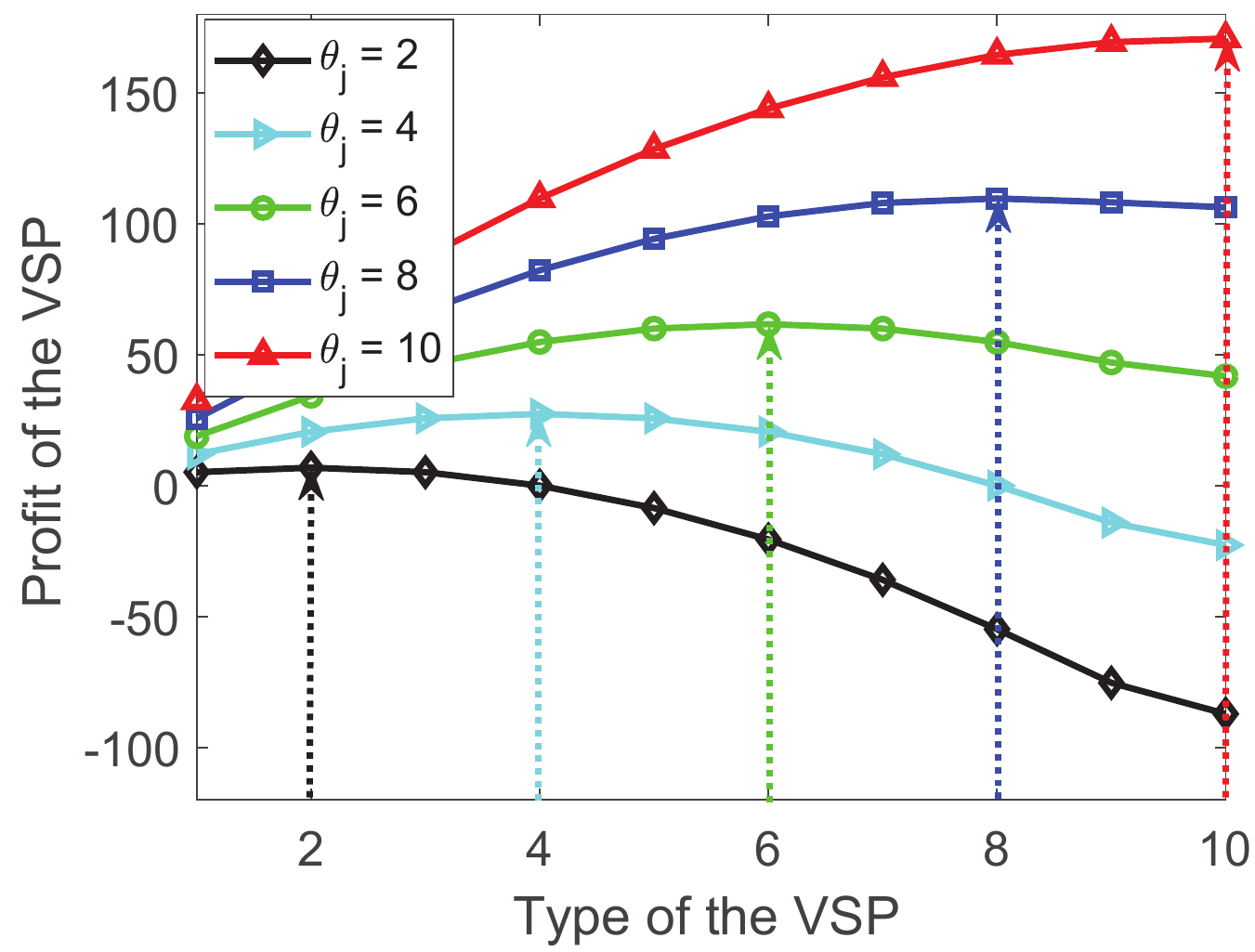} \\ [-0.1cm]
		\text{\footnotesize (c) IC validity of 5 SVs} & \text{\footnotesize (d) IC validity of 10 SVs} \\ [0.2cm]
		\end{array}$
		\vspace*{-0.3cm}
		\caption{The validity of IR and IC constraints for i.i.d scenario.}
		\vspace{-0.5em}
		\label{fig:IR_IC_iid}
	\end{center}
\end{figure}

\begin{figure}[!]
	\begin{center}
		$\begin{array}{cc} 
		\epsfxsize=1.65 in \epsffile{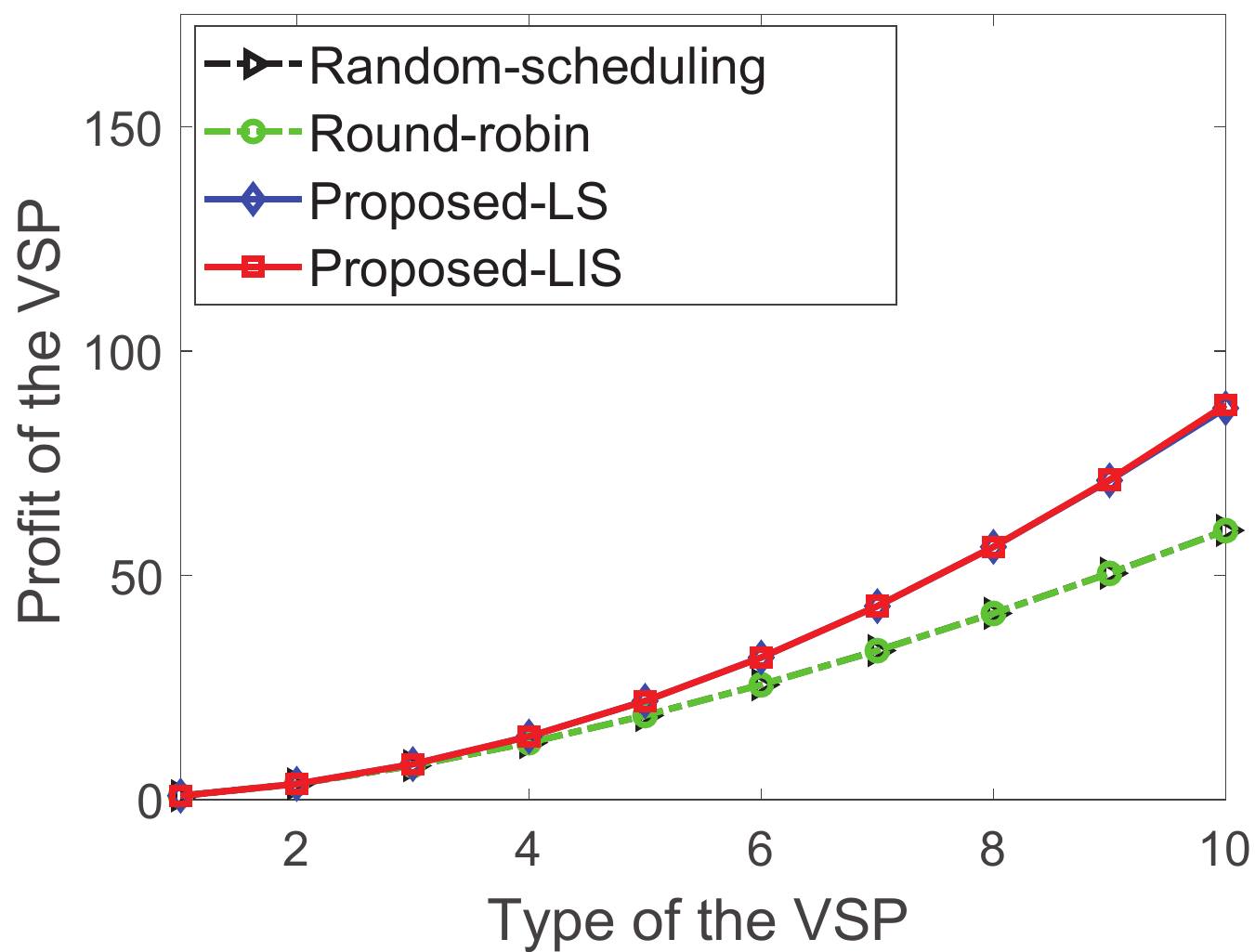} &
		\hspace*{-.3cm}
		\epsfxsize=1.65 in \epsffile{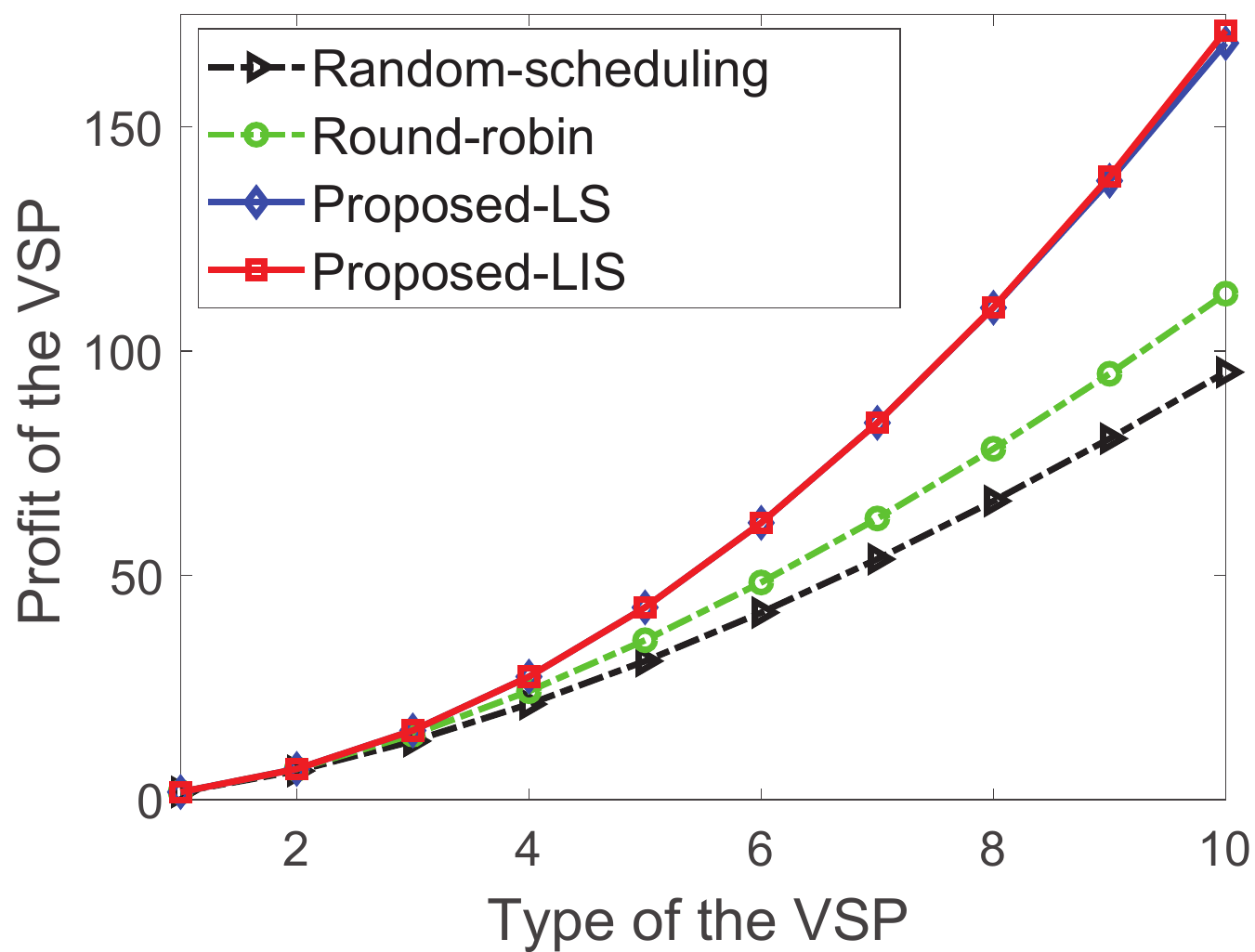} \\ [-0.1cm]
		\text{\footnotesize (a) IR validity of 5 SVs} & \text{\footnotesize (b) IR validity of 10 SVs} \\ [0.2cm]
		\vspace*{-0cm}
		\epsfxsize=1.65 in \epsffile{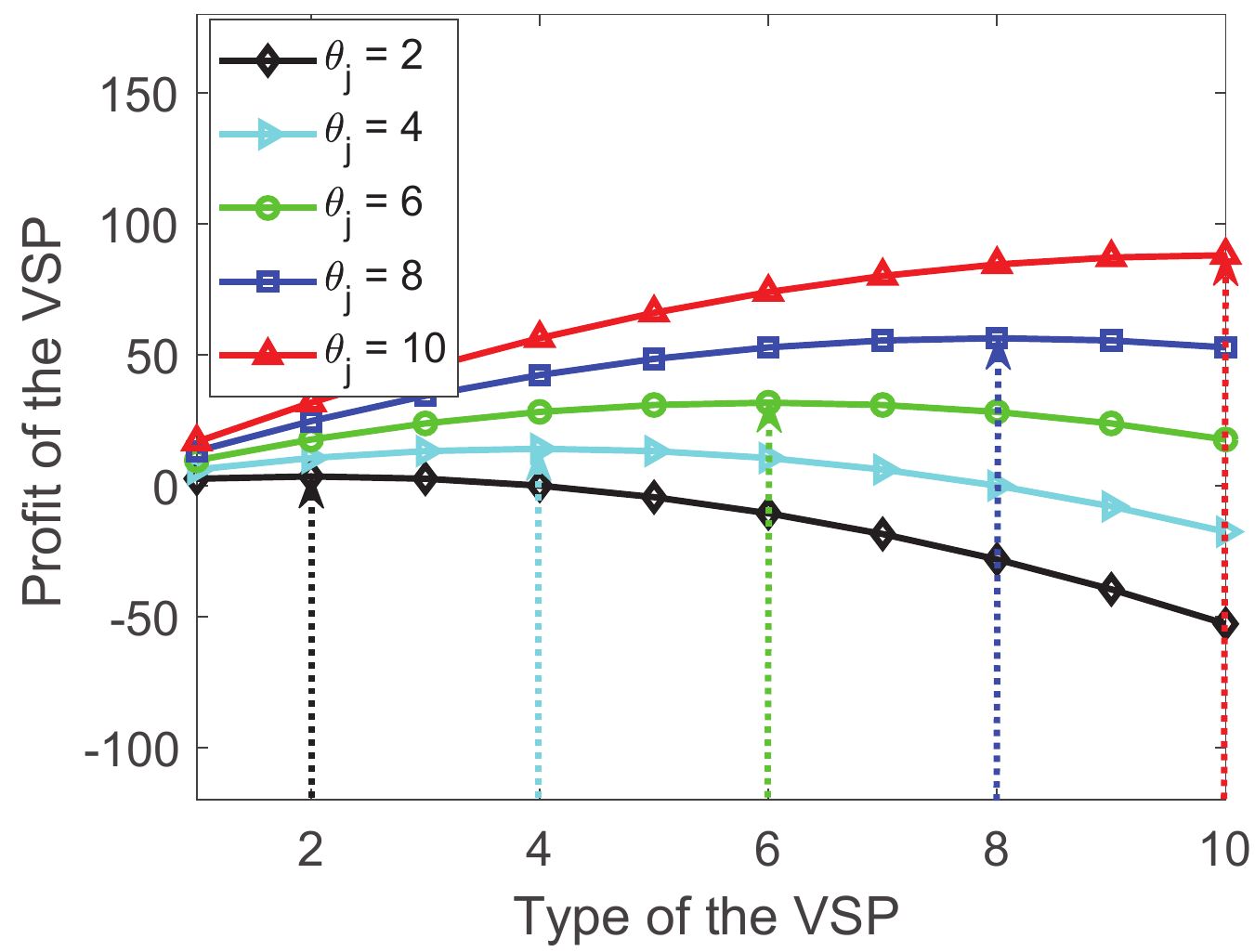} &
		\hspace*{-.3cm}
		\epsfxsize=1.65 in \epsffile{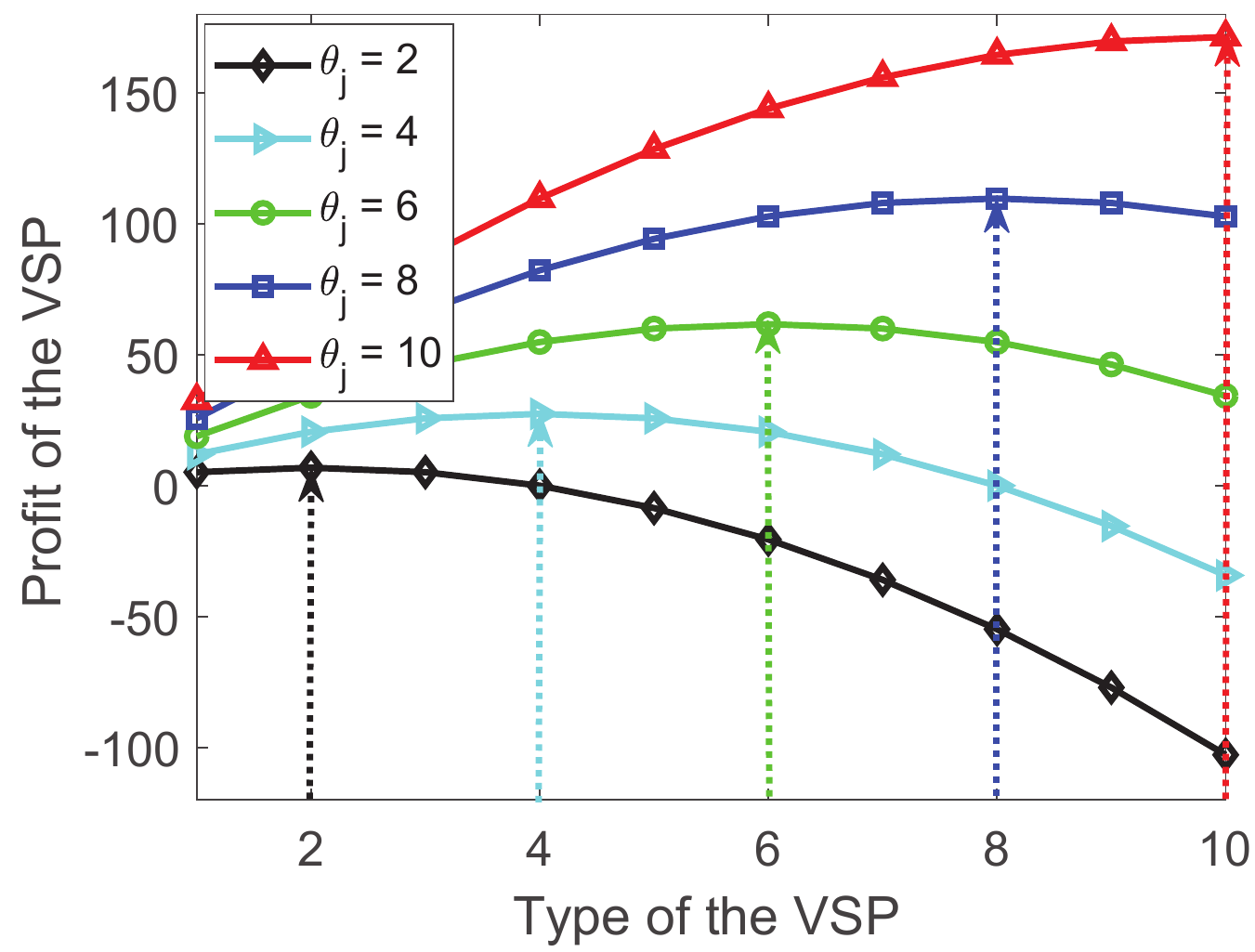} \\ [-0.1cm]
		\text{\footnotesize (c) IC validity of 5 SVs} & \text{\footnotesize (d) IC validity of 10 SVs} \\ [0.2cm]
		\end{array}$
		\vspace*{-0.3cm}
		\caption{The validity of IR and IC constraints for non-i.i.d scenario.}
		\vspace{-0.5em}
		\label{fig:IR_IC_niid}
	\end{center}
\end{figure}

We first evaluate the validity of the VSP's common constraints, i.e., IR and IC constraints, considered in the contract optimization at a specific learning round, e.g., the first round in this case. Specifically, we can observe the IR constraints of the VSP for all methods in Fig.~\ref{fig:IR_IC_iid}(a)-(b) and Fig.~\ref{fig:IR_IC_niid}(a)-(b) for i.i.d and non-i.i.d scenarios, respectively, with various numbers of learning SVs. In this case, the profits of proposed methods coincide with each other due to the VSP's limited budget constraint for the VSP's each type. From these figures, the VSP can always obtain positive profits for all types of the VSP, thereby ensuring the feasibility of the IR constraints as shown in~(\ref{eqn:for7a}). These profits monotonically increase with respect to the VSP's type such that the VSP with a higher type will obtain a higher profit. The reason is that the VSP is willing to pay more money to the learning SVs due to higher payment budget. Additionally, we can observe that more learning SVs, i.e., 10 learning SVs, will produce higher profits for the VSP up to 1.95 times, respectively, compared with that of 5 learning SVs. This is because more learning SVs can provide larger training data sizes and better information significance values for the FL process. Meanwhile, compared with other methods, the proposed-LIS can achieve higher VSP's profit up to 2.63 and 1.8 times for i.i.d and non-i.i.d scenarios, respectively.
For Fig.~\ref{fig:IR_IC_iid}(c)-(d) and Fig.~\ref{fig:IR_IC_niid}(c)-(d), we observe that the VSP can always obtain the highest profit when it utilizes the right contract for its true type. In this way, the IC constraints for all types of the VSP are also satisfied. Particularly, the VSP with types 2, 4, 6, 8, and 10 will produce the highest profit when the suitable contracts for those types are used. As both IR and IC constraints hold, we can find the feasible contracts for all learning SVs to maximize their profits at each round.

\subsubsection{The profits of SVs and social welfare of the IoV}

\begin{figure}[!]
	\begin{center}
		$\begin{array}{cc} 
		\epsfxsize=1.65 in \epsffile{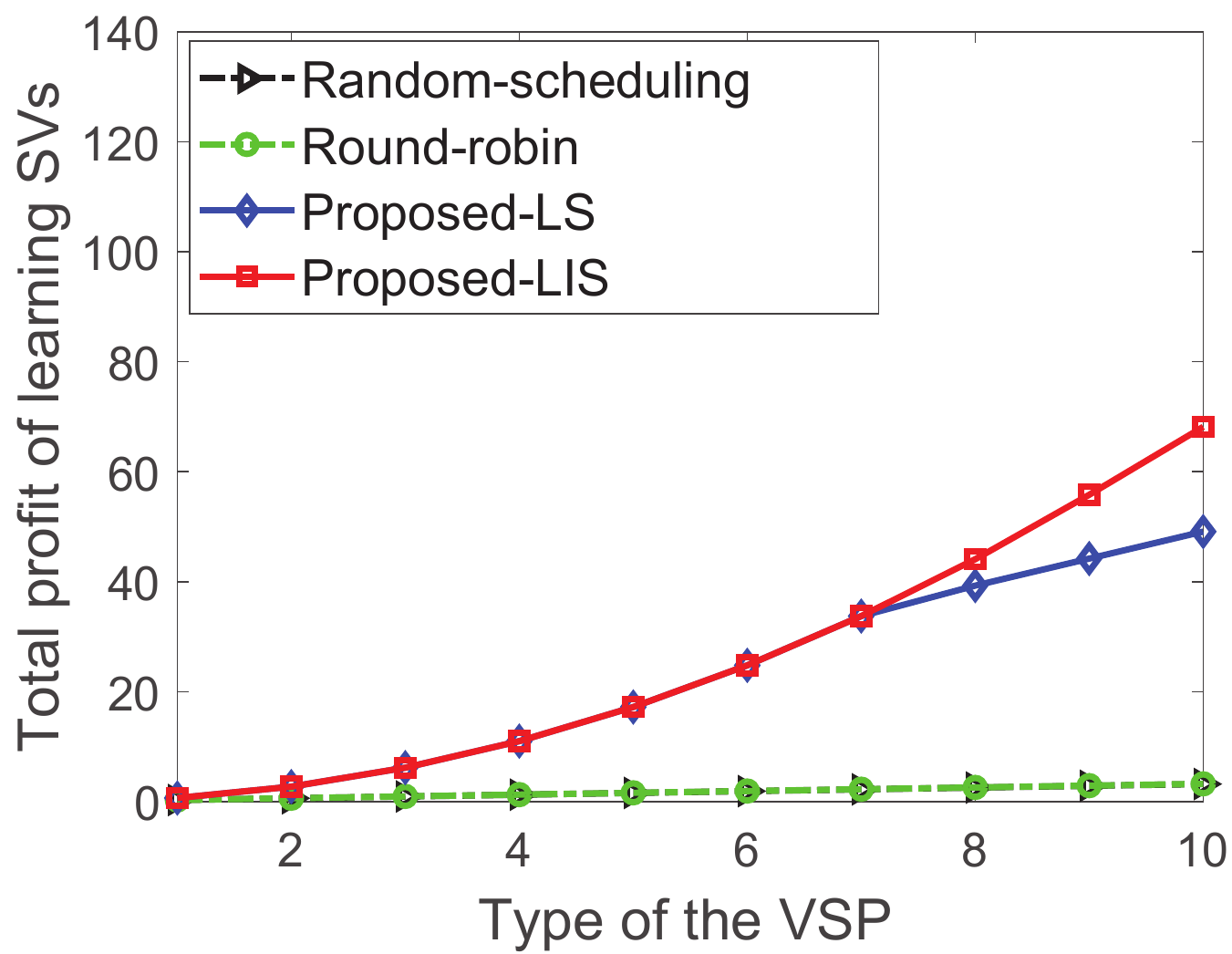} &
		\hspace*{-.3cm}
		\epsfxsize=1.65 in \epsffile{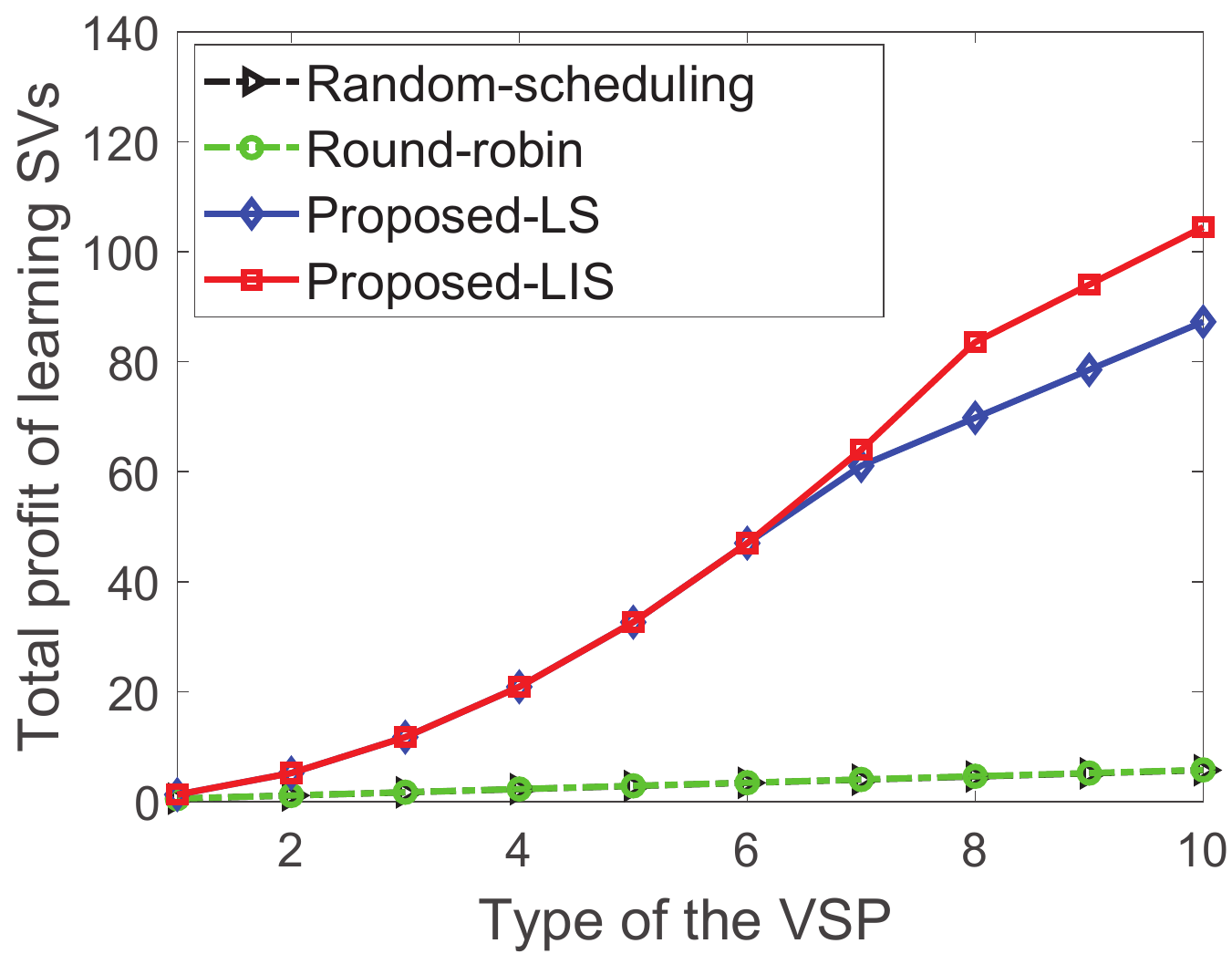} \\ [-0.1cm]
		\text{\footnotesize (a) Total profit of 5 SVs} & \text{\footnotesize (b) Total profit of 10 SVs} \\ [0.2cm]
		\vspace*{-0cm}
		\epsfxsize=1.65 in \epsffile{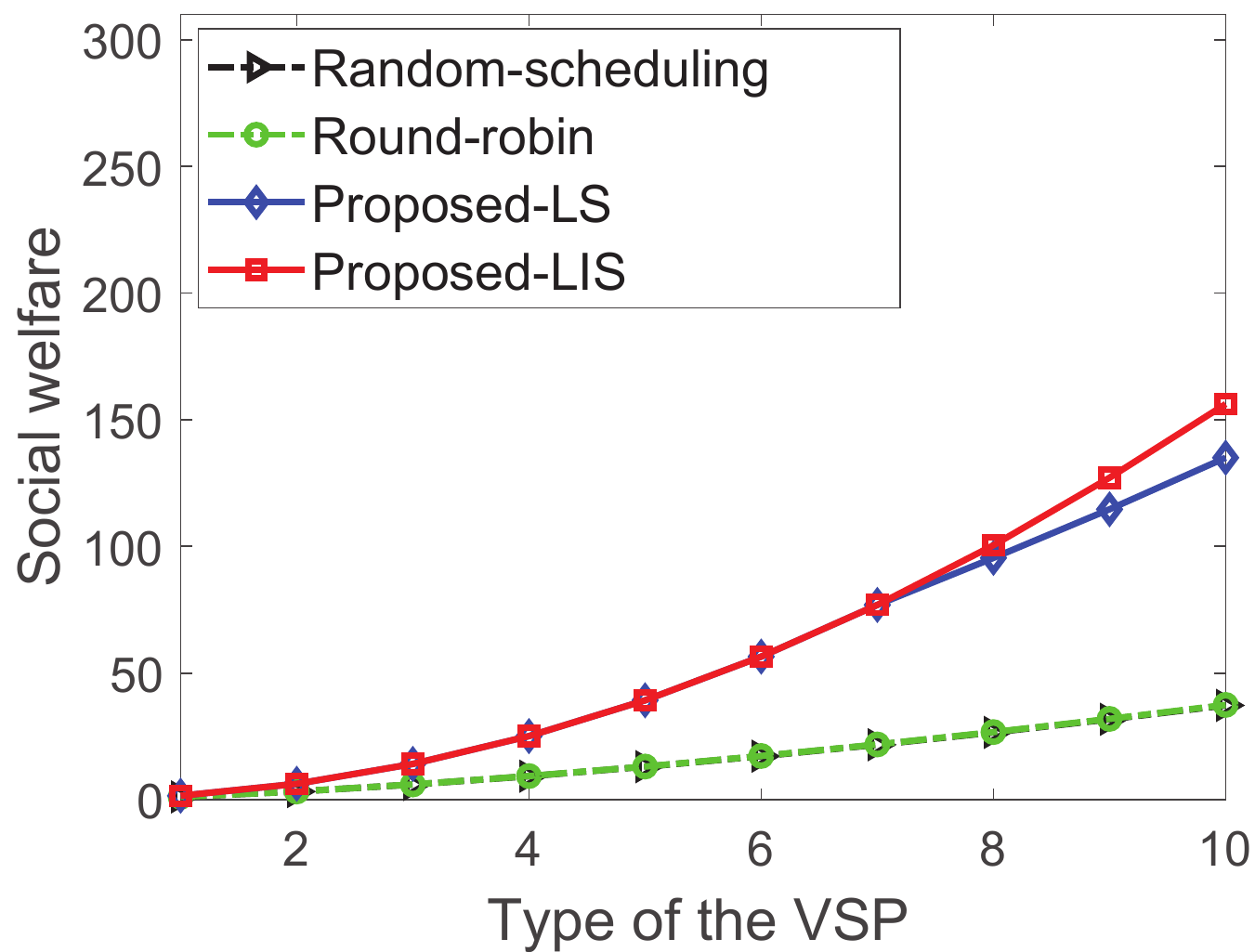} &
		\hspace*{-.3cm}
		\epsfxsize=1.65 in \epsffile{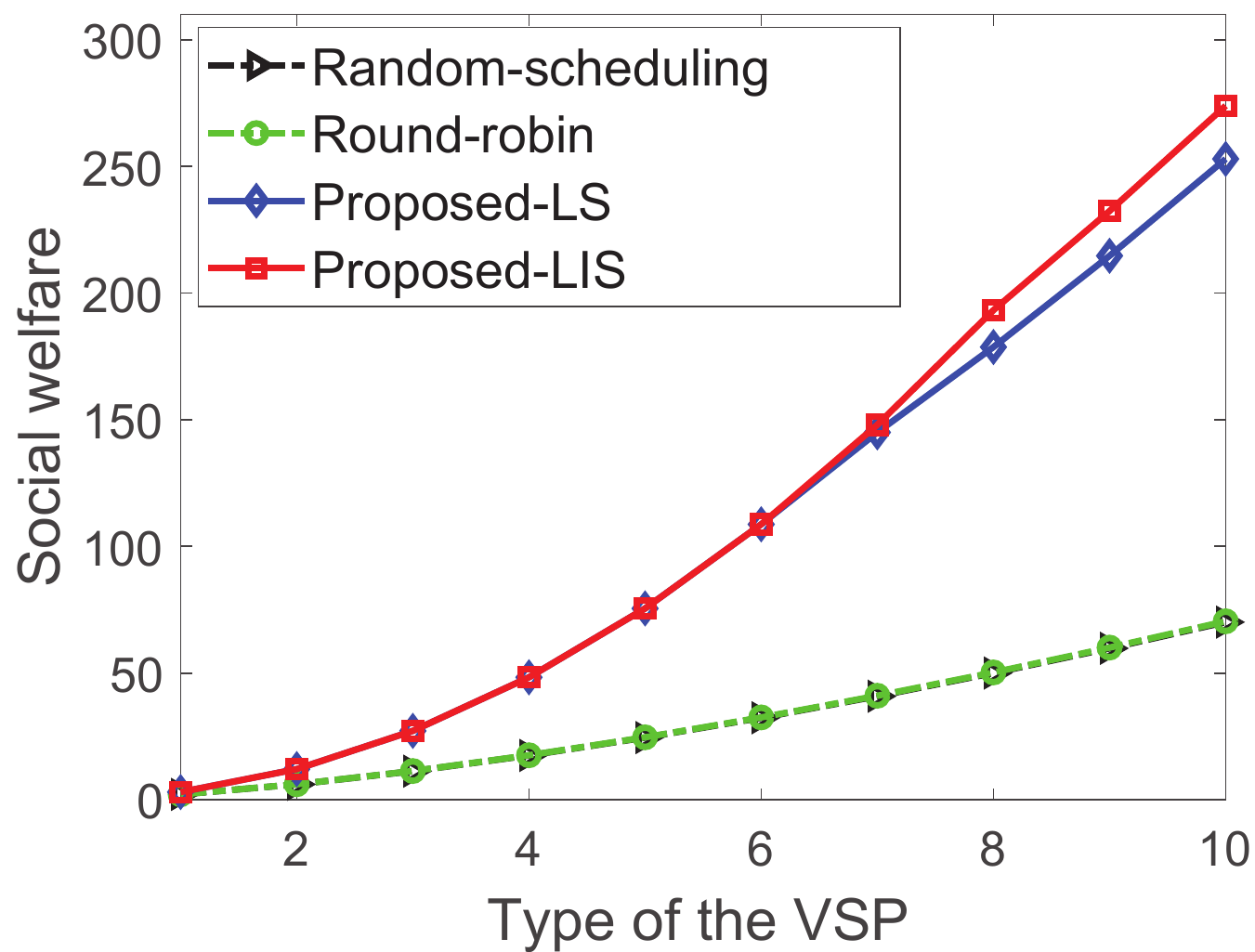} \\ [-0.1cm]
		\text{\footnotesize (c) Social welfare with 5 SVs} & \text{\footnotesize (d) Social welfare with 10 SVs} \\ [0.2cm]
		\end{array}$
		\vspace*{-0.3cm}
		\caption{The total profit of learning SVs and social welfare for i.i.d scenario.}
		\vspace{-0.5em}
		\label{fig:SV_profit_iid}
	\end{center}
\end{figure}

\begin{figure}[!]
	\begin{center}
		$\begin{array}{cc} 
		\epsfxsize=1.65 in \epsffile{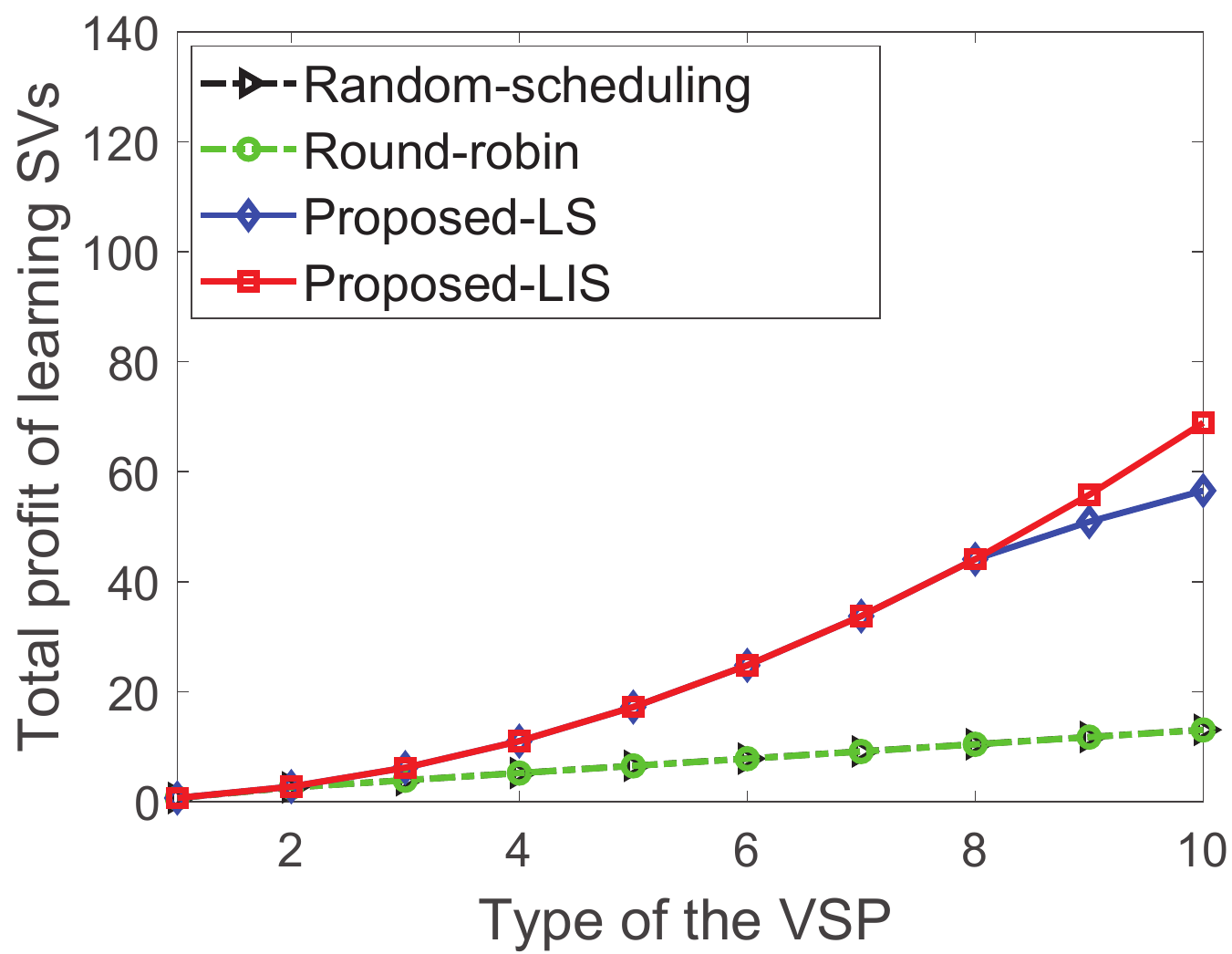} &
		\hspace*{-.3cm}
		\epsfxsize=1.65 in \epsffile{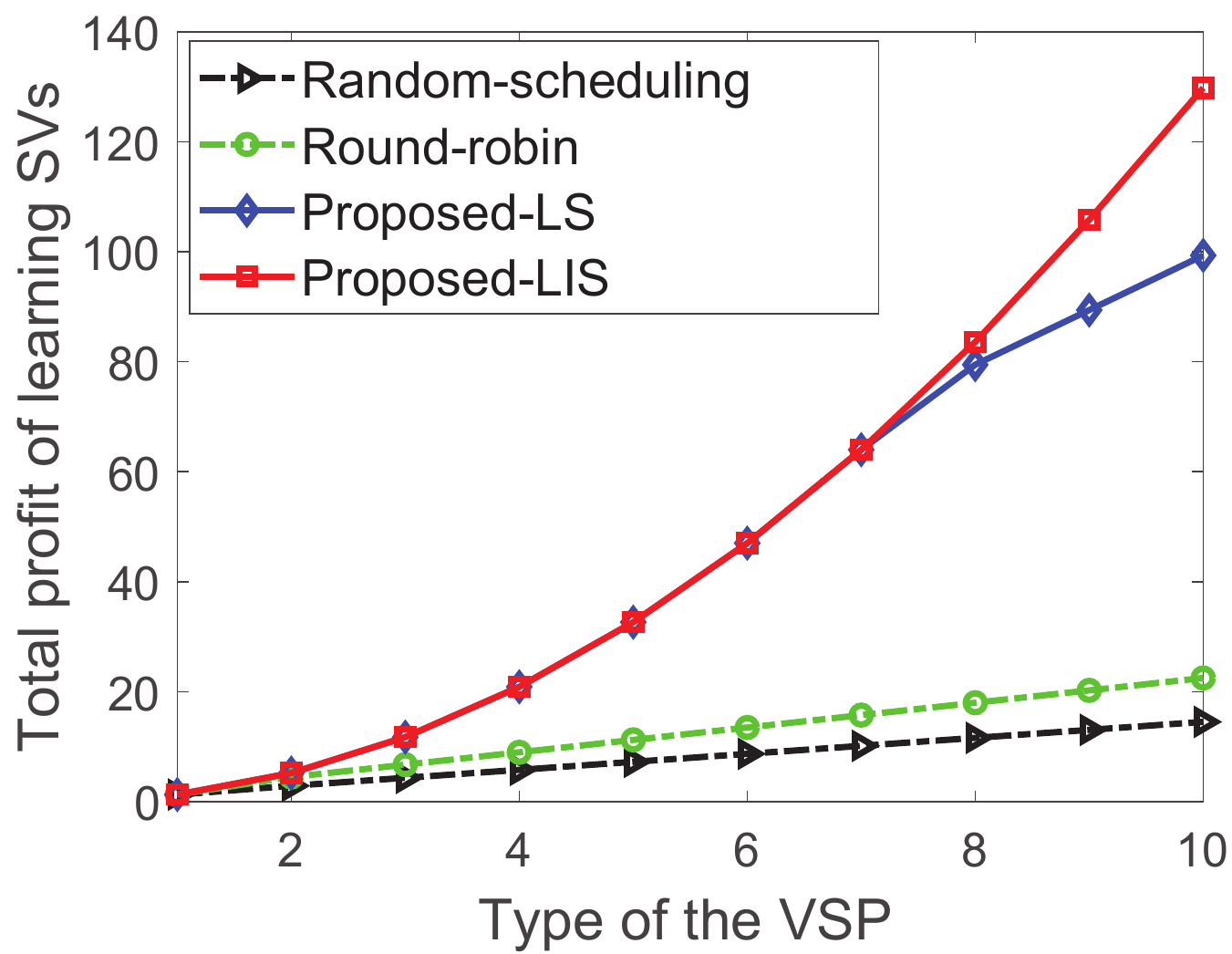} \\ [-0.1cm]
		\text{\footnotesize (a) Total profit of 5 SVs} & \text{\footnotesize (b) Total profit of 10 SVs} \\ [0.2cm]
		\vspace*{-0cm}
		\epsfxsize=1.65 in \epsffile{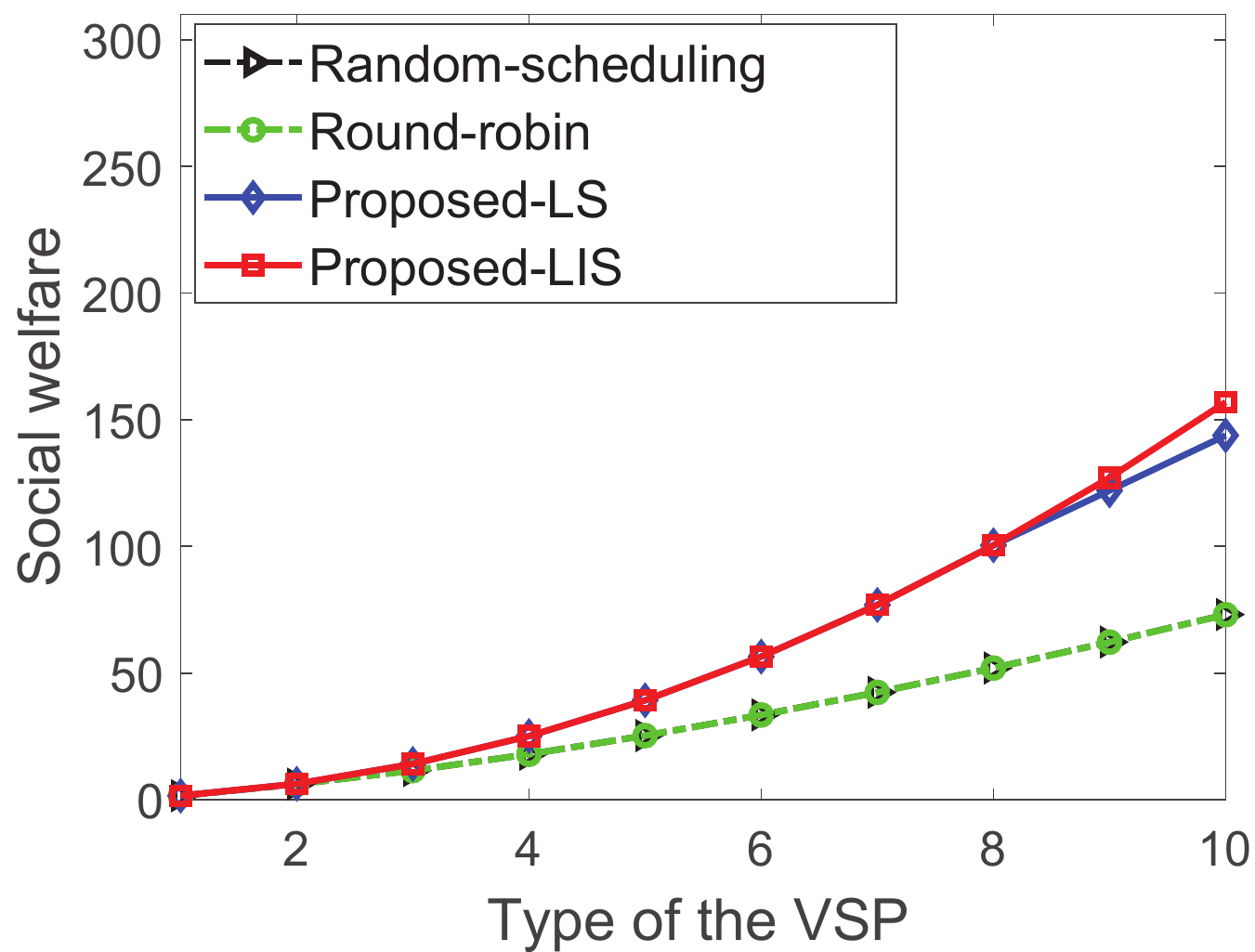} &
		\hspace*{-.3cm}
		\epsfxsize=1.65 in \epsffile{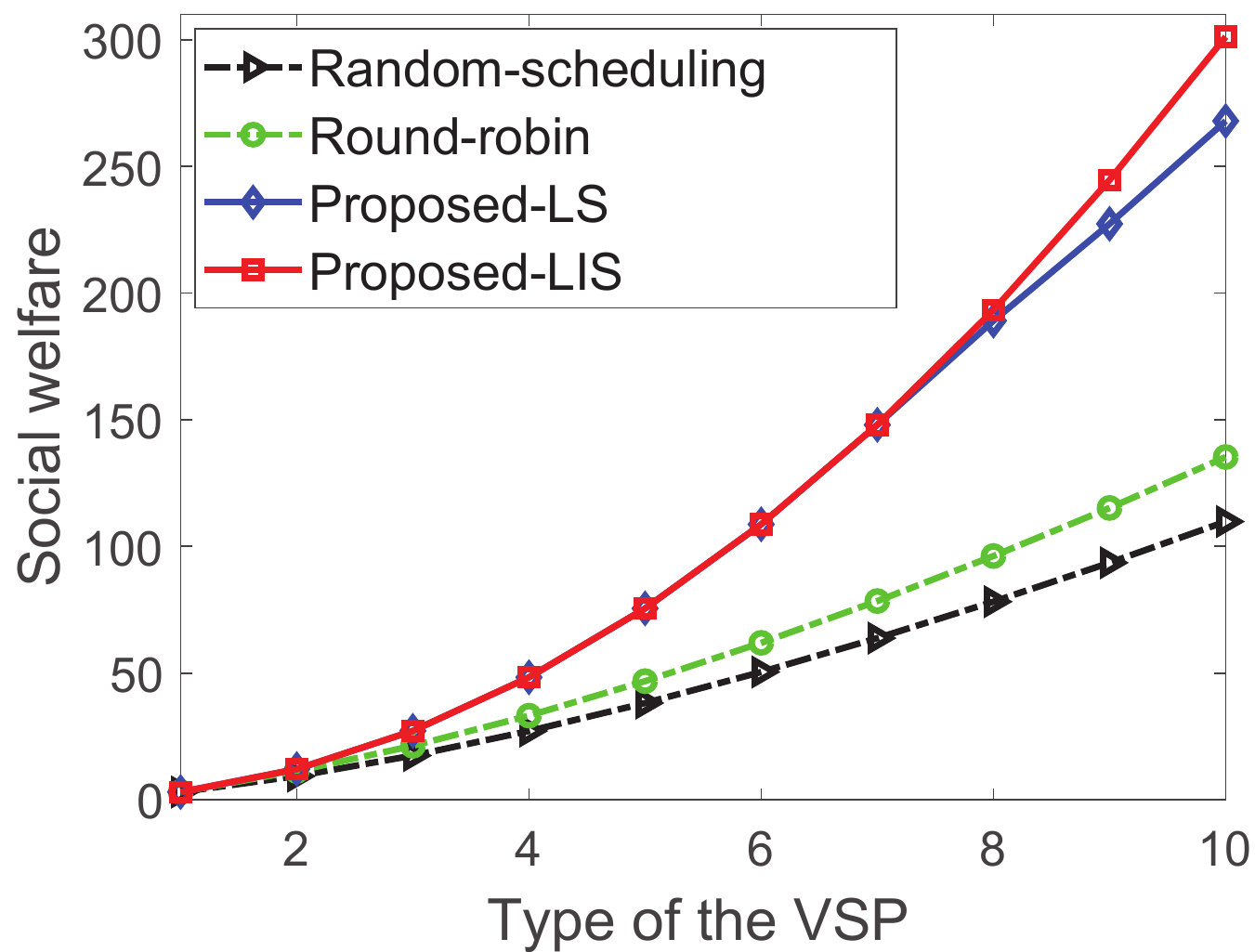} \\ [-0.1cm]
		\text{\footnotesize (c) Social welfare with 5 SVs} & \text{\footnotesize (d) Social welfare with 10 SVs} \\ [0.2cm]
		\end{array}$
		\vspace*{-0.3cm}
		\caption{The total profit of learning SVs and social welfare for non-i.i.d scenario.}
		\vspace{-0.5em}
		\label{fig:SV_profit_niid}
	\end{center}
\end{figure}

We then analyze the total profit of learning SVs and social welfare of the IoV network at the first round as shown in Fig.~\ref{fig:SV_profit_iid} and Fig.~\ref{fig:SV_profit_niid}. Similar to the VSP's profit performance, the proposed-LIS can achieve higher learning SVs' total profits up to 21 and 18.2 times for 5 and 10 learning SVs, respectively, in the i.i.d scenario as well as 5.27 and 8.93 times for 5 and 10 learning SVs, respectively, in the non-i.i.d scenario, compared with those of random and round-robin scheduling methods. The reason is that both methods are likely to train the dataset with very low data size and quality without considering the location-information significance. We also observe that although the proposed-LS can achieve the total profit of learning SVs close to that of proposed-LIS, the proposed-LIS can increase the performance gap by 1.38 and 1.31 times for i.i.d and non-i.i.d scenarios, respectively, when a higher type is employed by the VSP.
The above performance aligns with the social welfare of the network in Fig.~\ref{fig:SV_profit_iid}(c)-(d) and Fig.~\ref{fig:SV_profit_niid}(c)-(d). Specifically, the proposed-LS and proposed-LIS can improve the social welfare up to 3.63 and 4.2 times (i.i.d scenario) as well as 2.44 and 2.74 times (non-i.i.d scenario), respectively, compared with those of random and round-robin scheduling methods.

\subsection{Dynamic Federated Learning Accuracy Performance}


\subsubsection{I.i.d scenario}

\begin{figure}[!]
	\begin{center}
		$\begin{array}{cc} 
		\epsfxsize=1.65 in \epsffile{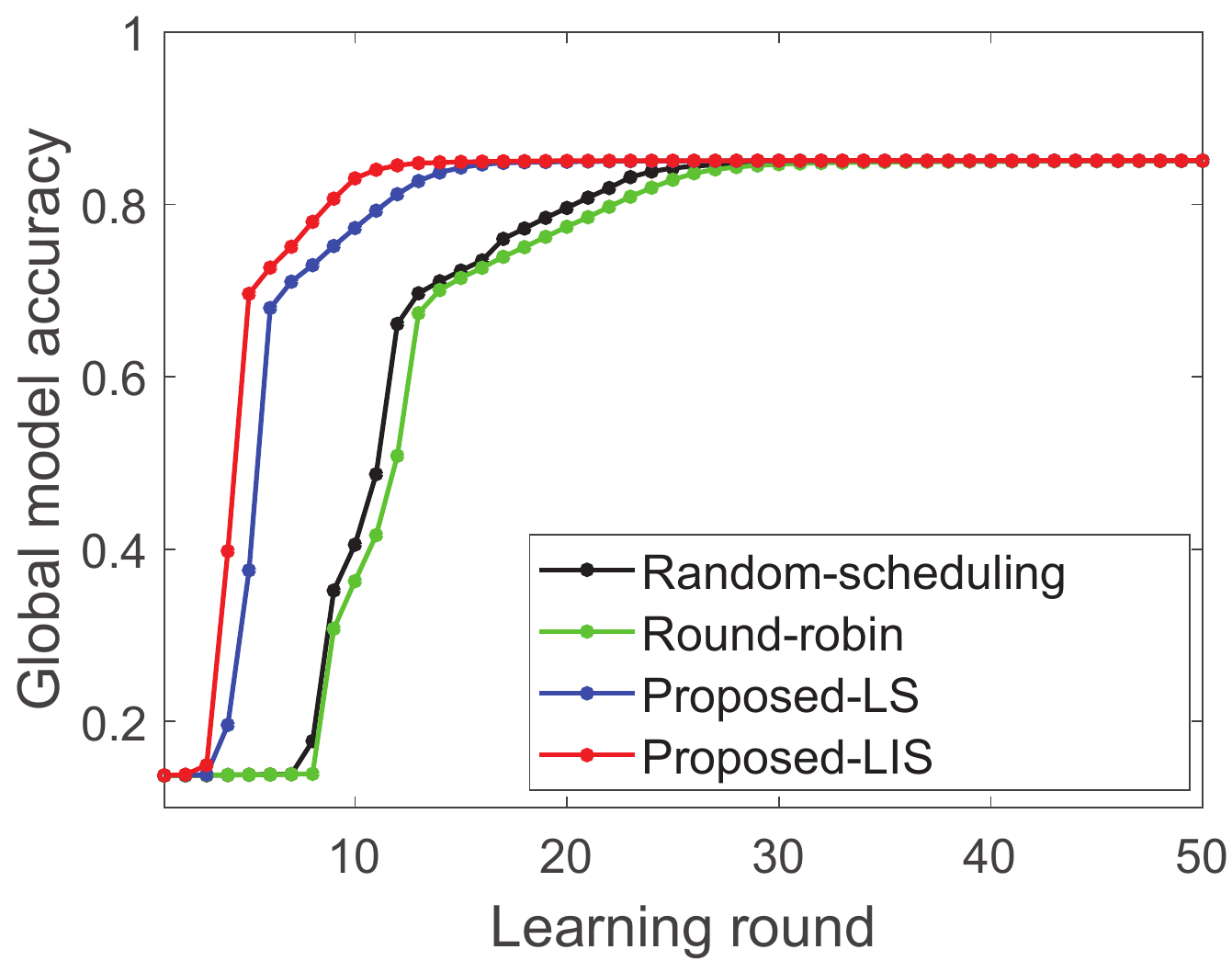} &
		\hspace*{-.3cm}
		\epsfxsize=1.65 in \epsffile{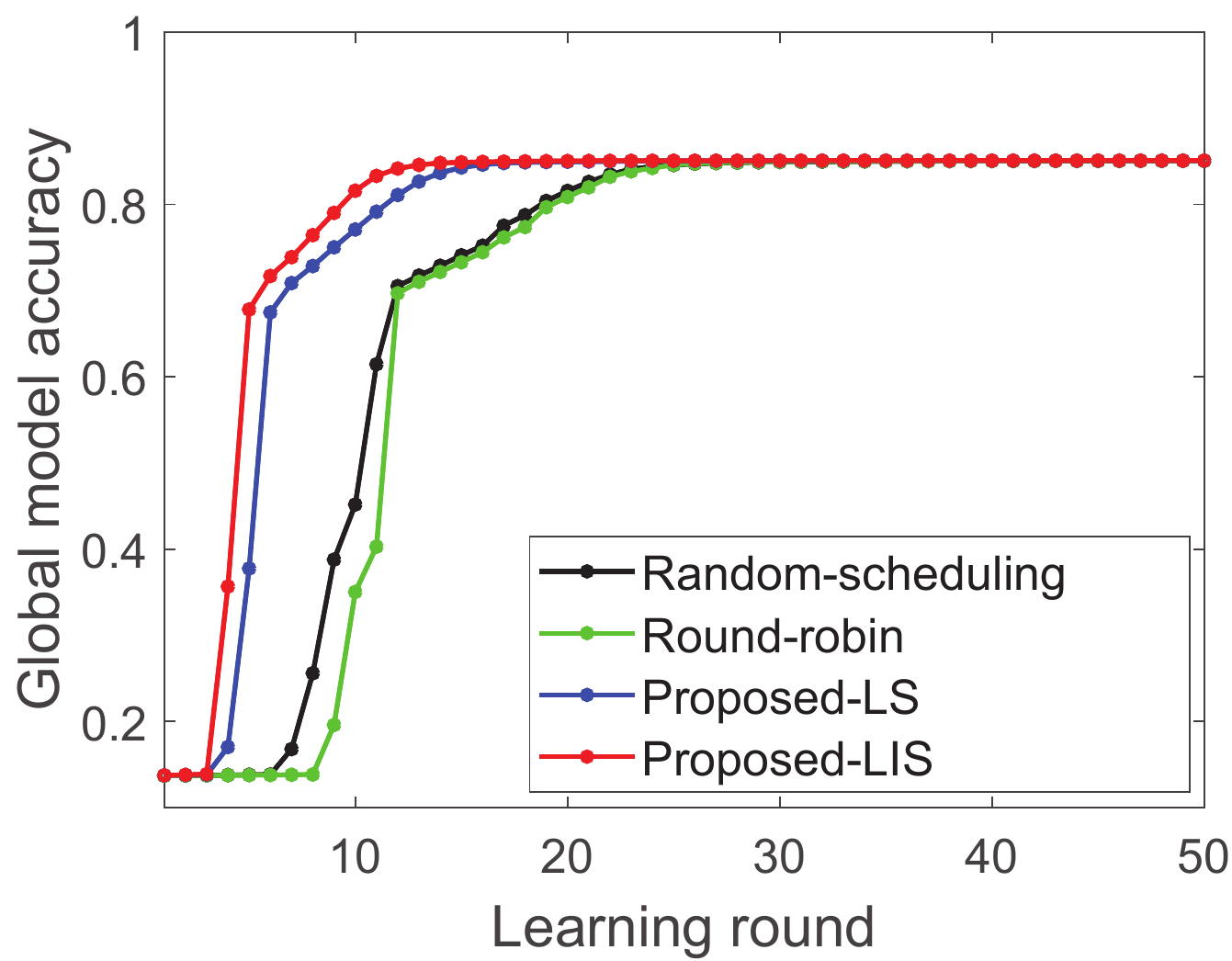} \\ [-0.1cm]
		\text{\footnotesize (a) Accuracy using 5 SVs} & \text{\footnotesize (b) Accuracy using 10 SVs}  \\ [0.2cm]
		\vspace*{-0cm}
		\epsfxsize=1.65 in \epsffile{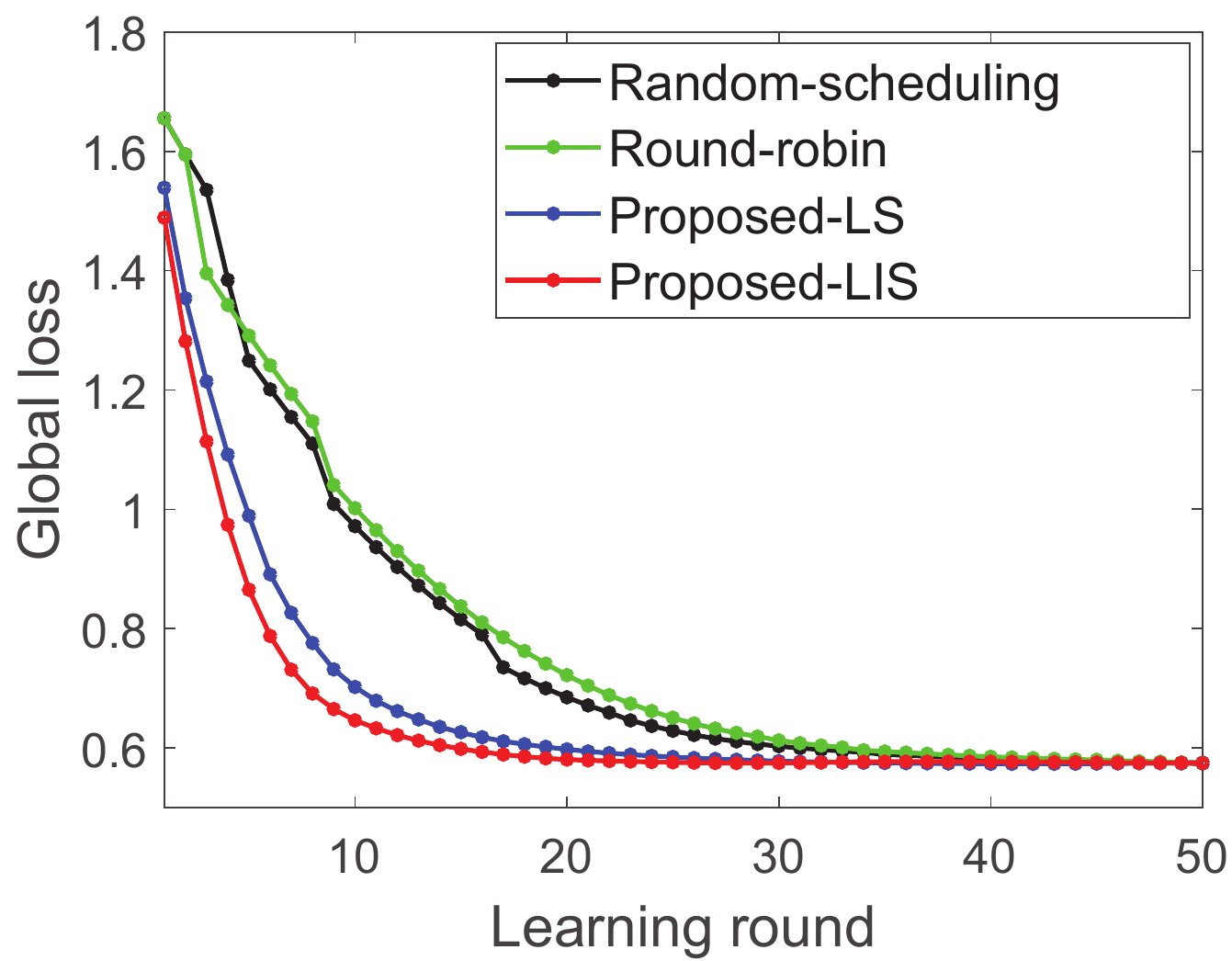} &
		\hspace*{-.3cm}
		\epsfxsize=1.65 in \epsffile{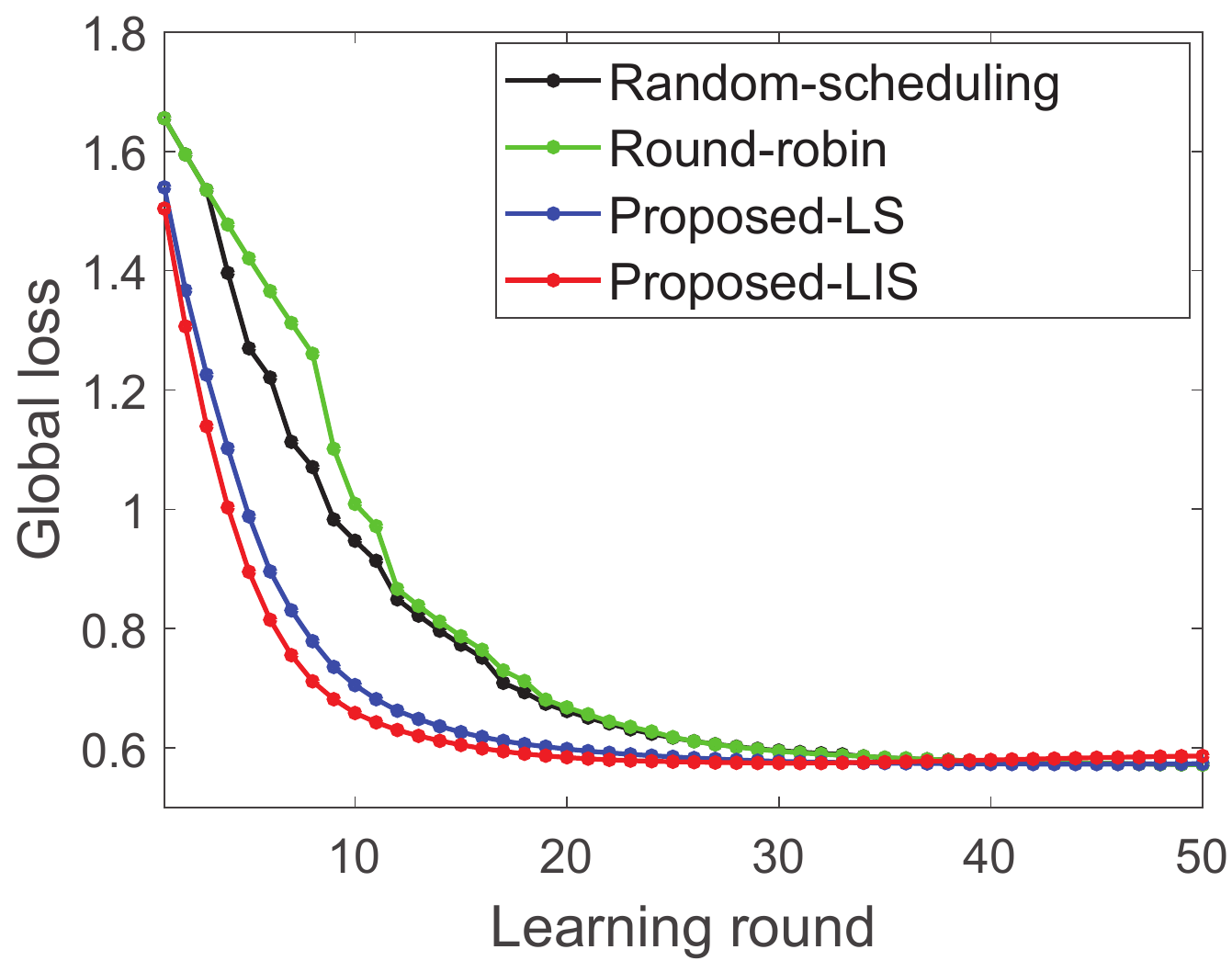} \\ [-0.1cm]
		\text{\footnotesize (c) Global loss using 5 SVs} & \text{\footnotesize (d) Global loss using 10 SVs} \\ [0.2cm]
		\end{array}$
		\vspace*{-0.3cm}
		\caption{The performance of proposed FL with SV selection using various learning SVs with i.i.d datasets.}
		\vspace{-0.5em}
		\label{fig:accuracy_iid}
	\end{center}
\end{figure}

We first investigate the global model accuracy for the i.i.d scenario when the number of learning SVs, i.e., $N$, increases and the VSP has type 10. In particular, when 5 SVs are used to perform the FL algorithm for each round as shown in Fig.~\ref{fig:accuracy_iid}(a), the proposed-LIS can converge 57\% (or 17 rounds) earlier than those of the random scheduling and round-robin with an accuracy of 85\%. Furthermore, the proposed-LIS can achieve the convergence 28\% (or 5 rounds) faster than that of the proposed-LS at the same accuracy level. This is because the proposed-LIS only selects an SV with high information significance within the significant areas. Meanwhile, the baseline FL methods select the learning SVs from the total active SVs without considering the location and information significance-based SV selection. In this case, there exists high probability that they only use \emph{l-SVs} most of the learning rounds, leading to a slow convergence. For the proposed-LS, it only selects SVs in the significant areas without information significance consideration. As a result, its convergence accuracy will be slightly delayed. These accuracy performances align with the global loss ones in Fig.~\ref{fig:accuracy_iid}(c) where the proposed-LIS can reach the fastest minimum convergence compared with other learning methods.

The interesting point of accuracy performance can be observed when 10 learning SVs are considered in Fig.~\ref{fig:accuracy_iid}(b). Specifically, the learning methods other than the proposed-LIS can improve the accuracy convergence speed by 18\% (or 6 rounds) because of the existence of more accurate/meaningful local models from more learning SVs. Nevertheless, this observation does not apply for the proposed-LIS. The reason is that when each learning SV contains i.i.d local dataset with high information significance, the aggregation of all local models at the VSP will generate the same model accuracy compared with the condition when one \emph{h-SV} is used. As such, the additional datasets with the same high information significance from other SVs do not further improve the convergence speed~\cite{Amiri:2020}.

\subsubsection{Non-i.i.d scenario}

\begin{figure}[!]
	\begin{center}
		$\begin{array}{cc} 
		\epsfxsize=1.65 in \epsffile{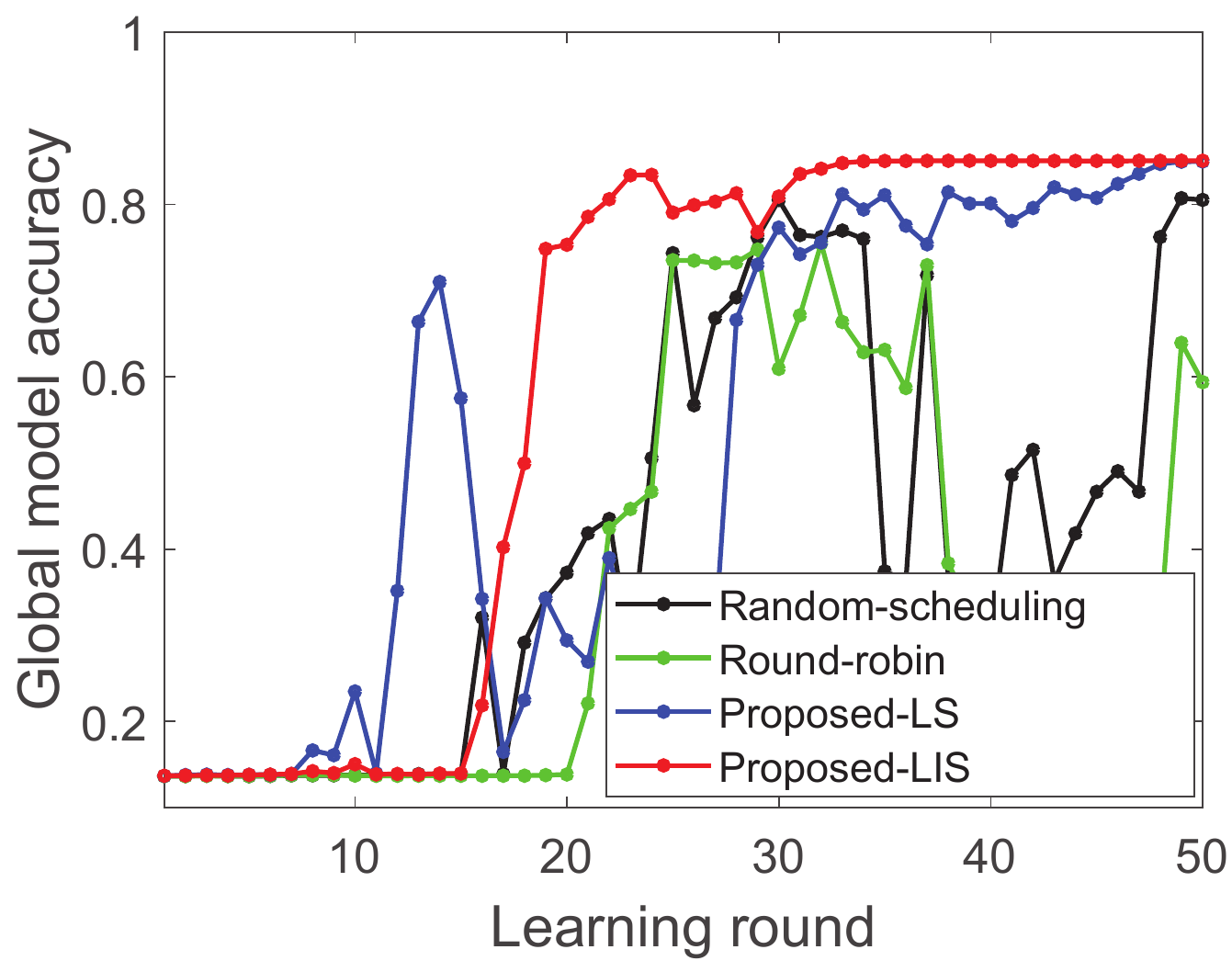} &
		\hspace*{-.3cm}
		\epsfxsize=1.65 in \epsffile{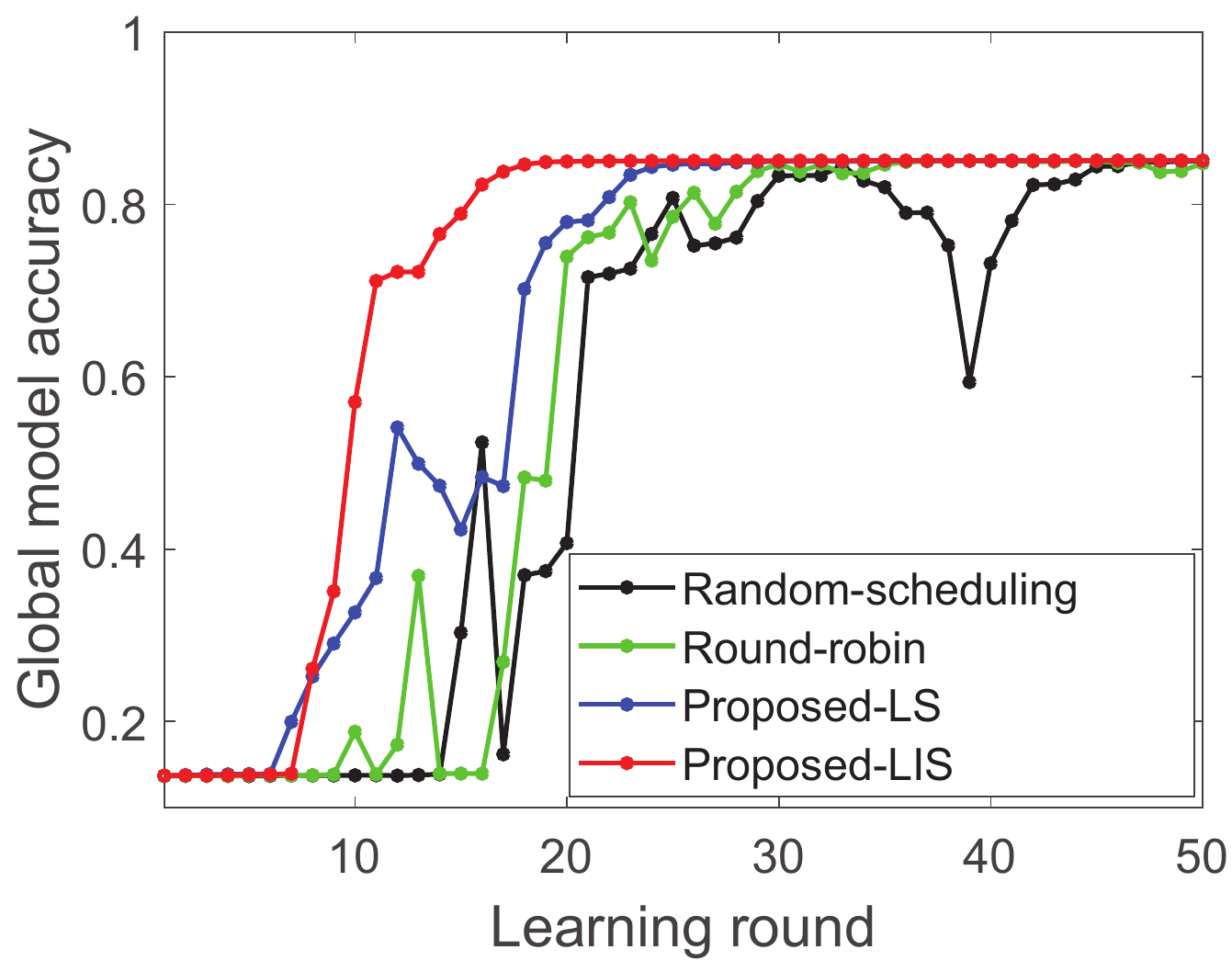} \\ [-0.1cm]
		\text{\footnotesize (a) Accuracy using 5 SVs} & \text{\footnotesize (b) Accuracy using 10 SVs} \\ [0.2cm]
		\vspace*{-0cm}
		\epsfxsize=1.65 in \epsffile{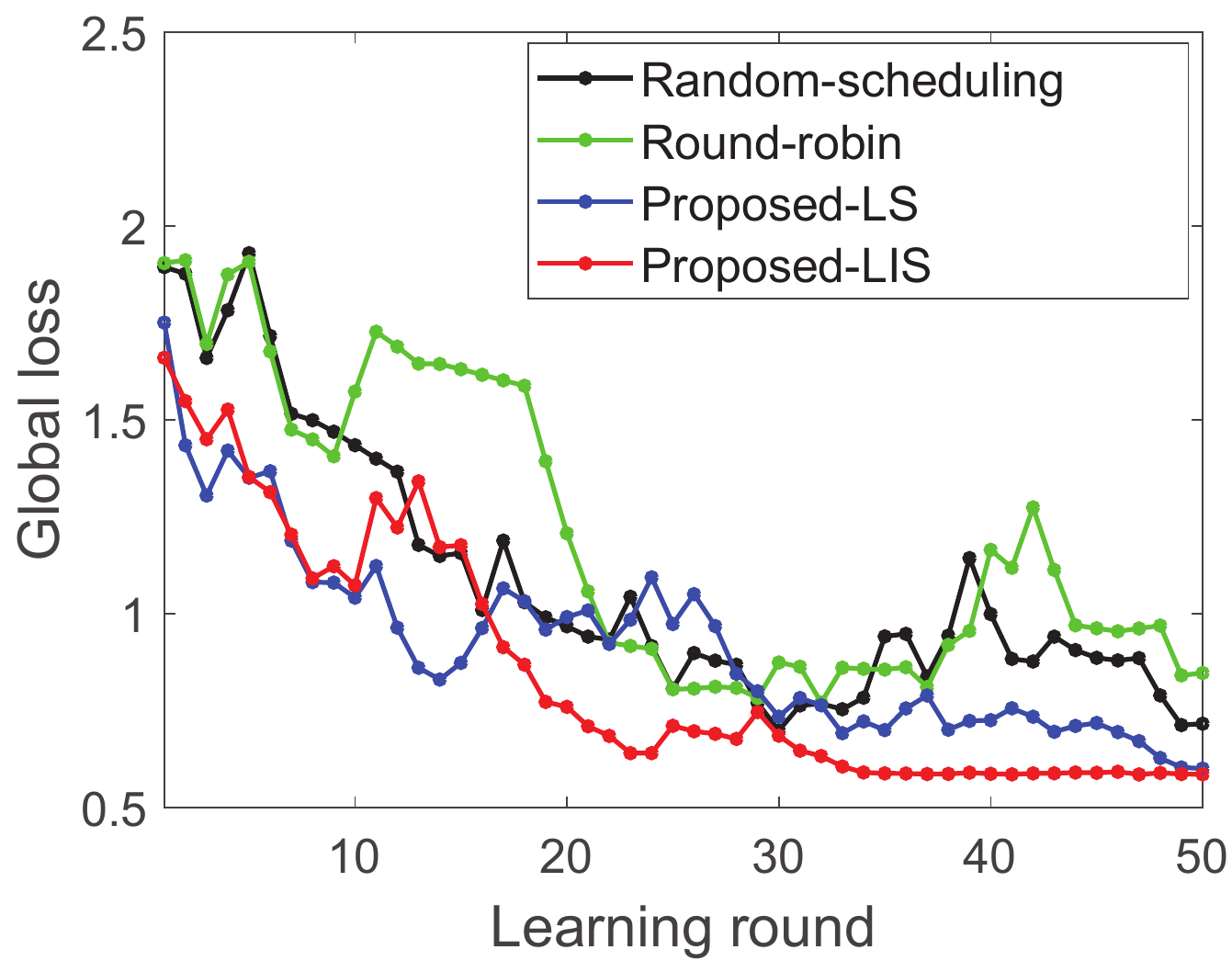} &
		\hspace*{-.3cm}
		\epsfxsize=1.65 in \epsffile{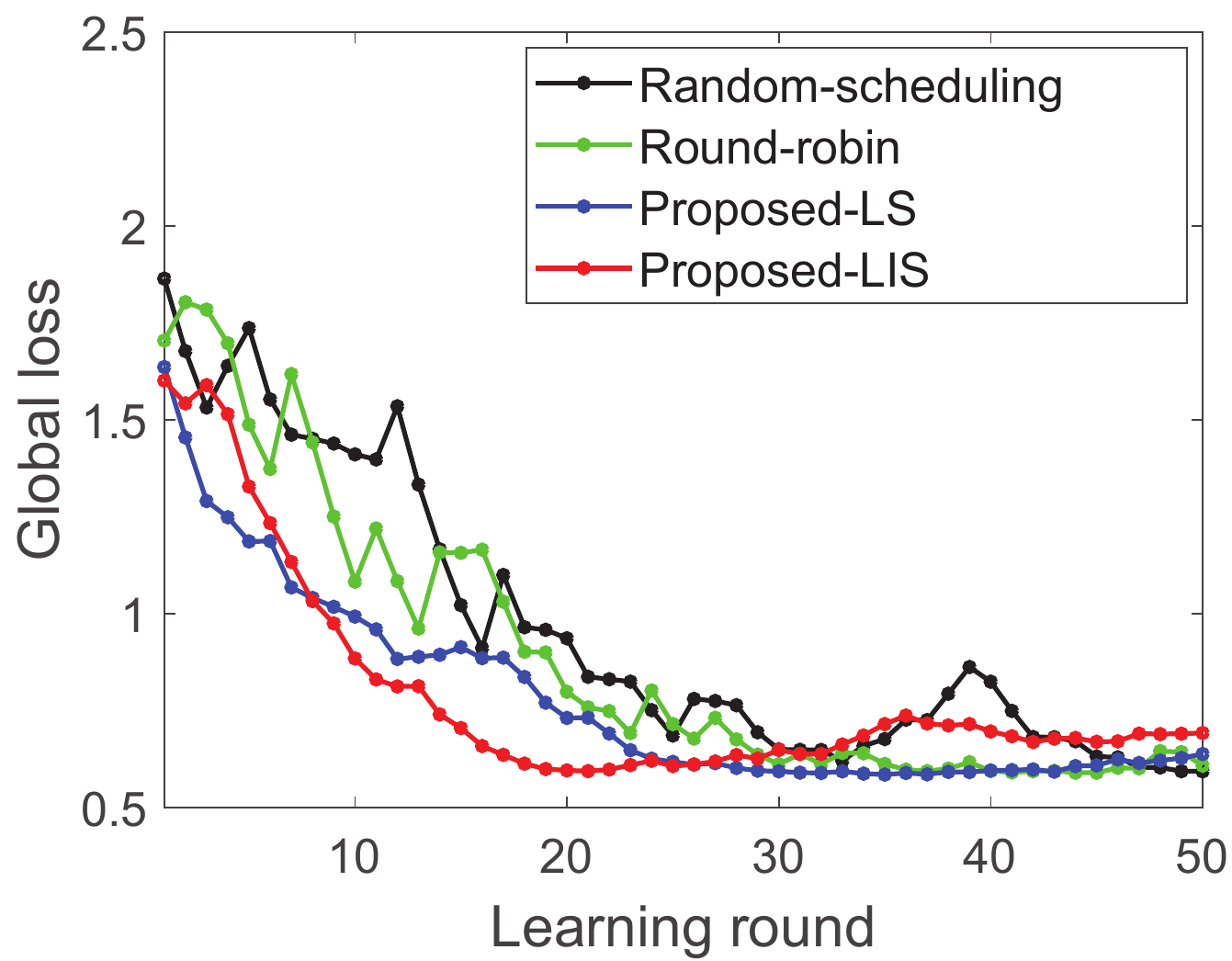} \\ [-0.1cm]
		\text{\footnotesize (c) Global loss using 5 SVs} & \text{\footnotesize (d) Global loss using 10 SVs} \\ [0.2cm]
		\end{array}$
		\vspace*{-0.3cm}
		\caption{The performance of proposed FL with SV selection using various learning SVs with non-i.i.d datasets.}
		\vspace{-0.5em}
		\label{fig:accuracy_niid}
	\end{center}
\end{figure}

In contrast to the i.i.d scenario, the convergence speed for the proposed-LIS gets faster when more learning SVs are selected at each round for the non-i.i.d scenario. As shown in Fig.~\ref{fig:accuracy_niid}(a), all FL methods suffer from the fluctuated/unstable learning performance when the VSP utilizes 5 learning SVs (due to biased local dataset from each active SV~\cite{Amiri:2020}). Nevertheless, the proposed-LIS eventually can achieve the convergence within 30 learning rounds with accuracy level 85\%. Specifically, the proposed-LIS still outperforms the proposed-LS up to 17.5\% (or 7 rounds) in terms of convergence speed due to the selection guarantee of \emph{h-SVs}. Moreover, other baseline methods cannot even reach the convergence because they likely to select \emph{l-SVs} with very biased local datasets. Similar to the i.i.d scenario, the global loss performance in Fig.~\ref{fig:accuracy_niid}(c) aligns with the accuracy one where the proposed-LS and proposed-LIS can reach the fastest minimum convergence compared with those of the random scheduling and round-robin methods. 

When more learning SVs are selected by the VSP, the accuracy performances for all FL methods can be improved. As observed in Fig.~\ref{fig:accuracy_niid}(b), the proposed-LIS can speed up the accuracy convergence by 45\% (or 15 rounds) when the VSP schedules 10 learning SVs for each round. The reason is that the proposed-LIS can train more \emph{h-SVs} and \emph{m-SVs} with less biased local datasets, thereby reducing the unfairness among different SVs. Furthermore, the proposed-LIS can achieve the convergence speed up to 62\% (or 29 rounds) faster than those of other methods when 10 learning SVs are used, respectively. These results can provide useful information for the VSP in practice to determine the best SV selection method for the FL process in terms of stability, robustness, and flexibility.

\subsection{Relationship Between Contract and Federated Learning Performance} 

\begin{figure}[!]
	\begin{center}
		$\begin{array}{cc} 
		\epsfxsize=1.65 in \epsffile{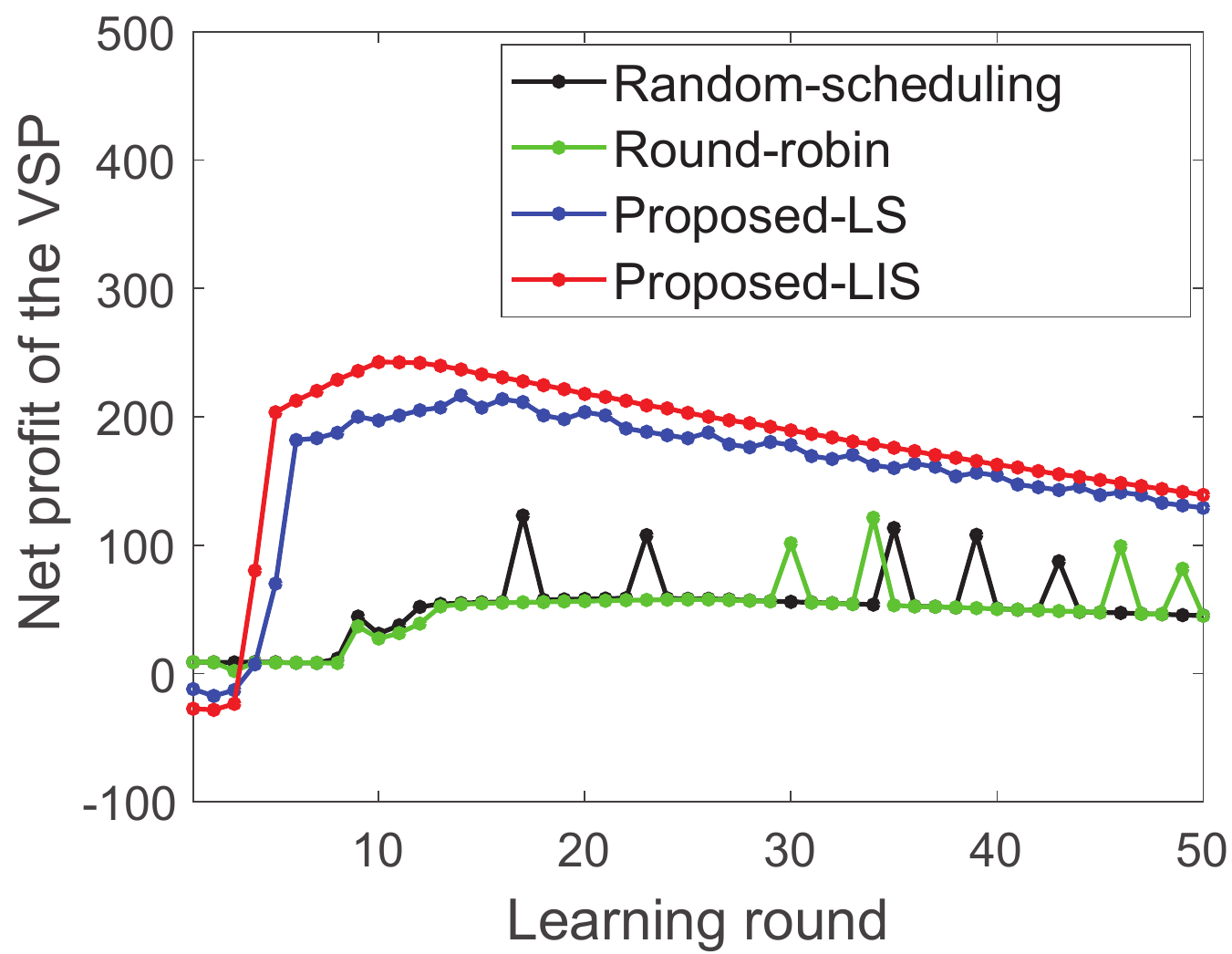} &
		\hspace*{-.3cm}
		\epsfxsize=1.65 in \epsffile{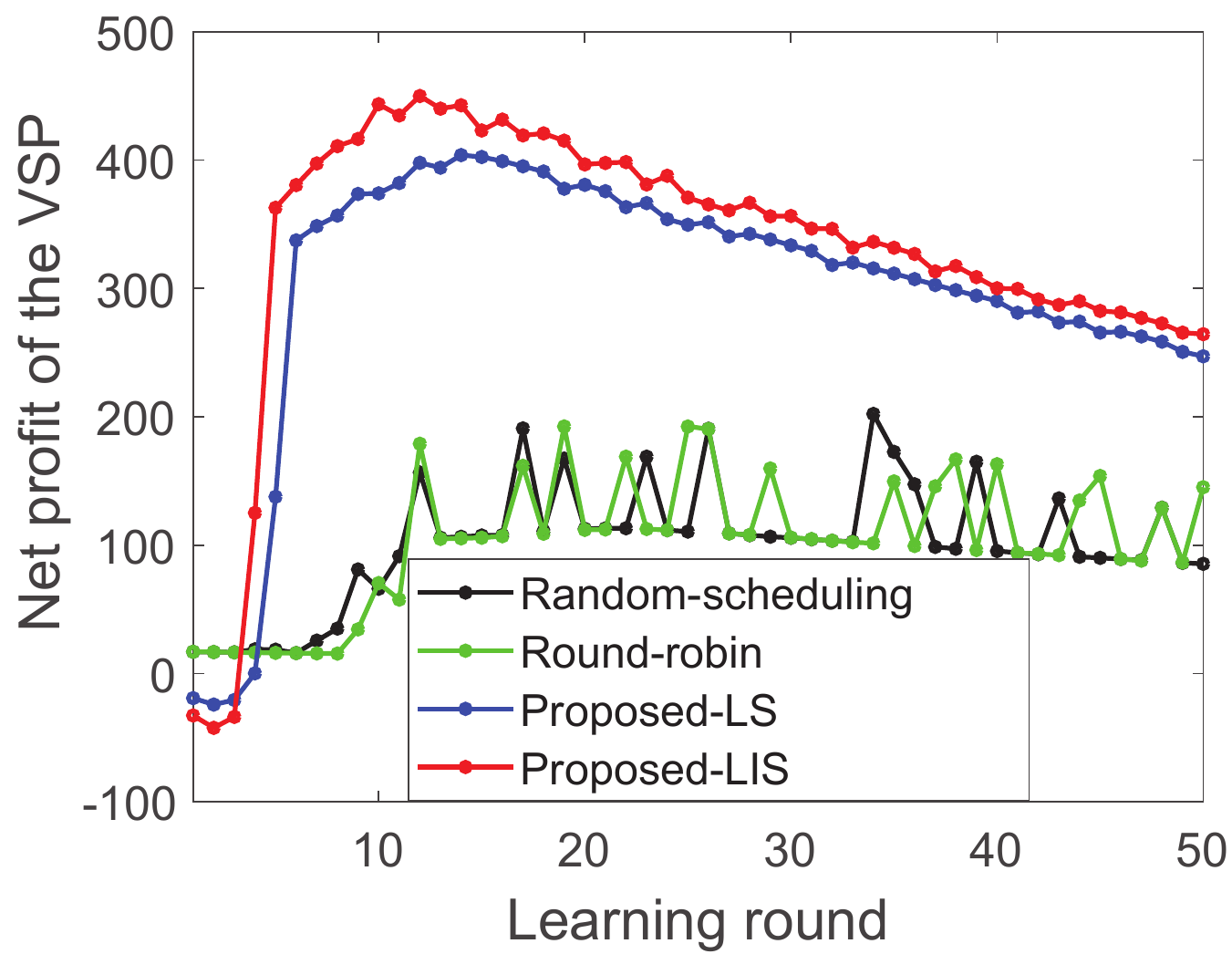} \\ [-0.1cm]
		\text{\footnotesize (a) VSP's net profit with 5 SVs} & \text{\footnotesize (b) VSP's net profit with 10 SVs} \\ [0.2cm]
		\vspace*{-0cm}
		\epsfxsize=1.65 in \epsffile{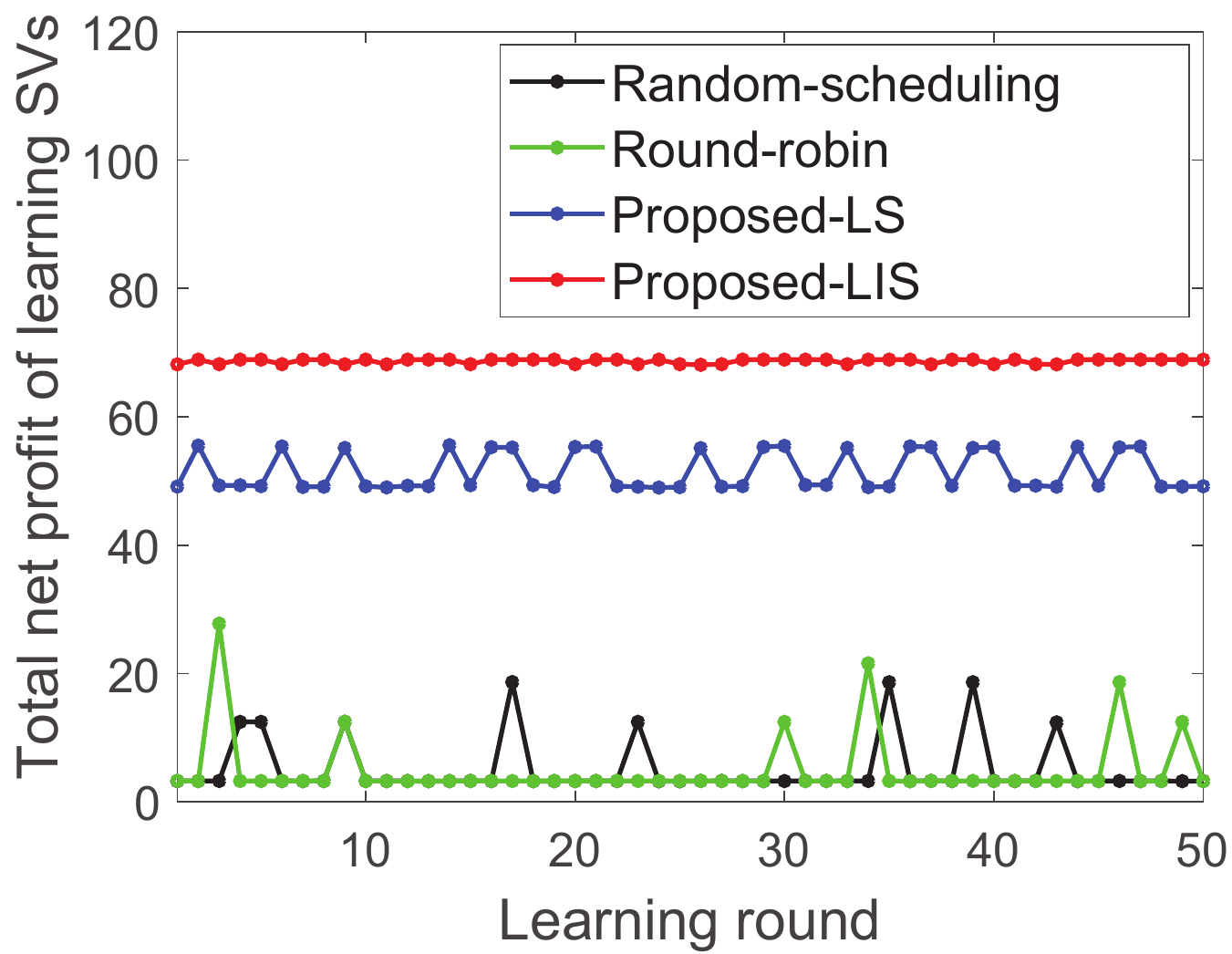} &
		\hspace*{-.3cm}
		\epsfxsize=1.65 in \epsffile{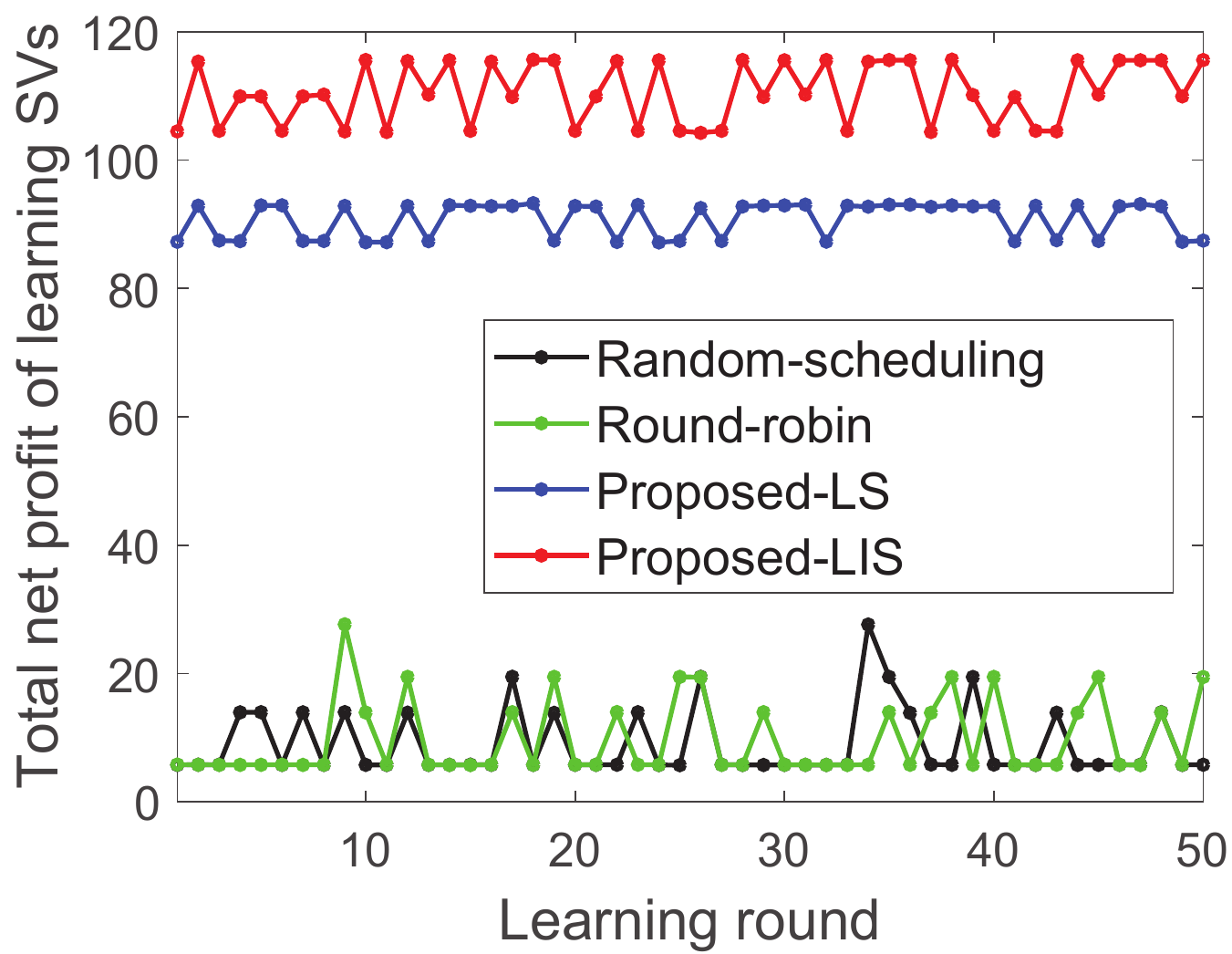} \\ [-0.1cm]
		\text{\footnotesize (c) Total net profit of 5 SVs} & \text{\footnotesize (d) Total net profit of 10 SVs} \\ [0.2cm]
		\end{array}$
		\vspace*{-0.3cm}
		\caption{Net profits of the VSP and learning SVs for all learning rounds when i.i.d datasets are trained.}
		\vspace{-0.5em}
		\label{fig:VSP_SV_profit_iid_allrounds}
	\end{center}
\end{figure}

\begin{figure}[!]
	\begin{center}
		$\begin{array}{cc} 
		\epsfxsize=1.65 in \epsffile{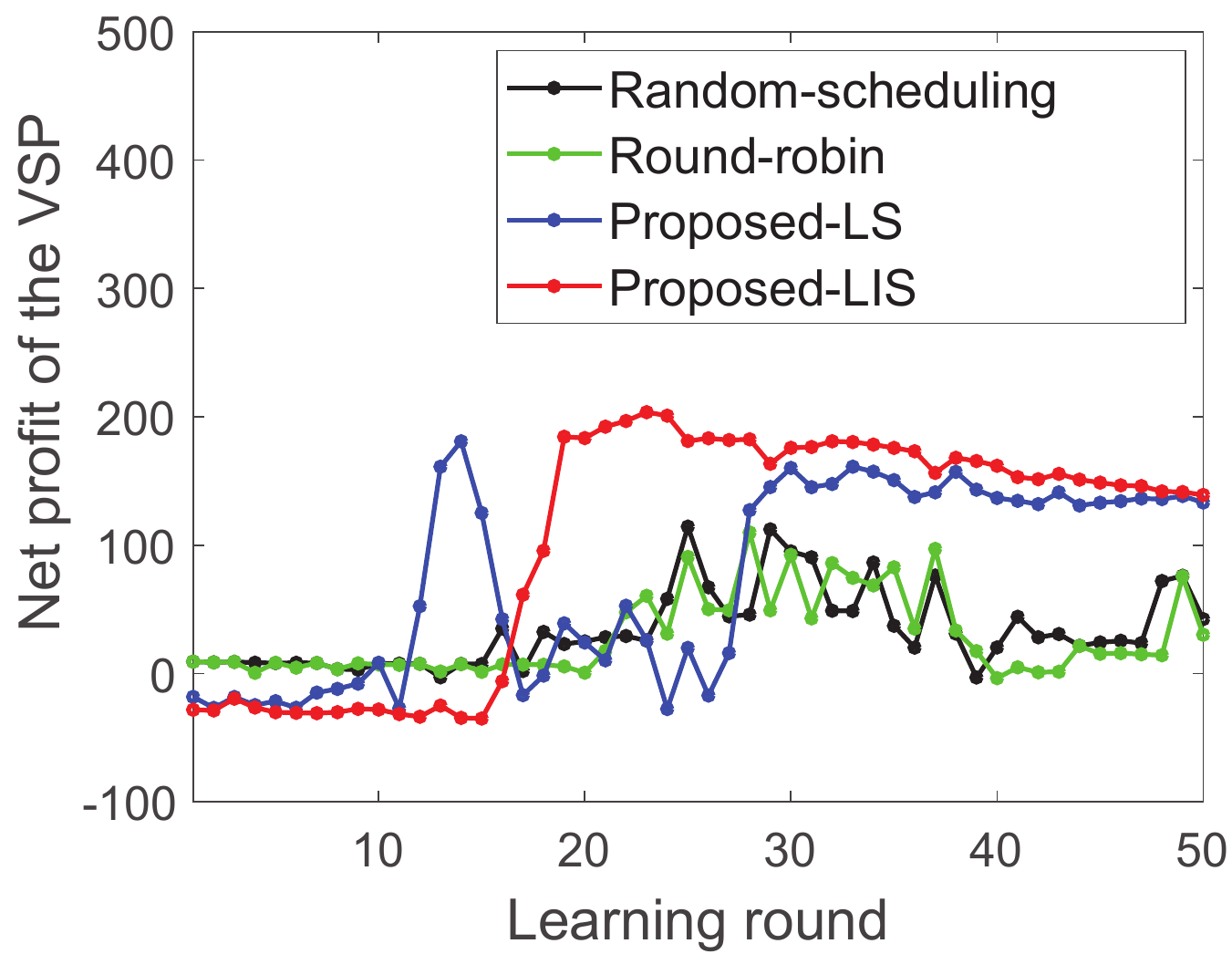} &
		\hspace*{-.3cm}
		\epsfxsize=1.65 in \epsffile{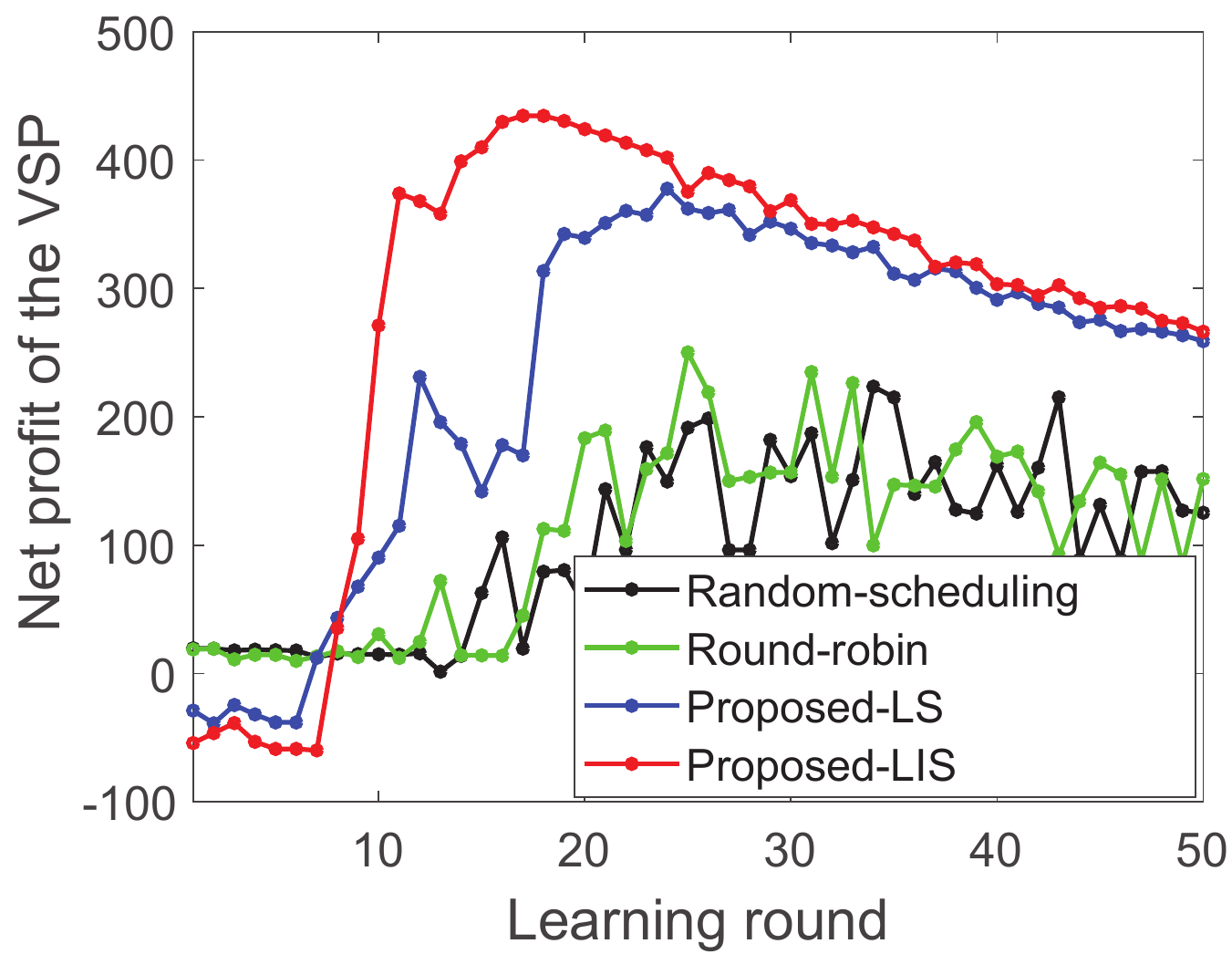} \\ [-0.1cm]
		\text{\footnotesize (a) VSP's net profit with 5 SVs} & \text{\footnotesize (b) VSP's net profit with 10 SVs} \\ [0.2cm]
		\vspace*{-0cm}
		\epsfxsize=1.65 in \epsffile{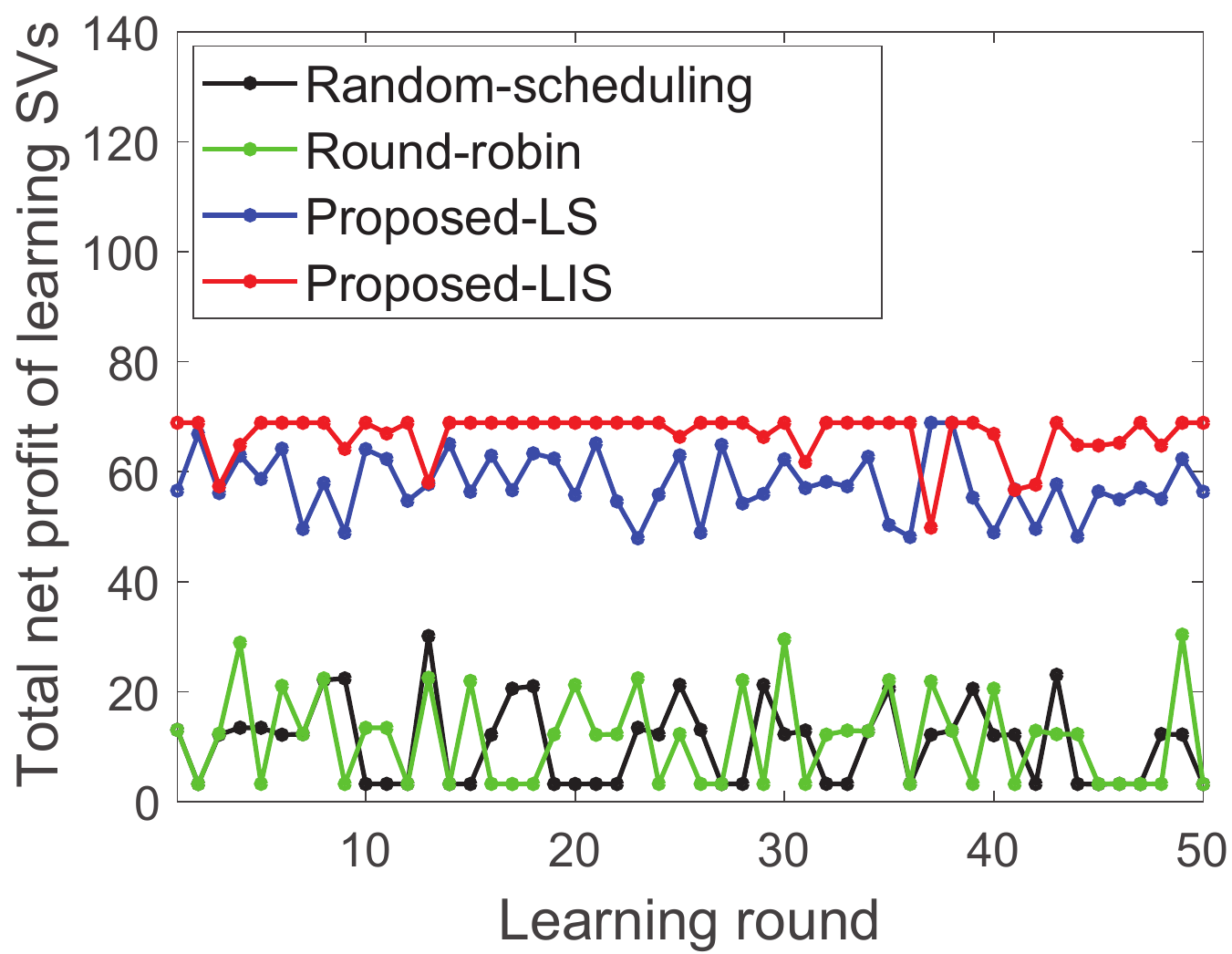} &
		\hspace*{-.3cm}
		\epsfxsize=1.65 in \epsffile{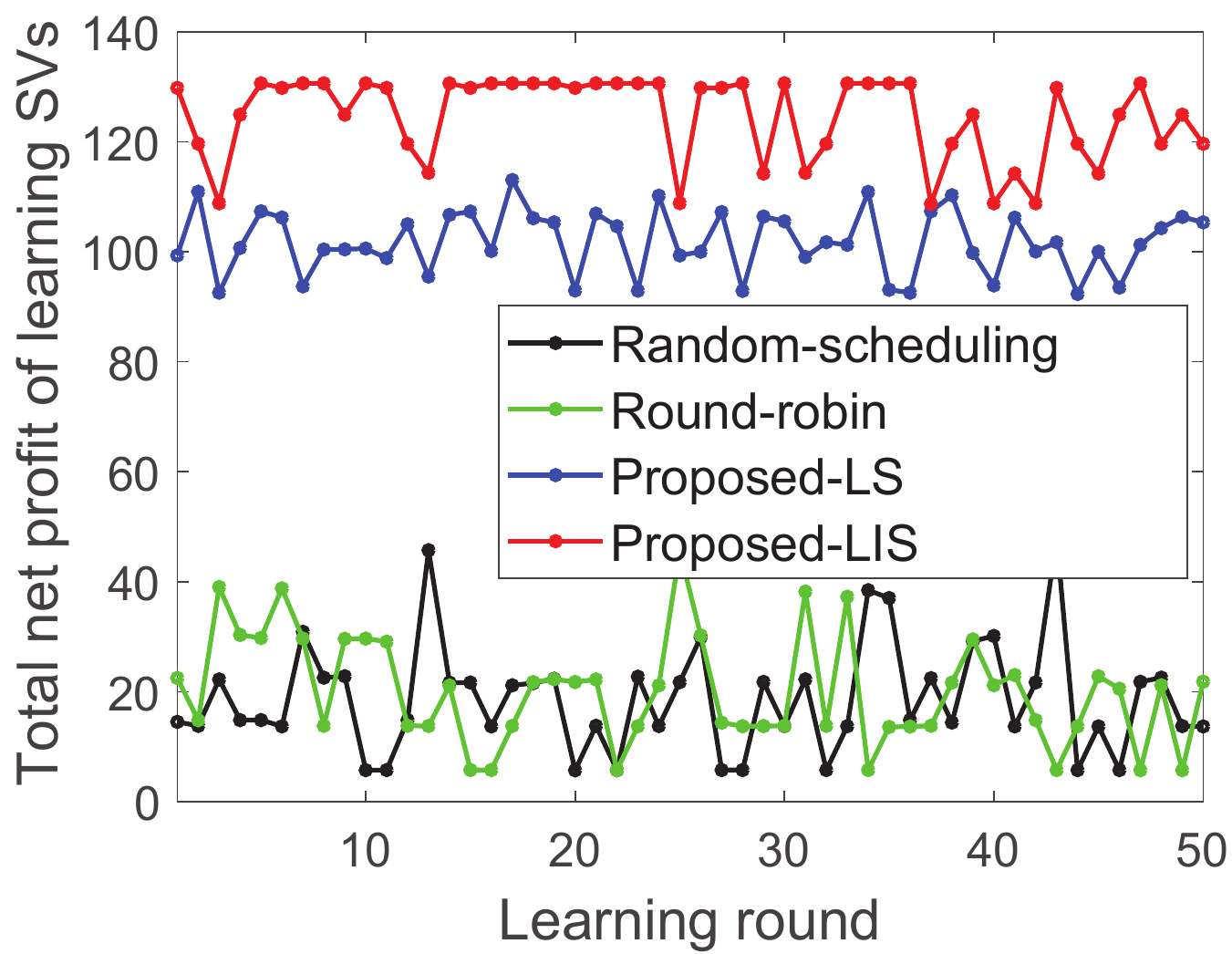} \\ [-0.1cm]
		\text{\footnotesize (c) Total net profit of 5 SVs} & \text{\footnotesize (d) Total net profit of 10 SVs} \\ [0.2cm]
		\end{array}$
		\vspace*{-0.3cm}
		\caption{Net profits of the VSP and learning SVs for all learning rounds when non-i.i.d datasets are trained.}
		\vspace{-0.5em}
		\label{fig:VSP_SV_profit_niid_allrounds}
	\end{center}
\end{figure}

\begin{figure}[!]
	\begin{center}
		$\begin{array}{cc} 
		\epsfxsize=1.65 in \epsffile{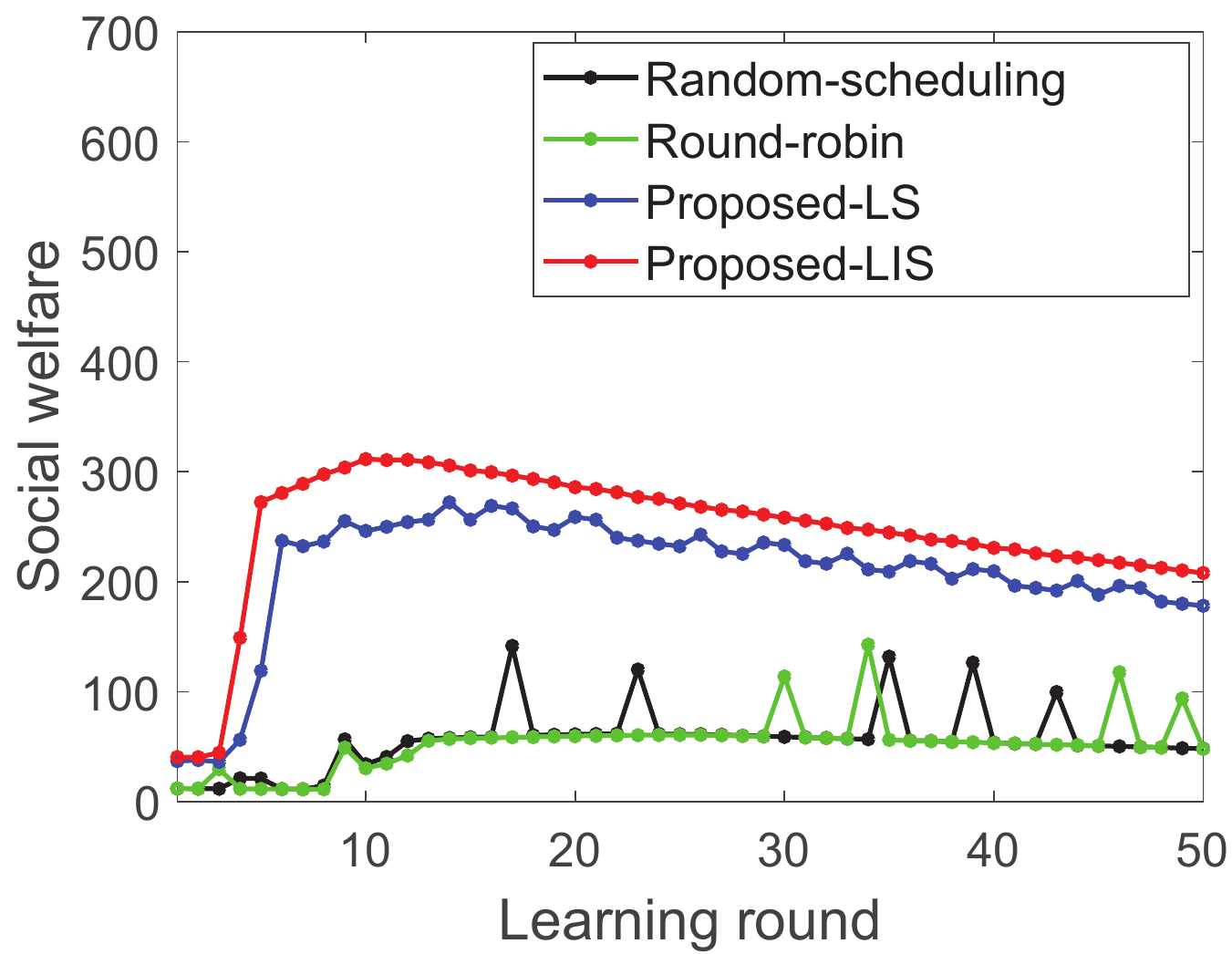} &
		\hspace*{-.3cm}
		\epsfxsize=1.65 in \epsffile{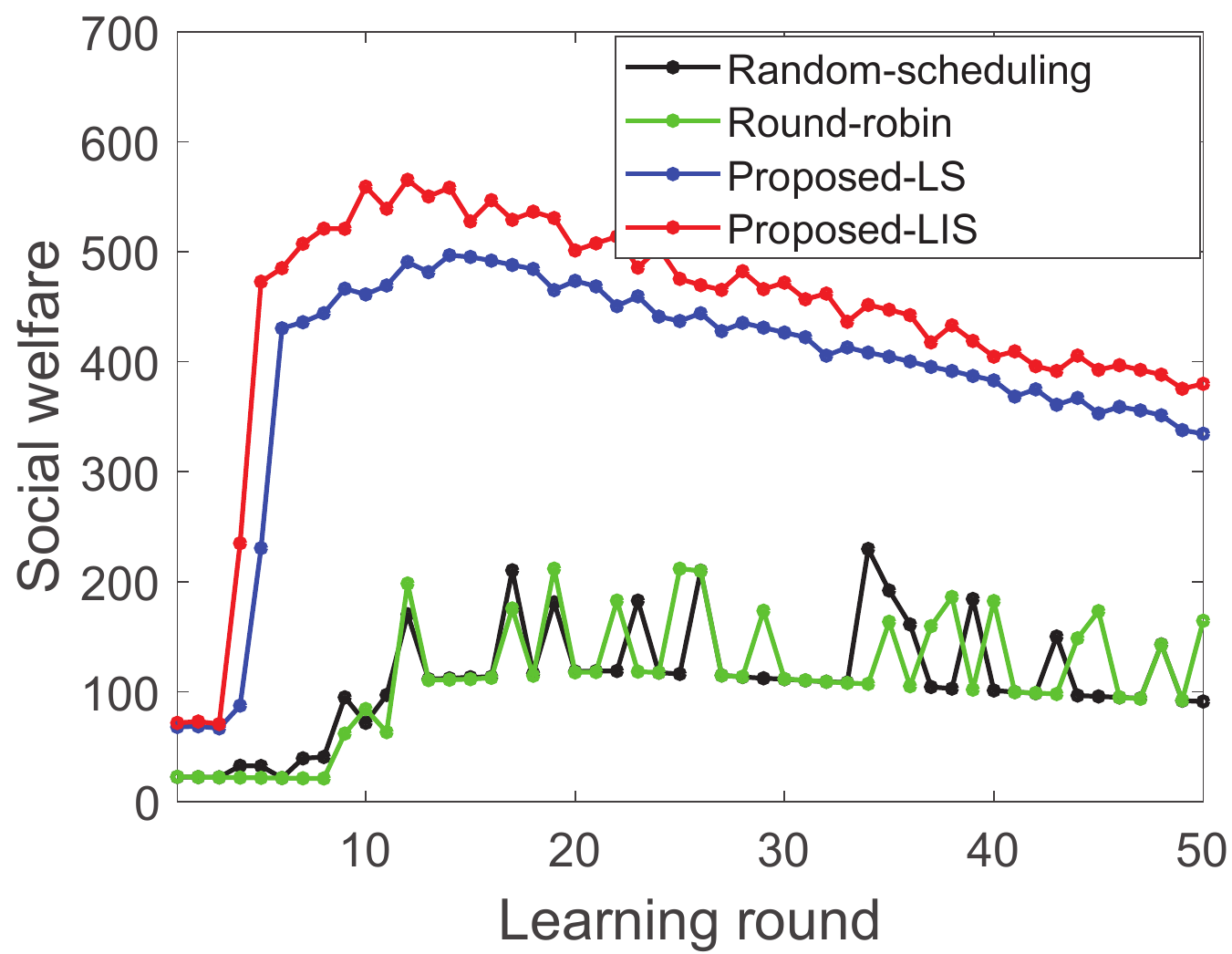} \\ [-0.1cm]
		\text{\footnotesize (a) 5 SVs (i.i.d)} & \text{\footnotesize (b) 10 SVs (i.i.d)} \\ [0.2cm]
		\vspace*{-0cm}
		\epsfxsize=1.65 in \epsffile{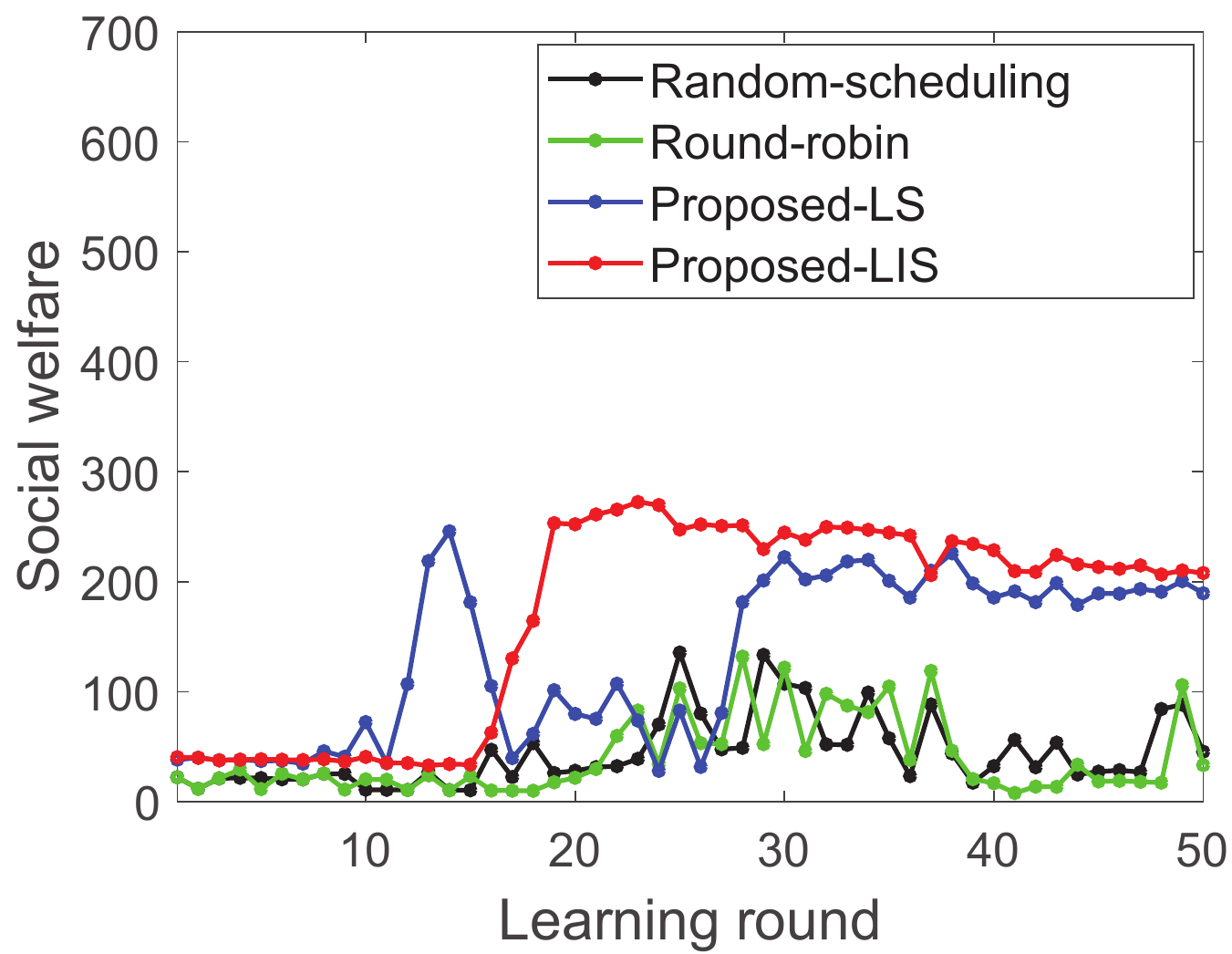} &
		\hspace*{-.3cm}
		\epsfxsize=1.65 in \epsffile{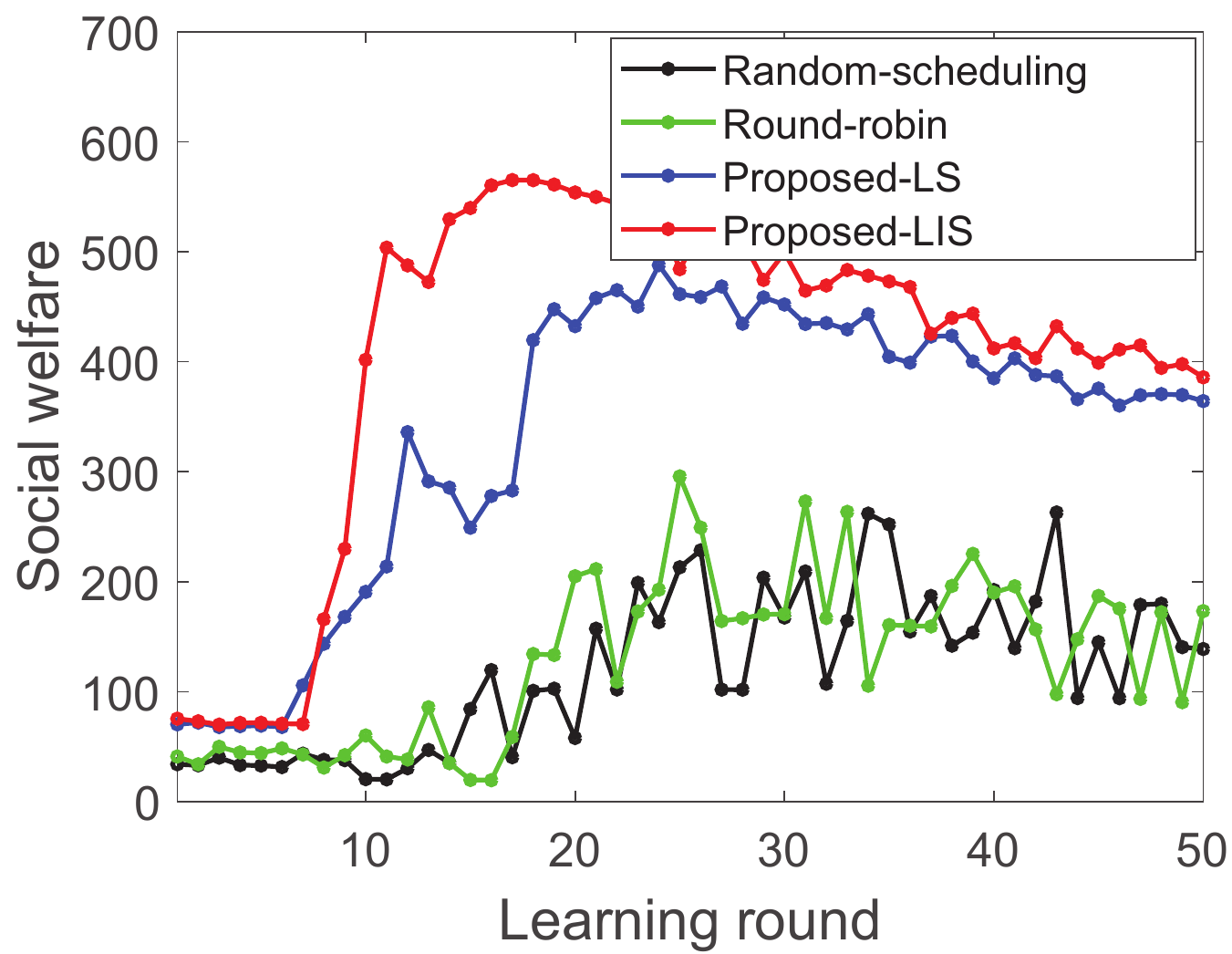} \\ [-0.1cm]
		\text{\footnotesize (c) 5 SVs (non-i.i.d)} & \text{\footnotesize (d) 10 SVs (non-i.i.d)} \\ [0.2cm]
		\end{array}$
		\vspace*{-0.3cm}
		\caption{Final social welfare for all learning rounds when i.i.d and non-i.i.d datasets are trained.}
		\vspace{-0.5em}
		\label{fig:SW_allrounds}
	\end{center}
\end{figure}

This section discusses how the global model accuracy and the global model freshness obtained from the learning process at each round can influence the profits of the VSP (with type 10) and learning SVs as well as the social welfare of the network obtained from the contract optimization. As observed in Fig.~\ref{fig:VSP_SV_profit_iid_allrounds}(a)-(b) when i.i.d datasets are used, the net profit of the VSP (when using the proposed-LS and proposed-LIS) first suffers from negative profit during the first three learning rounds for all cases. The reason is that although the VSP can obtain high information significance values from the learning SVs at these rounds, the obtained global model accuracy cannot compensate the satisfaction function due to the higher cost, i.e., high payment for the learning SVs with high QoI. As the learning round increases, the global model accuracy start to improve significantly by 60\%, thereby enhancing the net profit of the VSP (with the maximum profit between round 10 and 15). Due to the global model freshness impact and less accuracy improvement at further learning rounds, the profit of the VSP gradually decreases. Nonetheless, our proposed-LS and proposed-LIS can still outperform other methods up to 7.25 and 8.92 times in terms of the VSP's net profit, respectively. 

For the total net profit of learning SVs, the performance is not affected by the global model accuracy and freshness. Alternatively, it remains the same as the one obtained from the contract optimization. The reason is that the profits of learning SVs depend only on their current round's trained datasets with certain information significance values, and computation as well as memory costs at the learning SVs. As can be seen in Fig.~\ref{fig:VSP_SV_profit_iid_allrounds}(c)-(d), our proposed-LIS can achieve the highest total profit performance of learning SVs for most rounds due to the use of more \emph{h-SVs}. Similar trends with more fluctuated profit performance can be observed in Fig.~\ref{fig:VSP_SV_profit_niid_allrounds} when non-i.i.d datasets are used for the cases of 5 and 10 learning SVs. However, the use of 5 learning SVs in Fig.~\ref{fig:VSP_SV_profit_niid_allrounds}(a) and Fig.~\ref{fig:VSP_SV_profit_niid_allrounds}(c) triggers low net profits for both VSP and learning SVs of the proposed-LS and proposed-LIS during the first 16 and 27 learning rounds, respectively (due to the unstable accuracy performance in Fig.~\ref{fig:accuracy_niid}(a)). To support the performances in Fig.~\ref{fig:VSP_SV_profit_iid_allrounds} and Fig.~\ref{fig:VSP_SV_profit_niid_allrounds}, we show the final social welfare for both i.i.d and non-i.i.d scenarios in Fig.~\ref{fig:SW_allrounds}. Specifically, the trend of social welfare follows the one of the VSP's net profit with increased profit values due to addition of learning SVs' total profit. In this case, our proposed-LIS can achieve the social welfare up to 24.5 and 27.2 times for i.i.d and non-i.i.d scenarios, respectively, compared with those of other methods.  
To this end, we can summarize that the i.i.d scenario can provide more stable economic benefit performance throughout the learning rounds compared with that of the non-i.i.d scenario. However, the non-i.i.d scenario will provide more practical economic benefit performance because of the unique and biased dataset (generated from diverse mobile activities) at each learning SV in the practical IoV network.

\section{Conclusion}
\label{sec:Conc}

In this paper, we have proposed the novel economic framework for the IoV network to maximize the profits for the VSP and learning SVs in the dynamic FL process. Specifically, we have introduced the dynamic SV selection method to determine a set of active SVs based on the significance of their locations and information at each learning round. From the set of selected SVs, each SV can first collect on-road data and then offer a payment contract to the VSP with respect to its QoI. 
In this case, we have developed the MPOA-based contract optimization problem under the competition of learning SVs and the VSP's common constraints (i.e., the IR and IC constraints) as well as incomplete information of payment budget. To find the optimal contracts, we have transformed the problem into the equivalent low-complexity problem and then implemented the iterative contract algorithm which can achieve the equilibrium solution for all the learning SVs at each learning round. We have also analyzed the convergence of proposed FL and investigated on how the accuracy and freshness of the global prediction model obtained by the FL process at each round affect the net profits of the VSP and learning SVs as well as the social welfare of the network. Through the experimental results, we have shown that our proposed framework can significantly improve the FL convergence speed, profits of the VSP and learning SVs, and the social welfare of the network compared with other baseline FL methods.



\ifCLASSOPTIONcaptionsoff
\newpage
\fi

\vspace{0cm}

\clearpage
\appendices

	\section{Proof of Lemma 1}
	\label{appx:lemma1}
	
	We first prove that $\boldsymbol{\zeta}_j(t) \geq \boldsymbol{\zeta}_{j^*}(t)$ if and only if $\theta_j \geq \theta_{j^*}$ using the IC constraint in~(\ref{eqn:for7b}), which is
	\begin{align}
	&\theta_jS(\boldsymbol{\hat \varrho}_j(t),\boldsymbol{\zeta}_j(t)) - C(\boldsymbol{\hat \varrho}_j(t), \boldsymbol{\varphi}_j(t)) \geq \nonumber\\
	&\theta_jS(\boldsymbol{\hat \varrho}_{j^*}(t),\boldsymbol{\zeta}_{j^*}(t)) - C(\boldsymbol{\hat \varrho}_{j^*}(t), \boldsymbol{\varphi}_{j^*}(t)). \label{eqn:for7c}
	\end{align}
	If the VSP has type $j^*$, then we obtain
	\begin{align}
	&\theta_{j^*}S(\boldsymbol{\hat \varrho}_{j^*}(t),\boldsymbol{\zeta}_{j^*}(t)) - C(\boldsymbol{\hat \varrho}_{j^*}(t), \boldsymbol{\varphi}_{j^*}(t)) \geq \nonumber\\
	&\theta_{j^*}S(\boldsymbol{\hat \varrho}_j(t),\boldsymbol{\zeta}_{j}(t)) - C(\boldsymbol{\hat \varrho}_j(t), \boldsymbol{\varphi}_{j}(t)), \label{eqn:for7d}
	\end{align}
	where $j \neq {j^*}$ and $j, {j^*} \in \mathcal{J}$. From~(\ref{eqn:for7c}) and~(\ref{eqn:for7d}), we can obtain
	\begin{equation}
	\label{eqn:for7e}
	\begin{aligned}
	&\theta_jS(\boldsymbol{\hat \varrho}_j(t),\boldsymbol{\zeta}_j(t)) + \theta_{j^*}S(\boldsymbol{\hat \varrho}_{j^*}(t),\boldsymbol{\zeta}_{j^*}(t)) \geq \\ &\theta_jS(\boldsymbol{\hat \varrho}_{j^*}(t),\boldsymbol{\zeta}_{j^*}(t)) + \theta_{j^*}S(\boldsymbol{\hat \varrho}_j(t),\boldsymbol{\zeta}_{j}(t)),\\
	&\theta_jS(\boldsymbol{\hat \varrho}_j(t),\boldsymbol{\zeta}_j(t)) - \theta_{j^*}S(\boldsymbol{\hat \varrho}_j(t),\boldsymbol{\zeta}_{j}(t)) \geq \\ &\theta_jS(\boldsymbol{\hat \varrho}_{j^*}(t),\boldsymbol{\zeta}_{j^*}(t)) - \theta_{j^*}S(\boldsymbol{\hat \varrho}_{j^*}(t),\boldsymbol{\zeta}_{j^*}(t)), \\
	&S(\boldsymbol{\hat \varrho}_j(t),\boldsymbol{\zeta}_j(t))\big(\theta_j - \theta_{j^*}\big) \geq S(\boldsymbol{\hat \varrho}_{j^*}(t),\boldsymbol{\zeta}_{j^*}(t))\big(\theta_j - \theta_{j^*}\big).
	\end{aligned}
	\end{equation}
	By removing $(\theta_j - \theta_{j^*})$, we have $S(\boldsymbol{\hat \varrho}_j(t),\boldsymbol{\zeta}_j(t)) \geq S(\boldsymbol{\hat \varrho}_{j^*}(t),\boldsymbol{\zeta}_{j^*}(t))$. Since the satisfaction function in~(\ref{eqn:for30a}) is monotonically increasing in $\boldsymbol{\zeta}_j(t)$, thus if $S(\boldsymbol{\hat \varrho}_j(t),\boldsymbol{\zeta}_j(t)) \geq S(\boldsymbol{\hat \varrho}_{j^*}(t),\boldsymbol{\zeta}_{j^*}(t))$, then $\boldsymbol{\zeta}_j(t) \geq \boldsymbol{\zeta}_{j^*}(t)$ given $\boldsymbol{\hat \varrho}_{j} = \boldsymbol{\hat \varrho}_{j^*}$.
	
	Then, we prove that $\theta_j \geq \theta_{j^*}$ if and only if $S(\boldsymbol{\hat \varrho}_j(t),\boldsymbol{\zeta}_j(t)) \geq S(\boldsymbol{\hat \varrho}_{j^*}(t),\boldsymbol{\zeta}_{j^*}(t))$. From~(\ref{eqn:for7c})-(\ref{eqn:for7e}), we have
	\begin{equation}
	\label{eqn:for7f}
	\begin{aligned}
	&\theta_j \big(S(\boldsymbol{\hat \varrho}_j(t),\boldsymbol{\zeta}_j(t)) - S(\boldsymbol{\hat \varrho}_{j^*}(t),\boldsymbol{\zeta}_{j^*}(t))\big) \geq \\ &\theta_{j^*} \big(S(\boldsymbol{\hat \varrho}_j(t),\boldsymbol{\zeta}_j(t)) - S(\boldsymbol{\hat \varrho}_{j^*}(t),\boldsymbol{\zeta}_{j^*}(t))\big).
	\end{aligned}
	\end{equation}
	By eliminating both $S(\boldsymbol{\hat \varrho}_j(t),\boldsymbol{\zeta}_j(t)) - S(\boldsymbol{\hat \varrho}_{j^*}(t),\boldsymbol{\zeta}_{j^*}(t))$, we obtain $\theta_j \geq \theta_{j^*}$ considering $S(\boldsymbol{\hat \varrho}_j(t),\boldsymbol{\zeta}_j(t)) - S(\boldsymbol{\hat \varrho}_{j^*}(t),\boldsymbol{\zeta}_{j^*}(t)) \geq 0$ since $S(\boldsymbol{\hat \varrho}_j(t),\boldsymbol{\zeta}_j(t)) \geq S(\boldsymbol{\hat \varrho}_{j^*}(t),\boldsymbol{\zeta}_{j^*}(t))$. As a result, we prove that if $S(\boldsymbol{\hat \varrho}_j(t),\boldsymbol{\zeta}_j(t)) \geq S(\boldsymbol{\hat \varrho}_{j^*}(t),\boldsymbol{\zeta}_{j^*}(t))$, then $\theta_j \geq \theta_{j^*}$. This concludes the proof.

	\section{Proof of Proposition 2}
	\label{appx:prop1}
	
Based on~(\ref{eqn:for7c}) and (\ref{eqn:for7d}), we can derive the following expressions:
\begin{equation}
\label{eqn:for1401a}
\begin{aligned}
&C(\boldsymbol{\hat \varrho}_j(t), \boldsymbol{\varphi}_j(t))- C(\boldsymbol{\hat \varrho}_{j^*}(t), \boldsymbol{\varphi}_{j^*}(t)) \leq \\
&\theta_j \Big(S(\boldsymbol{\hat \varrho}_j(t),\boldsymbol{\zeta}_j(t)) - S(\boldsymbol{\hat \varrho}_{j^*}(t),\boldsymbol{\zeta}_{j^*}(t))\Big),
\end{aligned}
\end{equation}
\begin{equation}
\label{eqn:for1401b}
\begin{aligned}
&C(\boldsymbol{\hat \varrho}_j(t), \boldsymbol{\varphi}_j(t))- C(\boldsymbol{\hat \varrho}_{j^*}(t), \boldsymbol{\varphi}_{j^*}(t)) \geq \\ &\theta_{j^*} \Big(S(\boldsymbol{\hat \varrho}_j(t),\boldsymbol{\zeta}_j(t)) - S(\boldsymbol{\hat \varrho}_{j^*}(t),\boldsymbol{\zeta}_{j^*}(t))\Big).
\end{aligned}
\end{equation}
Given $\boldsymbol{\hat \varrho}_{j} = \boldsymbol{\hat \varrho}_{j^*}$, if $S(\boldsymbol{\hat \varrho}_j(t),\boldsymbol{\zeta}_j(t)) \geq S(\boldsymbol{\hat \varrho}_{j^*}(t),\boldsymbol{\zeta}_{j^*}(t))$, then we can obtain $C(\boldsymbol{\hat \varrho}_j(t), \boldsymbol{\varphi}_j(t)) \geq C(\boldsymbol{\hat \varrho}_{j^*}(t), \boldsymbol{\varphi}_{j^*}(t))$ from~(\ref{eqn:for1401b}), and thus $\boldsymbol{\varphi}_j(t) \geq \boldsymbol{\varphi}_{j^*}(t)$ since the cost function~(\ref{eqn:for30a2}) is monotonically increasing in $\boldsymbol{\varphi}_j(t)$. Moreover, if $C(\boldsymbol{\hat \varrho}_j(t), \boldsymbol{\varphi}_j(t)) \geq C(\boldsymbol{\hat \varrho}_{j^*}(t), \boldsymbol{\varphi}_{j^*}(t))$, then we can obtain from~(\ref{eqn:for1401a}) that $S(\boldsymbol{\hat \varrho}_j(t),\boldsymbol{\zeta}_j(t)) \geq S(\boldsymbol{\hat \varrho}_{j^*}(t),\boldsymbol{\zeta}_{j^*}(t))$, and thus $\boldsymbol{\zeta}_j(t)) \geq \boldsymbol{\zeta}_{j^*}(t))$.

	\section{Proof of Proposition 3}
	\label{appx:prop2}
	
From Lemma~\ref{lemma1} and Proposition~\ref{prop1}, we have $\boldsymbol{\zeta}_j(t)) \geq \boldsymbol{\zeta}_{j^*}(t))$ and $\boldsymbol{\varphi}_j(t) \geq \boldsymbol{\varphi}_{j^*}(t)$. If $\theta_j \geq \theta_{j^*}$, then
\begin{align}
&\theta_jS(\boldsymbol{\hat \varrho}_j(t),\boldsymbol{\zeta}_j(t)) - C(\boldsymbol{\hat \varrho}_j(t), \boldsymbol{\varphi}_j(t)) \nonumber \geq \\ 
&\theta_jS(\boldsymbol{\hat \varrho}_{j^*}(t),\boldsymbol{\zeta}_{j^*}(t)) - C(\boldsymbol{\hat \varrho}_{j^*}(t), \boldsymbol{\varphi}_{j^*}(t))  \label{eqn:for1402} \geq \\
&\theta_{j^*}S(\boldsymbol{\hat \varrho}_{j^*}(t),\boldsymbol{\zeta}_{j^*}(t)) - C(\boldsymbol{\hat \varrho}_{j^*}(t), \boldsymbol{\varphi}_{j^*}(t)) \geq 0. \nonumber
\end{align}
Thus, $\theta_jS(\boldsymbol{\hat \varrho}_j(t),\boldsymbol{\zeta}_j(t)) - C(\boldsymbol{\hat \varrho}_j(t), \boldsymbol{\varphi}_j(t)) \geq  \theta_{j^*}S(\boldsymbol{\hat \varrho}_{j^*}(t),\boldsymbol{\zeta}_{j^*}(t)) - C(\boldsymbol{\hat \varrho}_{j^*}(t), \boldsymbol{\varphi}_{j^*}(t))$ when $\theta_j \geq \theta_{j^*}$.

	\section{Proof of Lemma 2}
	\label{appx:lemma2}
	
	We first define downward ICs (DICs) and upward ICs (UICs) as the IC contraints between the VSP with type indices $j$ and $j^*$, $\forall j^* \in \{1, \ldots, j-1\}$, and between the VSP with type indices $j$ and $j^*$, $\forall j^* \in \{j+1, \ldots, J\}$, respectively, in which
	\begin{align}
	&\theta_jS(\boldsymbol{\hat \varrho}_j(t),\boldsymbol{\zeta}_j(t)) - C(\boldsymbol{\hat \varrho}_j(t), \boldsymbol{\varphi}_j(t)) \geq \nonumber\\
	&\theta_jS(\boldsymbol{\hat \varrho}_{j^*}(t),\boldsymbol{\zeta}_{j^*}(t)) - C(\boldsymbol{\hat \varrho}_{j^*}(t), \boldsymbol{\varphi}_{j^*}(t)). \label{eqn:for7c1}
	\end{align}
	Then, we prove that DICs can be simplified into two consecutive types called local DICs (LDICs), which are the IC constraints between the VSP with type indices $j$ and $j-1$. Consider that $\theta_{j-1} < \theta_j < \theta_{j+1}, j \in \{2, \ldots, J-1\}$, we can derive
	\begin{align}
	&\theta_{j+1}S(\boldsymbol{\hat \varrho}_{j+1}(t),\boldsymbol{\zeta}_{j+1}(t)) - C(\boldsymbol{\hat \varrho}_{j+1}(t), \boldsymbol{\varphi}_{j+1}(t)) \geq  \nonumber\\
	&\theta_{j+1}S(\boldsymbol{\hat \varrho}_j(t),\boldsymbol{\zeta}_{j}(t)) - C(\boldsymbol{\hat \varrho}_j(t), \boldsymbol{\varphi}_{j}(t)), \label{eqn:for7c2}
	\end{align}
	\begin{align}
	&\theta_{j}S(\boldsymbol{\hat \varrho}_j(t),\boldsymbol{\zeta}_{j}(t)) - C(\boldsymbol{\hat \varrho}_j(t), \boldsymbol{\varphi}_{j}(t)) \geq \nonumber\\ 
	&\theta_{j}S(\boldsymbol{\hat \varrho}_{j-1}(t),\boldsymbol{\zeta}_{j-1}(t)) - C(\boldsymbol{\hat \varrho}_{j-1}(t), \boldsymbol{\varphi}_{j-1}(t)). \label{eqn:for7c3}
	\end{align}
	From Lemma~\ref{lemma1}, we have proven that $S(\boldsymbol{\hat \varrho}_j(t),\boldsymbol{\zeta}_j(t)) \geq S(\boldsymbol{\hat \varrho}_{j^*}(t),\boldsymbol{\zeta}_{j^*}(t))$ if $\theta_j \geq \theta_{j^*}$. Then, we can obtain
	\begin{align}
	&\theta_{j+1} (S(\boldsymbol{\hat \varrho}_j(t),\boldsymbol{\zeta}_j(t)) - S(\boldsymbol{\hat \varrho}_{j-1}(t),\boldsymbol{\zeta}_{j-1}(t))) \geq \nonumber\\ 
	&\theta_{j} (S(\boldsymbol{\hat \varrho}_j(t),\boldsymbol{\zeta}_j(t)) - S(\boldsymbol{\hat \varrho}_{j-1}(t),\boldsymbol{\zeta}_{j-1}(t))) \geq \nonumber\\ 
	&C(\boldsymbol{\hat \varrho}_j(t), \boldsymbol{\varphi}_j(t))- C(\boldsymbol{\hat \varrho}_{j-1}(t), \boldsymbol{\varphi}_{j-1}(t)), 	\label{eqn:for7c4}
	\end{align}
	and
	\begin{align}
	&\theta_{j+1}S(\boldsymbol{\hat \varrho}_{j+1}(t),\boldsymbol{\zeta}_{j+1}(t)) - C(\boldsymbol{\hat \varrho}_{j+1}(t), \boldsymbol{\varphi}_{j+1}(t)) \geq \nonumber\\
	&\theta_{j+1}S(\boldsymbol{\hat \varrho}_j(t),\boldsymbol{\zeta}_{j}(t)) - C(\boldsymbol{\hat \varrho}_j(t), \boldsymbol{\varphi}_{j}(t)) \geq \nonumber\\ &\theta_{j+1}S(\boldsymbol{\hat \varrho}_{j-1}(t),\boldsymbol{\zeta}_{j-1}(t)) - C(\boldsymbol{\hat \varrho}_{j-1}(t), \boldsymbol{\varphi}_{j-1}(t)). \label{eqn:for7c5}
	\end{align}
	As $\theta_{j+1}S(\boldsymbol{\hat \varrho}_{j+1}(t),\boldsymbol{\zeta}_{j+1}(t)) - C(\boldsymbol{\hat \varrho}_{j+1}(t), \boldsymbol{\varphi}_{j+1}(t)) \geq \theta_{j+1}S(\boldsymbol{\hat \varrho}_{j-1}(t),\boldsymbol{\zeta}_{j-1}(t)) - C(\boldsymbol{\hat \varrho}_{j-1}(t), \boldsymbol{\varphi}_{j-1}(t))$,
	it can be expanded from the VSP with type index $j-1$ to type index $1$ to show that all the DICs satisfy, i.e.,
	\begin{align}
	&\theta_{j+1}S(\boldsymbol{\hat \varrho}_{j+1}(t),\boldsymbol{\zeta}_{j+1}(t)) - C(\boldsymbol{\hat \varrho}_{j+1}(t), \boldsymbol{\varphi}_{j+1}(t)) \geq \label{eqn:for7c7}\\
	&\theta_{j+1}S(\boldsymbol{\hat \varrho}_{j-1}(t),\boldsymbol{\zeta}_{j-1}(t)) - C(\boldsymbol{\hat \varrho}_{j-1}(t), \boldsymbol{\varphi}_{j-1}(t)) \geq \ldots
	\geq \nonumber \\ &\theta_{j+1}S(\boldsymbol{\hat \varrho}_1(t),\boldsymbol{\zeta}_{1}(t)) - C(\boldsymbol{\hat \varrho}_1(t), \boldsymbol{\varphi}_{1}(t)), \forall j \in \{1,\ldots,J-1\} \nonumber.
	\end{align}
	As a result, we can conclude that by using the LDICs, all the DICs satisfy, and then the DICs can be reduced. Likewise, we can prove that all the UICs hold considering the local UICs (LUICs) and Lemma~\ref{lemma1}.

	\section{Proof of Theorem 1}
	\label{appx:theorem1}

	Consider $\mathbb{C}$ has $\Big(\boldsymbol{\zeta}^{-n}_{(\varsigma+1)}(t),\boldsymbol{\varphi}^{-n}_{(\varsigma+1)}(t)\Big)$ and $\Big(\boldsymbol{\zeta}^{-n}_{(\varsigma)}(t),\boldsymbol{\varphi}^{-n}_{\varsigma)}(t)\Big)$ such that
	\begin{equation}
	\label{eqn:for16a}
	\begin{aligned}
	\Big(\boldsymbol{\zeta}^{-n}_{(\varsigma+1)}(t),\boldsymbol{\varphi}^{-n}_{(\varsigma+1)}(t)\Big) > \Big(\boldsymbol{\zeta}^{-n}_{(\varsigma)}(t),\boldsymbol{\varphi}^{-n}_{\varsigma)}(t)\Big).
	\end{aligned}
	\end{equation}
	Given that ${\hat \Phi}_{(\varsigma+1)}^n\Big(\boldsymbol{\zeta}^{-n}_{(\varsigma)}(t),\boldsymbol{\varphi}^{-n}_{\varsigma)}(t)\Big) = \max \phantom{1} \Phi_{(\varsigma+1)}^n\Big(\boldsymbol{\zeta}^{-n}_{(\varsigma)}(t),\boldsymbol{\varphi}^{-n}_{\varsigma)}(t)\Big)$, we can show that the choice ${\hat \Phi}^n$ in $\Phi^n$ at SV-$n$ is increasing, which is
	\begin{equation}
	\label{eqn:for16b}
	\begin{aligned}
	{\hat \Phi}_{(\varsigma+2)}^n\Big(\boldsymbol{\zeta}_{(\varsigma+1)}^{-n}(t), \boldsymbol{\varphi}_{(\varsigma+1)}^{-n}(t)\Big) \geq {\hat \Phi}_{(\varsigma+1)}^n(t)\Big(\boldsymbol{\zeta}_{(\varsigma)}^{-n}(t), \boldsymbol{\varphi}_{(\varsigma)}^{-n}(t)\Big).
	\end{aligned}
	\end{equation}
	Using $\Big(\boldsymbol{\zeta}_{(\varsigma+1)}^{-n}(t), \boldsymbol{\varphi}_{(\varsigma+1)}^{-n}(t)\Big) > \Big(\boldsymbol{\zeta}_{(\varsigma)}^{-n}(t), \boldsymbol{\varphi}_{(\varsigma)}^{-n}(t)\Big)$ and ${\hat \Phi}_{(\varsigma+2)}^n\Big(\boldsymbol{\zeta}_{(\varsigma+1)}^{-n}(t), \boldsymbol{\varphi}_{(\varsigma+1)}^{-n}(t)\Big) < {\hat \Phi}_{(\varsigma+1)}^n\Big(\boldsymbol{\zeta}_{(\varsigma)}^{-n}(t), \boldsymbol{\varphi}_{(\varsigma)}^{-n}(t)\Big)$ to derive a contradiction, we can obtain
	\begin{align}
	&{\mu}_n\bigg({\hat \Phi}_{(\varsigma+1)}^n\Big(\boldsymbol{\zeta}_{(\varsigma)}^{-n}(t), \boldsymbol{\varphi}_{(\varsigma)}^{-n}(t)\Big),\boldsymbol{\zeta}_{(\varsigma+2)}^{-n}(t), \boldsymbol{\varphi}_{(\varsigma+2)}^{-n}(t)\bigg) + \nonumber\\ &{\mu}_n\bigg({\hat \Phi}_{(\varsigma+2)}^n\Big(\boldsymbol{\zeta}_{(\varsigma+1)}^{-n}(t), \boldsymbol{\varphi}_{(\varsigma+1)}^{-n}(t)\Big),\boldsymbol{\zeta}_{(\varsigma+1)}^{-n}(t), \boldsymbol{\varphi}_{(\varsigma+1)}^{-n}(t)\bigg) \geq \nonumber \\ &{\mu}_n\bigg({\hat \Phi}_{(\varsigma+2)}^n\Big(\boldsymbol{\zeta}_{(\varsigma+1)}^{-n}(t), \boldsymbol{\varphi}_{(\varsigma+1)}^{-n}(t)\Big),\boldsymbol{\zeta}_{(\varsigma+2)}^{-n}(t), \boldsymbol{\varphi}_{(\varsigma+2)}^{-n}(t)\bigg) +  \nonumber\\ &{\mu}_n\bigg({\hat \Phi}_{(\varsigma+1)}^n\Big(\boldsymbol{\zeta}_{(\varsigma)}^{-n}(t), \boldsymbol{\varphi}_{(\varsigma)}^{-n}(t)\Big),\boldsymbol{\zeta}_{(\varsigma+1)}^{-n}(t), \boldsymbol{\varphi}_{(\varsigma+1)}^{-n}(t)\bigg). \label{eqn:for16c}
	\end{align}
	Based on the definition of $\Phi_{(\varsigma+1)}^n$, ${\hat \Phi}_{(\varsigma+1)}^n\Big(\boldsymbol{\zeta}_{(\varsigma)}^{-n}(t), \boldsymbol{\varphi}_{(\varsigma)}^{-n}(t)\Big) \in \Phi_{(\varsigma+1)}^n\Big(\boldsymbol{\zeta}_{(\varsigma)}^{-n}(t), \boldsymbol{\varphi}_{(\varsigma)}^{-n}(t)\Big)$ describes that
	\begin{equation}
	\label{eqn:for16d}
	\begin{aligned}
	&{\mu}_n\bigg({\hat \Phi}_{(\varsigma+1)}^n\Big(\boldsymbol{\zeta}_{(\varsigma)}^{-n}(t), \boldsymbol{\varphi}_{(\varsigma)}^{-n}(t)\Big),\boldsymbol{\zeta}_{(\varsigma+1)}^{-n}(t), \boldsymbol{\varphi}_{(\varsigma+1)}^{-n}(t)\bigg) \geq \\ &{\mu}_n\bigg({\hat \Phi}_{(\varsigma+2)}^n\Big(\boldsymbol{\zeta}_{(\varsigma+1)}^{-n}(t), \boldsymbol{\varphi}_{(\varsigma+1)}^{-n}(t)\Big),\boldsymbol{\zeta}_{(\varsigma+1)}^{-n}(t), \boldsymbol{\varphi}_{(\varsigma+1)}^{-n}(t)\bigg).
	\end{aligned}
	\end{equation}
	Then, from~(\ref{eqn:for16c}) and~(\ref{eqn:for16d}), we have
	\begin{equation}
	\label{eqn:for16e}
	\begin{aligned}
	&{\mu}_n\bigg({\hat \Phi}_{(\varsigma+1)}^n\Big(\boldsymbol{\zeta}_{(\varsigma)}^{-n}(t), \boldsymbol{\varphi}_{(\varsigma)}^{-n}(t)\Big),\boldsymbol{\zeta}_{(\varsigma+2)}^{-n}(t), \boldsymbol{\varphi}_{(\varsigma+2)}^{-n}(t)\bigg) \geq \\ &{\mu}_n\bigg({\hat \Phi}_{(\varsigma+2)}^n\Big(\boldsymbol{\zeta}_{(\varsigma+1)}^{-n}(t), \boldsymbol{\varphi}_{(\varsigma+1)}^{-n}(t)\Big),\boldsymbol{\zeta}_{(\varsigma+2)}^{-n}(t), \boldsymbol{\varphi}_{(\varsigma+2)}^{-n}(t)\bigg),
	\end{aligned}
	\end{equation}
	and thus ${\hat \Phi}_{(\varsigma+1)}^n\Big(\boldsymbol{\zeta}_{(\varsigma)}^{-n}(t), \boldsymbol{\varphi}_{(\varsigma)}^{-n}(t)\Big) \in \Phi_{(\varsigma+2)}^n\Big(\boldsymbol{\zeta}_{(\varsigma+1)}^{-n}(t), \boldsymbol{\varphi}_{(\varsigma+1)}^{-n}(t)\Big)$. In other words, it can be noticed that ${\hat \Phi}_{(\varsigma+1)}^n\Big(\boldsymbol{\zeta}_{(\varsigma)}^{-n}(t),\boldsymbol{\varphi}_{(\varsigma)}^{-n}(t)\Big) > {\hat \Phi}_{(\varsigma+2)}^n\Big(\boldsymbol{\zeta}_{(\varsigma+1)}^{-n}(t), \boldsymbol{\varphi}_{(\varsigma+1)}^{-n}(t)\Big)$. Based on the definition of ${\hat \Phi}_{(\varsigma+1)}^n$, it can be observed that ${\hat \Phi}_{(\varsigma+2)}^n\Big(\boldsymbol{\zeta}_{(\varsigma+1)}^{-n}(t), \boldsymbol{\varphi}_{(\varsigma+1)}^{-n}(t)\Big) \geq {\hat \Phi}_{(\varsigma+1)}^n\Big(\boldsymbol{\zeta}_{(\varsigma)}^{-n}(t), \boldsymbol{\varphi}_{(\varsigma)}^{-n}(t)\Big)$ which contradicts with the above condition. As a result, ${\hat \Phi}^n$ is an increasing function. 
	
	Using the increasing function of ${\hat \Phi}^n$, the VSP can update the contract of SV-$n$ at iteration $\varsigma + 1$ 
	if the condition in~(\ref{eqn:for11c}) is satisfied. To this end, the use of $\gamma$ can be denoted as the optimality tolerance to stop the iterative algorithm. In particular, when all the SVs in $\mathcal{N}(t)$ do not hold the condition in~(\ref{eqn:for11c}) at iteration $\varsigma= \varkappa$, we have $\Big(\boldsymbol{\zeta}^n_{(\varkappa+1)}(t),\boldsymbol{\varphi}^n_{(\varkappa+1)}(t)\Big) = \Big(\boldsymbol{\zeta}^n_{(\varkappa)}(t),\boldsymbol{\varphi}^n_{(\varkappa)}(t)\Big), \forall n \in \mathcal{N}(t).$
	Hence, for the rest of $\varsigma$ values starting from $\varkappa$, the iterative algorithm provides the same ${\mu}_{n}\Big(\boldsymbol{\zeta}_{(\varsigma+1)}^n(t),\boldsymbol{\varphi}_{(\varsigma+1)}^n(t),\boldsymbol{\zeta}_{(\varsigma)}^{-n}(t), \boldsymbol{\varphi}_{(\varsigma)}^{-n}(t)\Big), \forall n \in \mathcal{N}(t)$. It means that the iterative algorithm converges under $\gamma$.
	
	\section{Proof of Theorem 2}
	\label{appx:theorem2}
	
	To prove that the iterative algorithm can converge to the equilibrium solution, we first show that an equilibrium exists by finding a fixed point of $\Phi$. Suppose that $\mathbb{C}^*$ contains $\big(\boldsymbol{\zeta}(t), \boldsymbol{\varphi}(t)\big) \in \mathbb{C}$. 
This $\mathbb{C}^*$ is a non-empty contract space because ${\hat \Phi}\big(\min \mathbb{C}\big) \geq \min \mathbb{C}$, ${\hat \Phi} \in \Phi$, and thus ${\hat \Phi}\big(\max \mathbb{C}^*\big) \geq \max \mathbb{C}^*$.
Since ${\hat \Phi}$ is an increasing function, we obtain ${\hat \Phi}\Big({\hat \Phi}\big(\max \mathbb{C}^*\big)\Big) \geq {\hat \Phi}\big(\max \mathbb{C}^*\big)$, and thus ${\hat \Phi}\big(\max \mathbb{C}^*\big) \in \mathbb{C}^*$. Then, we have ${\hat \Phi}\big(\max \mathbb{C}^*\big) \leq \max \mathbb{C}^*$ and obtain ${\hat \Phi}\big(\max \mathbb{C}^*\big) = \max \mathbb{C}^*$. Consequently, $\max \phantom{1} \mathbb{C}^*$ is a fixed point of $\Phi$ which contains $\Big(\boldsymbol{\hat \zeta}(t), \boldsymbol{\hat \varphi}(t)\Big)$ and indicates that the equilibrium exists.
Since the profits of all SV-$n$, $\forall n \in \mathcal{N}(t)$, follow the increasing function until they converge, no SVs can further improve the profits when their optimal contracts are obtained. In this case, it implies that the iterative algorithm must converge to the equilibrium contract solution $\Big(\boldsymbol{\hat \zeta}(t), \boldsymbol{\hat \varphi}(t)\Big)$, to guarantee that there exists no such SV-$n$ can improve its profit, i.e., $\mu_n\Big(\boldsymbol{\hat \zeta}^n(t), \boldsymbol{\hat \varphi}^n(t), \boldsymbol{\hat \zeta}^{-n}(t), \boldsymbol{\hat \varphi}^{-n}(t)\Big)$,  unilaterally~\cite{Bernheim:1986,Fraysse:1993}.

	\section{Proof of Theorem 3}
	\label{appx:theorem3}
	
	Let $n^\dagger \subset N$ denote the number of SVs in $\mathcal{N}(t)$ which complete the best response optimizations at iteration $\varsigma+1$, i.e., $\Phi_{(\varsigma+1)}^n\Big(\boldsymbol{\zeta}_{(\varsigma)}^{-n}(t), \boldsymbol{\varphi}_{(\varsigma)}^{-n}(t)\Big)$, without modifying their current contracts. We also define $U \in [0,1]$ as the current normalized candidate profit. Given that $R$ is the total number of contract policies, $n^\dagger$-th SV obtains $R-1$ new contract policies in $\mathbb{C}_{n^\dagger}$ (whose normalized candidate profits need to be compared with $U$) when the SV computes its $\Phi_{(\varsigma+1)}^{n^\dagger}\Big(\boldsymbol{\zeta}_{(\varsigma)}^{-n^\dagger}(t), \boldsymbol{\varphi}_{(\varsigma)}^{-n^\dagger}(t)\Big)$. In this case, the probability that all the $R-1$ policies cannot drive the next normalized candidate profit higher than $U$ is $U^{R-1}$. Thus, the $(n^\dagger+1)$-th SV can try to find its $\Phi_{(\varsigma+1)}^{n^\dagger+1}\Big(\boldsymbol{\zeta}_{(\varsigma)}^{-(n^\dagger+1)}(t), \boldsymbol{\varphi}_{(\varsigma)}^{-(n^\dagger+1)}(t)\Big)$. In contrast, the probability that one of the $R-1$ policies is the best response can be defined as $1-U^{R-1}$. As such, the normalized candidate profit changes from $U$ to a value larger than ${\hat \mu}_{n^\dagger}\Big(\boldsymbol{\zeta}_{(\varsigma)}^{n^\dagger}(t),\boldsymbol{\varphi}_{(\varsigma)}^{n^\dagger}(t),\boldsymbol{\zeta}_{(\varsigma)}^{-n^\dagger}(t), \boldsymbol{\varphi}_{(\varsigma)}^{-n^\dagger}(t)\Big) + \gamma$ with the probability $1-{\tilde \mu}_{n^\dagger}^{R-1}$, where ${\tilde \mu}_{n^\dagger} = {\hat \mu}_{n^\dagger}+\gamma$, and ${\hat \mu}_{n^\dagger} \in [0,1]$ is the normalized current profit of SV-$n^\dagger$ (we simplify the form of ${\hat \mu}_{n^\dagger}(.)$ into ${\hat \mu}_{n^\dagger}$). At this state, we set the number of SVs in $\mathcal {N}(t)$ that have finished the best response back to $1$. 

	Utilizing the Markov chain, let $\vartheta(U,n^\dagger)$ to be the number of steps of the Algorithm~\ref{ICEA}, when $n^\dagger$ SVs have finished the best response (without changing the contracts with normalized candidate profit $U$) before reaching the convergence. Then, we obtain
	\begin{equation}
	\label{eqn:for17a}
	\begin{aligned}
	\vartheta(U,n^\dagger) &= U^{R-1}\vartheta(U,n^\dagger+1)\\ &+\int_{U}^{1}(R-1){\tilde \mu}_{n^\dagger}^{R-2}(\vartheta({\tilde \mu}_{n^\dagger},1)+1)d{\tilde \mu}_{n^\dagger},
	\end{aligned}
	\end{equation}
	where the first component specifies that the $(n^\dagger+1)$-th SV does not modify its contract, i.e., $\vartheta(U,n^\dagger) = \vartheta(U,n^\dagger+1)$, with probability $U^{R-1}$. Additionally, the second component represents that the $(n^\dagger+1)$-th SV finds a new best response with normalized candidate profit ${\tilde \mu}_{n^\dagger}$ and probability density $(R-1){\tilde \mu}_{n^\dagger}^{R-2}$, and thus one additional step is produced, i.e., $\vartheta(U,n^\dagger) = \vartheta({\tilde \mu}_{n^\dagger},1)+1$. Given the boundary conditions $\vartheta(U,N) = 0$ and $\vartheta(1,n^\dagger) = 0$ for all $U$ and $n^\dagger$, respectively, and  $\Xi(U) = \overset{1}{\underset{U}{\int}}(R-1){\tilde \mu}_{n^\dagger}^{R-2}(\vartheta({\tilde \mu}_{n^\dagger},1)+1)d{\tilde \mu}_{n^\dagger}$, we have
	\begin{equation}
	\label{eqn:for17b}
	\begin{aligned}
	\left\{	\begin{array}{ll}
	\vartheta(U,1) &= U^{R-1}\vartheta(U,2) + \Xi(U),\\
	\vdots &= \vdots \\
	\vartheta(U,N-2) &= U^{R-1}\vartheta(U,N-1) + \Xi(U),\\
	\vartheta(U,N-1) &= \Xi(U).
	\end{array}	\right.\\
	\end{aligned}
	\end{equation}
	From (\ref{eqn:for17b}), we can obtain that $\vartheta(U,1) = \Xi(U)\Pi(U)$, where 
	\begin{equation}
	\label{eqn:for17b2}
	\begin{aligned}
	\Pi(U) =  1 + U^{R-1} + \ldots + U^{(R-1)(N-2)}.
	\end{aligned}
	\end{equation}
	As such, the differential equation of $\vartheta(U,1)$ with respect to $U$ is
	\begin{align}
	&\frac{d\vartheta(U,1)}{dU} + \Big((R-1)U^{R-2}\Pi(U) - \frac{1}{\Pi(U)}\frac{d\Pi(U)}{dU}\Big)\vartheta(U,1) = \nonumber\\ &-(R-1)U^{R-2}\Pi(U).
	\end{align}
	Given that the boundary condition $\vartheta(1,1) = 0$ and 
	\begin{equation}
	\label{eqn:for17e}
	\begin{aligned}
	\Omega(U) &= \int_{0}^{U}\Big((R-1){\tilde \mu}_{1}^{R-2}\Pi({\tilde \mu}_{1}) - \frac{1}{\Pi({\tilde \mu}_{1})}\frac{d\Pi({\tilde \mu}_{1})}{d{\tilde \mu}_{1}}\Big)d{\tilde \mu}_{1} \\
	&= \int_{0}^{U}(R-1){\tilde \mu}_{1}^{R-2}\Pi({\tilde \mu}_{1})d{\tilde \mu}_{1}  - \log(\Pi(U)),
	\end{aligned}
	\end{equation} 
	the function $\vartheta(U,1)$ becomes $\vartheta(U,1) = e^{-\Omega(U)}\int_{U}^{1}(R-1){\tilde \mu}_{1}^{R-2}\Pi({\tilde \mu}_{1})e^{\Omega({\tilde \mu}_{1})}d{\tilde \mu}_{1}.$
	As $\vartheta(U,1)$ is decreasing in $U$, the number of steps is upper-bounded by $\vartheta(0,1)$. Specifically, considering $\Omega(0) = 0$ from (\ref{eqn:for17e}) and $\Pi(0) = 1$ from (\ref{eqn:for17b2}), then
	\begin{equation}
	\label{eqn:for17f}
	\begin{aligned}
	\vartheta(0,1) &= \int_{0}^{1}(R-1){\tilde \mu}_{1}^{R-2}e^{\sum_{n=0}^{N-2}\frac{{\tilde \mu}_{1}^{(R-1)(n+1)}}{n+1}}d{\tilde \mu}_{1} \\
	&= \int_{0}^{1}e^{\sum_{n=0}^{N-2}\frac{{\tilde \mu}_{1}^{(R-1)(n+1)}}{n+1}}d{\tilde \mu}_{1}^{(R-1)} \\
	&= \int_{0}^{1 - \frac{1}{N}}e^{\sum_{n=1}^{N-1}\frac{{\tilde \mu}_{1}^{(R-1)n}}{n}}d{\tilde \mu}_{1}^{(R-1)} + \\&\int_{1 - \frac{1}{N}}^{1}e^{\sum_{n=1}^{N-1}\frac{{\tilde \mu}_{1}^{(R-1)n}}{n}}d{\tilde \mu}_{1}^{(R-1)} \\
	&\leq \int_{0}^{1 - \frac{1}{N}}e^{\sum_{n=1}^{\infty}\frac{{\tilde \mu}_{1}^{(R-1)n}}{n}}d{\tilde \mu}_{1}^{(R-1)} + \frac{1}{N}e^{\sum_{n=1}^{N-1}\frac{1}{n}} \\
	&= \int_{0}^{1 - \frac{1}{N}}\frac{d{\tilde \mu}_{1}^{(R-1)}}{1 - {\tilde \mu}_{1}^{(R-1)}} + e^z + O(1/N)\\
	&= \log(N) + e^z +O(1/N)),
	\end{aligned}
	\end{equation}
	where $z \approx 0.5772$ is the Euler constant. 
	From (\ref{eqn:for17f}), we can see that the number of steps of the Algorithm~\ref{ICEA} is bounded above to the polynomial complexity $\log(N) + e^z +O(1/N)$~\cite{Durand:2016}.

\section{Proof of Theorem 4}
\label{appx:theorem4}

We first consider that the local loss function for all learning SVs, i.e., $\epsilon_n, \forall n \in \mathcal{N}(t)$, are all $\delta_1$-smooth and $\delta_2$-strongly convex~\cite{Amiri:2020}, i.e.,
\begin{equation}
\label{eqn:for18a}
\begin{aligned}
\epsilon_n(\mathbf{W}(t)) - \epsilon_n(\mathbf{\hat W}(t)) &\leq \langle\mathbf{W}(t) - \mathbf{\hat W}(t), \nabla\epsilon_n(\mathbf{\hat W}(t))\rangle \\&+ 
\frac{\delta_1}{2}\|\mathbf{W}(t) - \mathbf{\hat W}(t)\|_2^2, \forall n \in \mathcal{N}(t),
\end{aligned}
\end{equation}
and
\begin{equation}
\label{eqn:for18b}
\begin{aligned}
\epsilon_n(\mathbf{W}(t)) - \epsilon_n(\mathbf{\hat W}(t)) &\geq \langle\mathbf{W}(t) - \mathbf{\hat W}(t), \nabla\epsilon_n(\mathbf{\hat W}(t))\rangle \\&+ 
\frac{\delta_2}{2}\|\mathbf{W}(t) - \mathbf{\hat W}(t)\|_2^2, \forall n \in \mathcal{N}(t),
\end{aligned}
\end{equation}
respectively. Then, given that $0 < \kappa(t) \leq \min(1,\frac{1}{\delta_2\tau_{\emph{\mbox{th}}}}), \forall t$, where $\delta_2$ is a positive constant and $\kappa(t) = \kappa_1^{(\tau_{\emph{\mbox{th}}})}(t) = \ldots = \kappa_N^{(\tau_{\emph{\mbox{th}}})}(t)$ from~(\ref{eqn3d}) with fixed $\beta_{p_n}^{(\tau_{\emph{\mbox{th}}})}(t)$ and $\beta_{q_n}^{(\tau_{\emph{\mbox{th}}})}(t)$, $\forall n \in \mathcal{N}(t)$, we can obtain the expected squared L2-norm global model gap~\cite{Amiri:2020} as follows:
\begin{equation}
\label{eqn18c}
\begin{aligned}
\mathbb{E}[\|\mathbf{W}(t) - \mathbf{W}^*\|_2^2] &\leq \bigg(\prod_{t^*=0}^{t-1}f_1(t^*)\bigg)\|\mathbf{W}(0) - \mathbf{W}^*\|_2^2 \\&+ \sum_{t^\dagger=0}^{t-1}f_2(t^\dagger)\prod_{t^*=t^\dagger+1}^{t-1}f_1(t^*),
\end{aligned}
\end{equation}
where
\begin{equation}
\label{eqn18d}
\begin{aligned}
f_1(t^*) \triangleq 1 - \delta_2\kappa(t^*)(\tau_{\emph{\mbox{th}}} - \kappa(t^*)(\tau_{\emph{\mbox{th}}}-1)),
\end{aligned}
\end{equation}
\begin{equation}
\label{eqn18e}
\begin{aligned}
&f_2(t^*) \triangleq \frac{(I-N)\kappa^2(t^*)\tau^2_{\emph{\mbox{th}}}\Upsilon^2}{N(I-1)} + 2\kappa(t^*)(\tau_{\emph{\mbox{th}}} - 1)\Lambda + \\
&(1+\delta_2(1-\kappa(t^*)))\kappa^2(t^*)\Upsilon^2\frac{\tau_{\emph{\mbox{th}}}(\tau_{\emph{\mbox{th}}}-1)(2\tau_{\emph{\mbox{th}}}-1)}{6} + \\
&\kappa^2(t^*)(\tau^2_{\emph{\mbox{th}}} + \tau_{\emph{\mbox{th}}} - 1)\Upsilon^2,
\end{aligned}
\end{equation}
\begin{equation}
\label{eqn18f}
\begin{aligned}
\Lambda \triangleq \Psi^*(\mathbf{W}^*) - \frac{1}{I}\sum_{i=1}^I\epsilon_i(\mathbf{W}^*) \geq 0,
\end{aligned}
\end{equation}
and $\Upsilon = \max \|\mathbf{W}(t)\|$ is the maximum global model for all learning rounds. In this case, the first component $\bigg(\prod_{t^*=0}^{t-1}f_1(t^*)\bigg)\|\mathbf{W}(0) - \mathbf{W}^*\|_2^2$ in~(\ref{eqn18c}) indicates the global model gap in the initial round. Meanwhile, the second component $\sum_{t^\dagger=0}^{t-1}f_2(t^\dagger)\prod_{t^*=t^\dagger+1}^{t-1}f_1(t^*)$ in~(\ref{eqn18c}) implies the cumulate impact of the proposed SV selection method on the global model convergence.

Using the $\delta$-smoothness function in (\ref{eqn:for18a}) for the global loss $\Psi(.)$, we have the below expression after $t^\diamond$ rounds, i.e., 
\begin{equation}
\label{eqn3l}
\begin{aligned}
&\mathbb{E}[\Psi(\mathbf{W}(t^\diamond))] - \Psi^*(\mathbf{W}^*) \leq \frac{\delta}{2}\mathbb{E}[\|\mathbf{W}(t^\diamond) - \mathbf{W}^*\|_2^2] \\
&\leq \frac{\delta}{2}\bigg(\prod_{t^*=0}^{t^\diamond-1}f_1(t^*)\bigg)\|\mathbf{W}(0) - \mathbf{W}^*\|_2^2 + \\ &\frac{\delta}{2}\sum_{t^\dagger=0}^{t^\diamond-1}f_2(t^\dagger)\prod_{t^*=t^\dagger+1}^{t^\diamond-1}f_1(t^*).
\end{aligned}
\end{equation}
Equation (\ref{eqn3l}) shows that the global loss gap is upper bounded by $\frac{\delta}{2}\mathbb{E}[\|\mathbf{W}(t^\diamond) - \mathbf{W}^*\|_2^2]$.
Then, due to the decreasing learning rate $\lim_{t\rightarrow\infty}\kappa(t) = 0$, we can observe that $\lim_{t^\diamond\rightarrow\infty}\big[\mathbb{E}[\Psi(\mathbf{W}(t^\diamond))] - \Psi^*(\mathbf{W}^*)\big] = 0$, which implies that the global loss gap eventually will reach zero. This concludes the proof.

\end{document}